\documentclass[a4paper,fleqn]{cas-sc}
\usepackage{amsmath}
\definecolor{cite_color}{rgb}{0.0, 0.58, 0.71}
\definecolor{db}{rgb}{0.0, 0.2, 0.7}
\usepackage{natbib}
\setcitestyle{authoryear,open={},close={}} 
\usepackage{hyperref}
\hypersetup{colorlinks=true,citecolor=cite_color,linkcolor=cite_color}
\usepackage{amsthm}

\newtheorem{thm}{Theorem}

\newtheorem{prop}[thm]{Proposition}

\newtheorem{rem}{Remark}

\def\tsc#1{\csdef{#1}{\textsc{\lowercase{#1}}\xspace}}
\tsc{WGM}
\tsc{QE}
\tsc{EP}
\tsc{PMS}
\tsc{BEC}
\tsc{DE}

\usepackage{graphicx}
\usepackage{subcaption}
\renewcommand{\figurename}{Fig.}

\makeatletter
\renewcommand*{\fnum@figure}[1]{\figurename~\thefigure.}
\makeatother

\begin{document}
\let\WriteBookmarks\relax
\let\printorcid\relax 
\setlength{\skip\footins}{24pt}
\def\floatpagepagefraction{1}
\def\textpagefraction{.001}
\shorttitle{}
\shortauthors{Sixu Li et~al.}

\title [mode = title]{Nonlinear Oscillatory Response of Automated Vehicle Car-following: Theoretical Analysis with Traffic State and Control Input Limits}

\author[1]{\textcolor{black}{Sixu Li}}

\credit{ Conceptualization, Methodology, Writing – original draft, Data curation, Software, Writing – review \& editing}


\author[1]{\textcolor{black}{Yang Zhou}}
\cormark[1]

\ead{yangzhou295@tamu.edu}

\credit{Conceptualization, Methodology, Writing – review \& editing, Supervision}

\address[1]{Zachry Department of Civil $\&$ Environmental Engineering, Texas A$\&$M University, College Station, TX 77843, USA}

\address{%
\vspace{-1em}%
\begingroup
\normalfont\rmfamily\upshape\small 
\begin{center}
\fbox{\parbox{0.99\linewidth}{%
\textbf{Published in Transportation Research Part B: Methodological, vol.~201.}\\
S.~Li and Y.~Zhou, ``Nonlinear oscillatory response of automated vehicle car-following:
Theoretical analysis with traffic state and control input limits,''
\textit{Transportation Research Part B}, 201 (2025), 103315.
\href{https://doi.org/10.1016/j.trb.2025.103315}{doi:10.1016/j.trb.2025.103315}.%
}}
\end{center}
\endgroup
\vspace{-0.1em}%
}

\cortext[cor1]{Corresponding author}

\begin{abstract}
This paper presents a framework grounded in the theory of describing function (DF) and incremental-input DF to theoretically analyze the nonlinear oscillatory response of automated vehicles (AVs) car-following (CF) amidst traffic oscillations, considering the limits of traffic state and control input. While prevailing approaches largely ignore these limits (i.e., saturation of acceleration/deceleration and speed) and focus on linear string stability analysis, this framework establishes a basis for theoretically analyzing the frequency response of AV systems with nonlinearities imposed by these limits. To this end, trajectories of CF pairs are decomposed into nominal and oscillatory trajectories, subsequently, the controlled AV system is repositioned within the oscillatory trajectory coordinates. Built on this base, DFs are employed to approximate the frequency responses of nonlinear saturation components by using their first harmonic output, thereby capturing the associated amplification ratio and phase shift. Considering the closed-loop nature of AV control systems, where system states and control input mutually influence each other, amplification ratios and phase shifts are balanced within the loop to ensure consistency. This balancing process may render multiple solutions, hence the incremental-input DF is further applied to identify the reasonable ones. The proposed method is validated by estimations from Simulink, and further comparisons with prevailing methods are conducted. Results confirm the alignment of our framework with Simulink results and exhibit its superior accuracy in analysis compared to the prevailing methods. Furthermore, the framework proves valuable in string stability analysis, especially when conventional linear methods offer misleading insights. 
\end{abstract}

\begin{keywords}
Automated vehicles \sep Car-following \sep Traffic oscillation \sep Saturation nonlinearity \sep Incremental-input describing function \sep Frequency domain analysis
\end{keywords}

\maketitle

\section{Introduction}

Traffic oscillation, alternatively termed stop-and-go traffic, is a recurrent pattern of deceleration and acceleration typically observed in congested traffic scenarios. This phenomenon significantly impairs traffic efficiency, compromises safety, and diminishes energy efficiency (\citep{bilbao2008costs,tian2024physically}). Extensive research has substantiated that a multitude of maneuvers can instigate traffic oscillations, including lane changing (\citep{ahn2007freeway,mauch2002freeway,laval2006lane}), merging, and diverging (\citep{ahn2010effects,cassidy2005increasing}). The rapid progression of vehicle automation technologies, particularly the advent of automated vehicles (AVs) (\citep{guanetti2018control,rajamani2011vehicle}), introduces a new traffic flow system: mixed traffic comprising a heterogeneous mix of AVs and human-driven vehicles (HDVs). This mixed traffic system is expected to reshape traditional traffic flow dynamics in unparalleled ways and improve traffic throughput, safety as well as energy efficiency during traffic oscillations (\cite{talebpour2016influence,li2024beyond,xu2022hierarchical}). Consequently, acquiring a profound understanding of the distinctive behaviors exhibited by AVs during traffic oscillations becomes imperative in order to model and analyze this new traffic flow system.

In recent decades, significant efforts have been made to analyze and understand the disturbances evolution of AV car-following (CF) in the context of oscillations. A predominant category of approaches involves conducting string stability analysis, which examines if traffic oscillations are dampened or remain bounded along a string of vehicles (\citep{feng2019string}). Various definitions of string stability have been proposed to this end. For instance, the original string stability (OSS) was defined in \citep{chu1974decentralized} by the boundedness and convergence of oscillations for all vehicles. The strong frequency domain string stability (SFSS) introduced in \citep{naus2010string} was defined as the $H_\infty$ norm of the transfer function of all vehicles being less than or equal to 1. In \citep{swaroop1996string}, the asymptotically time domain string stability (ATSS) was proposed, which intuitively means that the oscillations of all vehicles remain small if the oscillation of the leading vehicle starts small, and all oscillations converge asymptotically. The $L_p$ string stability (LPSS) was established in \citep{ploeg2013lp} as the $L_p$ norm of the output of all vehicles being bounded by $\kappa$ functions of the $L_p$ norms of control input and initial states. In string stability analysis, especially under the definitions of SFSS and LPSS, the methodology of mathematically deriving and analyzing the transfer function/frequency response of AV systems is prevalent, particularly due to the frequency response's ability to capture the amplification ratio and phase shift of oscillations (\citep{franklin2002feedback}). Using this approach, Naus et al. designed a cooperative adaptive cruise control system and established conditions on the transfer function for SFSS (\cite{naus2010string}). Xiao et al. examined the effect of parasitic delays and lags on the SFSS of vehicular strings equipped with adaptive cruise control (\citep{xiao2008stability}). Ploeg et al. established the conditions on the transfer function for LPSS and validated their findings through field experiments (\citep{ploeg2013lp}). More recently, the benefits of connectivity on SFSS with a multiple-predecessor following typology were investigated (\citep{darbha2018benefits,bian2019reducing}). In \citep{montanino2021string}, the stability of heterogeneous traffic flow was studied through the weak version of SFSS, and with linearized and parameter-uncertain CF models.

The aforementioned and similar work relied on linear models, which provide an elegant way of designing string-stable controllers in scenarios where the saturation limits on traffic state and control input (i.e., velocity and acceleration) are not reached. However, from an analysis perspective, such methods might produce unrealistic results. Consider a simple extreme example: suppose a vehicular string is string-unstable such that, at a specific frequency, the oscillation amplification ratio of each vehicle is greater than one, then as the number of vehicles increases, the amplitude of oscillations would approach infinity at the end of the vehicular string, making the acceleration/deceleration exceed the reasonable limit and vehicle may travel backward or surpass free flow speed. This highlights the necessity of analyzing the oscillatory response of AVs considering the presence of control and traffic state saturation limits, which are inherently nonlinear. Such limits are inherent in vehicular systems due to the physical limits on acceleration, deceleration, and velocity (\citep{stone2004automotive}). Other nontrivial sources of saturation arise from the clipping of control signals generated by an unconstrained controller and the imposition of control and/or state constraints by constrained optimal controllers. These saturations are typically imposed to accommodate the aforementioned physical limits or to curb overly aggressive behaviors of AVs (\citep{li2024sequencing,wang2014rolling}). It is crucial to note that saturation limits can be reached not only in string-unstable platoons, where oscillations are amplified along the vehicular string, but also in string-stable, heterogeneous platoons (\citep{naus2010string}). particularly when an AV (e.g., automated truck) has more restrictive limits compared to its predecessor. {For example, take an automated truck with a maximum absolute acceleration of $1\, \text{m/s}^2$ that follows a passenger car, which freely oscillates with a maximum absolute acceleration of $3\, \text{m/s}^2$. A string stable controller could generate an acceleration command ranging from $-2$ to $2\, \text{m/s}^2$, yet this still surpasses the acceleration limits of the automated truck.}

The existing body of literature is deficient in methodologies to derive the frequency response for the examination of disturbance amplification and phase shift of AVs during traffic oscillation in the presence of saturation nonlinearities. In \citep{zhou2022empirical}, acceleration saturation was considered, but its impact on string stability was studied numerically. In \citep{bingol2022string}, the focus was on the saturation of engine and braking forces. However, the main objective was to generate trajectories that avoid reaching saturation limits, rather than analyzing the system's response when these limits are reached. Recently, a notable study (\citep{zhou2023data}) offered a data-driven framework to quantitatively scrutinize the disturbance amplification ratio of automated vehicles in CF scenarios, however, it does not furnish theoretical insights into the behaviors observed. The describing functions (DF) method, perceived as an extension of the frequency response, is a nonlinear control technique that replaces the nonlinearity by its DF (i.e., a frequency-dependent equivalent, and possibly complex, gain) for analysis purposes by approximating the periodic output of nonlinear elements with the first harmonic alone (\citep{franklin2002feedback}). This method facilitates the analytical examination of systems with nonlinearity in the frequency domain. Li et al. pioneered the application of this technique for the analytical study of HDV behaviors during traffic oscillations  (\citep{li2011characterization,li2012prediction,li2014stop}). However, their framework is confined to a specific class of conventional car-following models (typically first-order differential with a single nonlinear element) and is not adaptable to automated vehicle systems, which are typically second-order or third-order systems exhibiting multiple saturation nonlinearities (\citep{swaroop1999constant,li2024sequencing}). Furthermore, their approach does not systematically address the stability of the forced oscillations solved by their framework. Recently, the DF method has been employed to derive theoretical results for assessing the outcomes of a data-driven framework analyzing the oscillation amplification behavior of AVs (\citep{zhou2023data}). However, this approach simplifies the AV system by incorporating the nonlinear element outside of the closed-loop, rather than integrating it within the loop. Moreover, the methodology presented is exclusively applicable to analyze the effect of a single saturation element.

{This paper aims to address the identified gaps by proposing a framework that, (1) for the first time, analytically discerns the oscillatory CF characteristics of \textit{AV closed-loop} systems laden with \textit{multiple} saturation nonlinearities, via a frequency response approach; and (2) it for the first time, examines the "realism" of solutions using oscillation stability analysis. Specifically, the trajectory decomposition concept proposed by \citep{li2011characterization} is utilized and the general CF problem of AVs considering both control and state saturation is modeled. The nonlinear saturation elements are replaced by their DF and a first-harmonic balancing is conducted to affirm the absence of limit cycles (self-sustained oscillation without external periodic excitement) and to solve for forced oscillation candidates of the AV. The incremental-input DF, representing the DF of a disturbance surrounding a specific oscillation candidate, is employed to scrutinize the stability of each candidate. This approach facilitates the identification of the stable oscillation and the acquisition of the frequency response. It is worth noting that the DF and incremental-input DF analysis serve as complementary steps in our framework. }The theoretical outcomes derived from this framework are compared with the linear theoretical method and the method proposed in \cite{zhou2023data}, validated by estimations obtained from Simulink, demonstrating the framework’s proficiency in accurately encapsulating the oscillation characteristics inherent to AVs.

The remainder of this paper is structured as follows. Section \ref{sec2} models the CF problem of AVs with saturation nonlinearities. Section \ref{sec3} presents the DF of control saturation and state saturation elements. The frequency response analysis based on the DF method is elucidated in Section \ref{sec4}. Section \ref{sec5} introduces the incremental-input DF and conducts stability analysis on forced oscillation candidates. Section \ref{sec6} compares the theoretical findings with other methods, using Simulink experiment results for validation, and provides a discussion on the observed outcomes. Finally, conclusions are drawn in Section \ref{sec7}.

\section{Modeling the CF problem of AVs\label{sec2}}

In this section, the modeling of the CF problem for closed-loop nonlinear AV systems is elucidated. Specifically, the methodology proposed by \citep{li2011characterization} is employed to dissect a vehicle trajectory into two components: a nominal trajectory and an oscillatory trajectory. This approach facilitates the modeling of closed-loop controlled systems of AVs, incorporating inherent nonlinearities.

Vehicle trajectories encapsulate both macroscopic and microscopic behaviors. The macroscopic behavior delineates equilibrium states, which are inherently tied to the fundamental diagram, while the microscopic behavior delineates the deviations of individual vehicles from these equilibrium states. We decompose a vehicle trajectory into a nominal trajectory, indicative of macroscopic behavior, and an oscillatory trajectory, indicative of microscopic behavior. As illustrated in Fig. \ref{fig:decomposition}, for any AV and its predecessor in the traffic, indexed by $n+1$ and $n$, respectively, their trajectories $p_{n+1}$ and $p_n$ are segregated into nominal and oscillatory trajectories. The nominal trajectories $\bar{p}_{n+1}$ and $\bar{p}_n$ encapsulate macroscopic behaviors, reflecting equilibrium spacing, velocity, etc., and the oscillatory trajectories $\tilde{p}_{n+1}$ and $\tilde{p}_n$ encapsulate the microscopic behaviors exhibited by the vehicles amidst traffic oscillations.

\begin{figure}[h]
    \centering
    \setlength{\abovecaptionskip}{0pt}
    \includegraphics[width=0.8\textwidth]{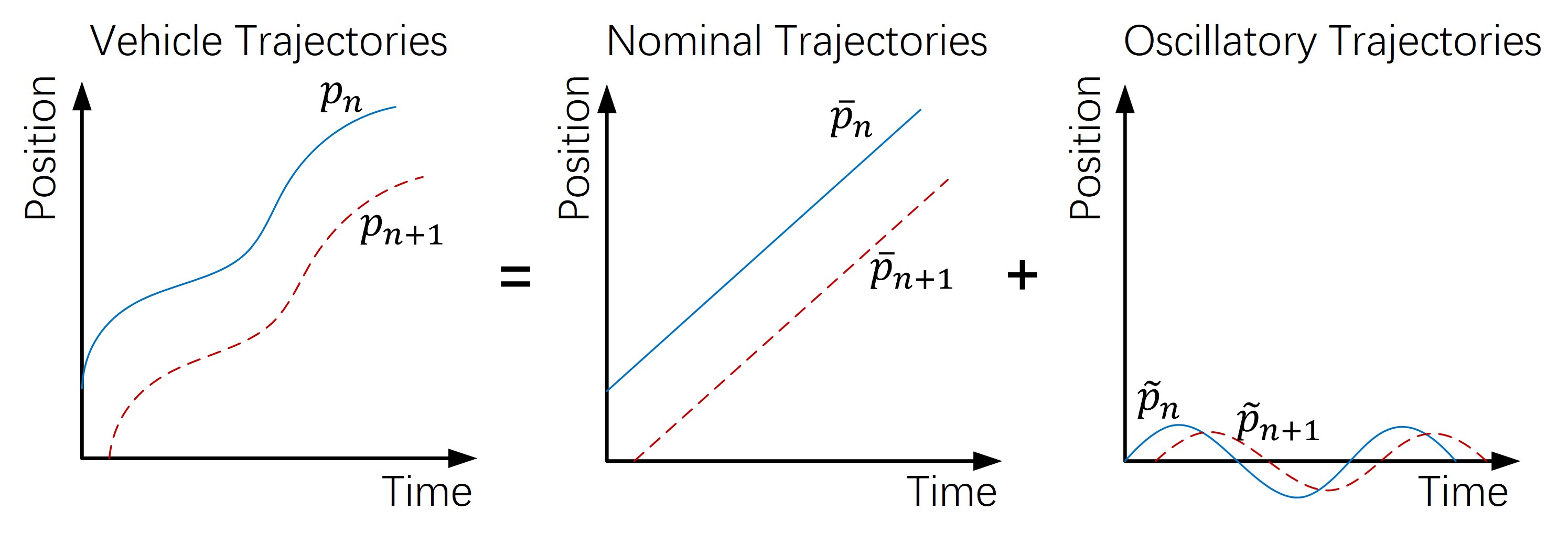}
    \caption{Vehicle trajectory decomposition}
    \label{fig:decomposition}
\end{figure}

{\begin{rem}
    The decomposition method was validated with field data in \citep{li2012prediction} and demonstrated proficiency in accurately representing fully developed oscillations. It has since been widely applied in traffic oscillation studies (e.g., \citep{rhoades2014heterogeneous,li2014stop,rhoades2016calibration,zhou2023data}). 
\end{rem}}

Considering a linear feedback CF control law of AVs:
\begin{equation}
{a_{n+1}(t)=k_d\Delta d_{n+1}(t)+k_v\Delta v_{n+1}(t)} \label{eq:generic control}
\end{equation}

\noindent where $a_{n+1}(t)$ is the acceleration of the AV $n+1$ at time $t$, $\Delta d_{n+1}(t)=p_n(t)-p_{n+1}(t)-d^*_{n+1}(t)$ is the spacing deviation of AV $n+1$ at time $t$, and $d^*_{n+1}(t)=l_{n+1}+\tau \dot{p}_{n+1}(t)$ is the desired spacing of AV $n+1$ at time $t$ wherein $l_{n+1}$ is a constant distance, $\tau$ is a constant time, and $\dot{p}_{n+1}(t)$ is the speed of AV $n+1$. $\Delta v_{n+1}(t)=\dot{p}_{n}(t)-\dot{p}_{n+1}(t)$ is the speed difference between AV $n+1$ and its predecessor at time $t$. $k_d$ and $k_v$ are the feedback gains on $\Delta d_{n+1}$ and $\Delta v_{n+1}$, respectively. It is imperative to note that $d^*_{n+1}(t)=l_{n+1}+\tau \dot{p}_{n+1}(t)$ is a generic form. For example, under the two most prevalent spacing policies—constant time gap policy (\citep{naus2010string,kianfar2015control}) and constant distance policy (\citep{swaroop1999constant,naus2010string,li2024sequencing})—$\tau$ equals to the pre-defined time gap and zero, respectively. Based on the trajectory decomposition shown in Fig. \ref{fig:decomposition}, we obtain:

\begin{equation}
\begin{aligned}
\label{eq:control decompose1}   
\Delta d_{n+1}(t)&=(\bar{p}_n(t)-\bar{p}_{n+1}(t))+(\tilde{p}_n(t)-\tilde{p}_{n+1}(t))-(l_{n+1}+\tau (v_e+\dot{\tilde{p}}_{n+1}(t)))\\  
&=\tilde{p}_n(t)-\tilde{p}_{n+1}(t)-\tau \dot{\tilde{p}}_{n+1}(t))
\end{aligned}
\end{equation}

\begin{equation}
\begin{aligned}
\label{eq:control decompose2}   
\Delta v_{n+1}(t)&=(\dot{\bar{p}}_n(t)-\dot{\bar{p}}_{n+1}(t))+(\dot{\tilde{p}}_n(t)-\dot{\tilde{p}}_{n+1}(t))\\  
&=\dot{\tilde{p}}_n(t)-\dot{\tilde{p}}_{n+1}(t)
\end{aligned}
\end{equation}

\noindent where $v_e$ is the equilibrium speed. Therefore, given $v_e$, based on Eqs. (\ref{eq:control decompose1}) and (\ref{eq:control decompose2}), an equivalent form of Eq. (\ref{eq:generic control}) can be obtained as:



\begin{equation}
a_{n+1}(t)=k_{1}(\tilde{p}_n(t)-\tilde{p}_{n+1}(t))+k_{2}\dot{\tilde{p}}_n(t)+k_{3}\dot{\tilde{p}}_{n+1}(t) \label{eq:new linear generic control}
\end{equation}

\noindent with $k_1=k_d$, $k_2=k_v$, and $k_3=-k_v-k_d\tau$. 

While Eq. (\ref{eq:new linear generic control}) illustrates that the control inputs (i.e., acceleration/deceleration) are solely dependent on the states of the oscillatory trajectory, it is also evident that these control inputs exclusively influence the oscillatory trajectory, given that the nominal trajectory progresses at a constant velocity, $v_e$. Consequently, the entire controlled AV system model with saturation can be transposed into the coordinate system of the oscillatory trajectory as follows:

\begin{equation}
\dot{\tilde{p}}_{n+1}(t)=\begin{cases} 
\tilde{v}_{min} & \text{if } \tilde{v}_{n+1}(t) < \tilde{v}_{min} \\
\tilde{v}_{n+1}(t) & \text{if } \tilde{v}_{min} \leq \tilde{v}_{n+1}(t) \leq \tilde{v}_{max} \\
\tilde{v}_{max} & \text{if } \tilde{v}_{n+1}(t) > \tilde{v}_{max} 
\end{cases}   \label{eq:AV model1}
\end{equation}

\begin{equation}
\dot{\tilde{v}}_{n+1}(t)=\begin{cases} 
{a}_{min} & \text{if } {a}_{n+1}(t) < {a}_{min} \\
{a}_{n+1}(t) & \text{if } {a}_{min} \leq {a}_{n+1}(t) \leq {a}_{max} \\
{a}_{max} & \text{if } {a}_{n+1}(t) > {a}_{max} 
\end{cases}   \label{eq:AV model2}
\end{equation}

\begin{equation}
a_{n+1}(t)=k_{1}(\tilde{p}_n(t)-\tilde{p}_{n+1}(t))+k_{2}\dot{\tilde{p}}_n(t)+k_{3}\dot{\tilde{p}}_{n+1}(t) \label{eq:AV model3}
\end{equation}

\noindent where $\tilde{v}_{min}$ and $\tilde{v}_{max}$ are the lower and upper limits of the oscillatory speed of AV $n+1$, respectively. These can be readily derived from the speed limits given the equilibrium speed. ${a}_{min}$ and ${a}_{max}$ represent the lower and upper limits of the acceleration of AV $n+1$, respectively. The controlled AV system Eqs. (\ref{eq:AV model1})-(\ref{eq:AV model3}) is graphically depicted in Fig. \ref{fig:block diagram}, where saturation 1 imposes acceleration limits and saturation 2 imposes oscillatory speed limits.  

\begin{rem}
\label{rem1}
    Please note that, within this paper, we denote unsaturated acceleration and oscillatory velocity with $a$ and $\tilde{v}$, respectively, while $\dot{\tilde{v}}$ and $\dot{\tilde{p}}$ represent the saturated acceleration and oscillatory velocity, respectively. Furthermore, given that $k_d>0$, $k_v>0$, and $\tau>0$ in controller design, it follows that $k_1>0$, $k_2>0$, and $k_3<0$.
\end{rem}

\begin{figure}[h]
    \centering
    \setlength{\abovecaptionskip}{0pt}
    \includegraphics[width=0.8\textwidth]{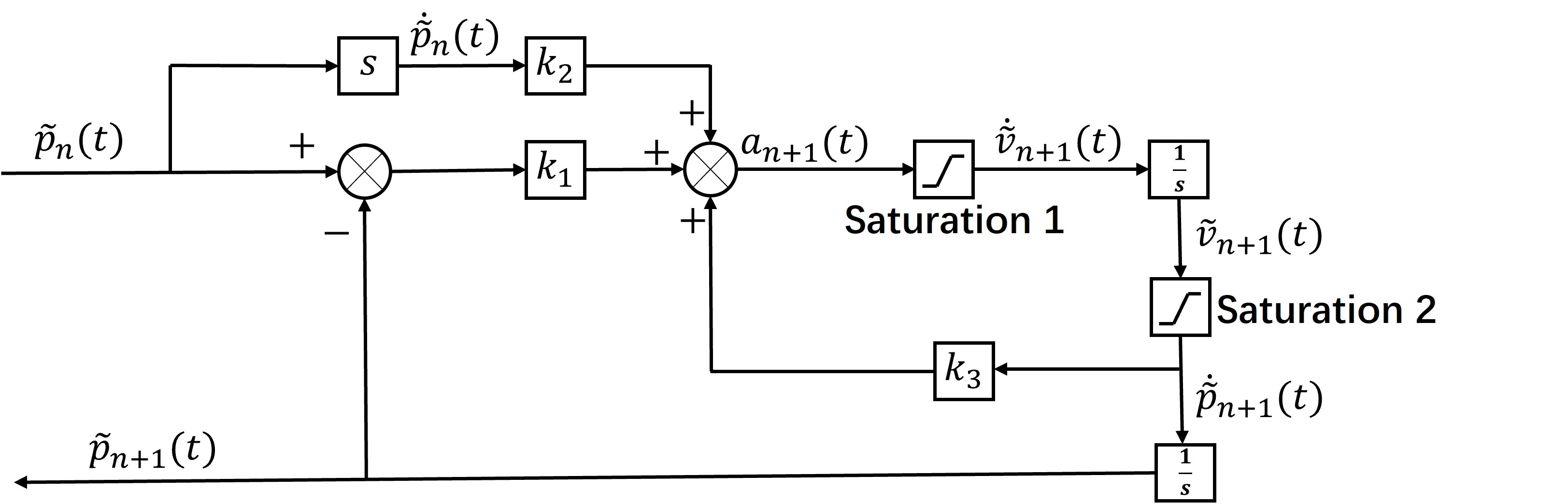}
    \caption{Controlled AV system block diagram}
    \label{fig:block diagram}
\end{figure}

{\begin{rem}
    While Eq. (\ref{eq:AV model3}) has also been applied to describe HDVs in the literature, they typically rely on strong assumptions about deterministic behavior and further linearization simplifications. To maintain a clear and rigorous scope, this study focuses on AVs.
\end{rem}}

\begin{rem}
In practice, time delays could impact the response of AV systems. As this paper does not address connectivity issues, focusing on actuation delays is sufficient. We can introduce an actuation delay into our framework via the block diagram shown in Fig. \ref{fig:block diagram delay}, where the added delay block signifies an actuation delay of $\tau$ seconds. A time delay is known to be linear time-invariant, with a frequency response magnitude of 1 and a phase shift of $-\omega\tau$ (\citep{franklin2002feedback}). Actuation delay can be effectively compensated in lower-level control design (\citep{xing2018smith,xu2020preview}), making it negligible for response analysis. Therefore, to simplify the presentation of our framework, we do not incorporate the delay component in our analysis that follows. It is worth noting that, the steps of our analysis procedure remain consistent even when incorporating the delay, as it merely represents an additional linear component without introducing any further nonlinearity.
\end{rem}

\begin{figure}[h]
    \centering
    \setlength{\abovecaptionskip}{0pt}
    \includegraphics[width=0.8\textwidth]{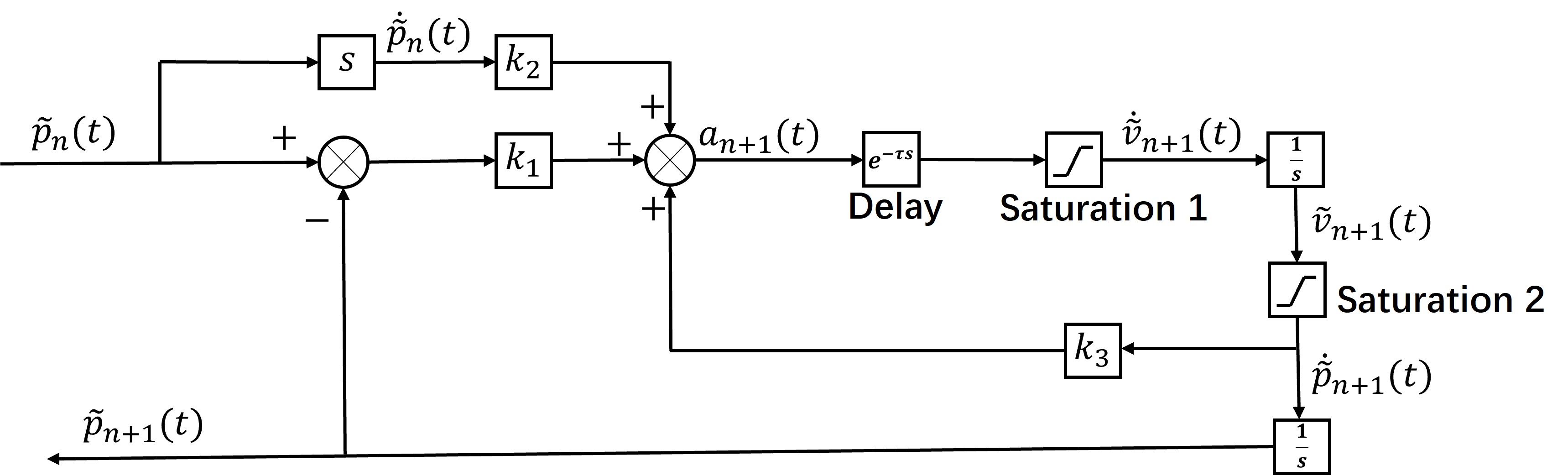}
    \caption{{Controlled AV system block diagram with actuation delay}}
    \label{fig:block diagram delay}
\end{figure}

{For clarity, before diving into the detailed derivations, we provide a brief overview of the main implementation steps, along with the equation numbers of key results, as shown in Fig. \ref{fig:steps}. In this figure, Cases 1 to 3 correspond to systems with only control saturation, systems with only state saturation, and systems with both control and state saturation, respectively. Step 1 involves calculating the DF of nonlinear elements (i.e., acceleration and/or velocity saturation elements) to represent them in the frequency domain. Step 2 requires substituting the DFs into the closed-loop frequency response analysis to solve for potential solutions of the AV response. In Step 3, the IDF of perturbations concerning the potential solutions identified in Step 2 is calculated. Step 4 involves substituting the IDFs into a closed-loop analysis to determine whether the perturbations amplify or attenuate, which in turn indicates the stability of the potential solutions identified in Step 2.}
 
\begin{figure}[h]
    \centering
    \setlength{\abovecaptionskip}{0pt}
    \includegraphics[width=0.99\textwidth]{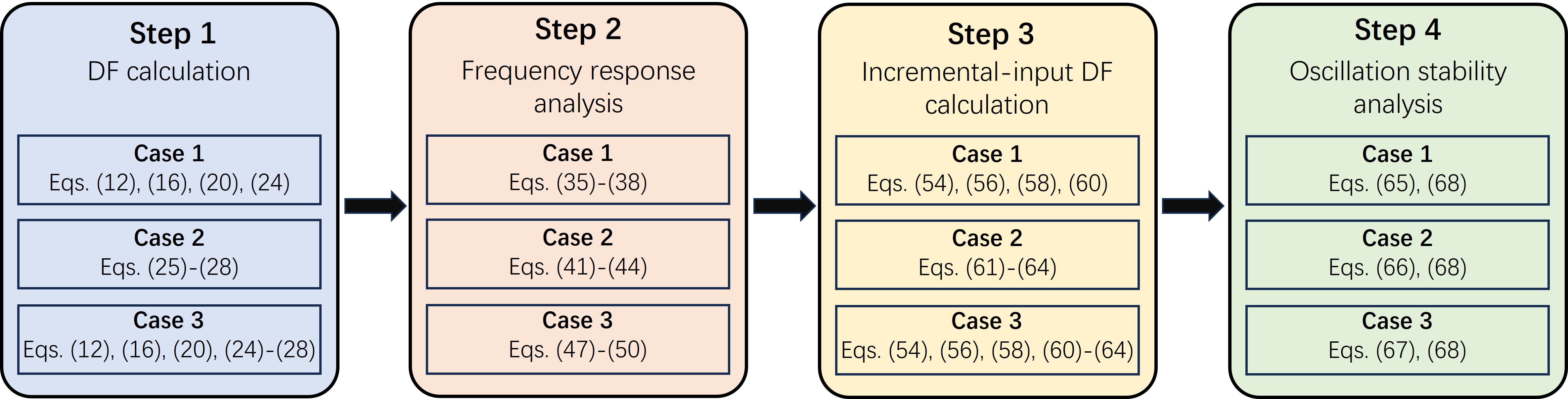}
    \caption{{Implementation steps and key equations overview}}
    \label{fig:steps}
\end{figure}

\section{DF analysis\label{sec3}}

 With the CF problem of AVs modeled, we now introduce the DF method, which serves as the foundation of our analysis. The DF method represents a nonlinear element subjected to a sinusoidal input by a quasi-linear optimal gain corresponding to each magnitude of the input to the nonlinearity (\citep{sastry2013nonlinear}). Given that microscopic behaviors during well-developed traffic oscillations are periodic and exhibit a predominant frequency (\citep{li2010measurement}), the oscillatory microscopic behavior can be approximated {by employing a sinusoidal trajectory characterized by this predominant frequency} (\citep{li2011characterization}). This makes the DF method particularly apt for analyzing traffic oscillations. In this section, Subsection \ref{sec3.1} provides an introduction to the DF method, while Subsection \ref{sec3.2} delineates the DF of control saturation and state saturation (corresponding to Saturation 1 and Saturation 2 in Fig. \ref{fig:block diagram}, respectively). 

 \subsection{The DF method \label{sec3.1}}

The fundamental idea of the DF method is to approximate the output of a nonlinear element by its first harmonics (i,e.,  the components in the Fourier series of the output that resonate at the input frequency). This approximation is grounded in two assumptions:  (1) The linear parts of the system exhibit a low-pass characteristic, filtering out terms with frequencies higher than the input. The integration elements within the AV system fulfill this requirement; (2) The output of the nonlinear element does not contain sub-harmonics (terms with a frequency lower than the input frequency), which is a practical assumption for "weak" nonlinearities like saturation (\citep{atherton2011introduction,vander1968multiple}).
 
{\begin{rem}Given these assumptions, the input-output properties of AV systems do not alter the frequency of traffic oscillations, that is, an AV will response at the same frequency as the oscillation it is subjected to.\end{rem}}

Consider the Fourier series of the output $y(t)$ of a nonlinear element subjected to a sinusoidal input $Bsin(\omega t)$:
\begin{equation}
y(t)=\sum_{k=0}^\infty Y_{k,1}\sin(k\omega t)+Y_{k,2}\cos(k\omega t) \label{eq:Fourier series}
\end{equation}
where the coefficients of the first harmonics, $Y_{1,1}$ and $Y_{1,2}$, are calculated by:
\begin{equation}
Y_{1,1}=\frac{1}{\pi}\int_0^{2\pi}y(t)\sin(\omega t)d(\omega t) \label{eq:Y11}
\end{equation}

\begin{equation}
Y_{1,2}=\frac{1}{\pi}\int_0^{2\pi}y(t)\cos(\omega t)d(\omega t) \label{eq:Y12}
\end{equation}

Based on the first harmonics, the DF serves as a complex optimal gain to characterize the impact of the nonlinear element on a sinusoidal input traversing through it, and is obtained as (\citep{atherton2011introduction,li2011characterization,zhou2023data}):
\begin{equation}
N(B)=\frac{Y_{1,1}+jY_{1,2}}{B} \label{eq:DF general}
\end{equation}

The amplification and phase shift imposed on the sinusoidal input can be computed using $|N(B)|=\frac{\sqrt{Y_{1,1}^2+Y_{1,2}^2}}{B}$ and $\angle N(B)=\tan^{-1}(\frac{Y_{1,2}}{Y_{1,1}})$, respectively (\citep{franklin2002feedback}). 

{\begin{rem}
    The DF representation has been proved to be an optimal representation in a form analogous to frequency response (\citep{vander1968multiple}). This approach has been validated in various fields (\citep{maraini2018nonlinear, shtessel2014sliding,davidson2013application}), including traffic oscillation analysis of HDVs (\citep{li2011characterization,rhoades2014heterogeneous}).
\end{rem}}

 \subsection{DFs of control and state saturation \label{sec3.2}}

In this subsection, we conduct an analysis of the DF for both the control saturation and state saturation elements within a controlled AV system.

\subsubsection{Control saturation \label{sec3.2.1}}
For the control saturation element (Saturation 1 in Fig. \ref{fig:block diagram}), let the input of it be denoted as $a_{n+1}(t)=B_{a}sin(\omega t)$, considering the acceleration limits as $[a_{min}, a_{max}]$. The DF of the control saturation, $N_{a,sat}(B_{a})$, can be characterized into four cases corresponding to different limit activeness, depending on the value of $B_{a}$, given the known values of $a_{min}$ and $a_{max}$:

\textbf{Case (1)}: limit inactive (i.e., $B_a\leq-a_{min}$ and $B_a\leq a_{max}$), as illustrated in Fig. \ref{fig:case_1} (a). 

\begin{figure}[h]
    \centering
    \setlength{\abovecaptionskip}{0pt}
    \subcaptionbox{}
    {\includegraphics[width=0.29\textwidth]{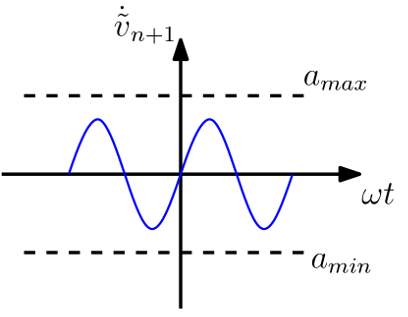}}
    \subcaptionbox{}
    {\includegraphics[width=0.3\textwidth]{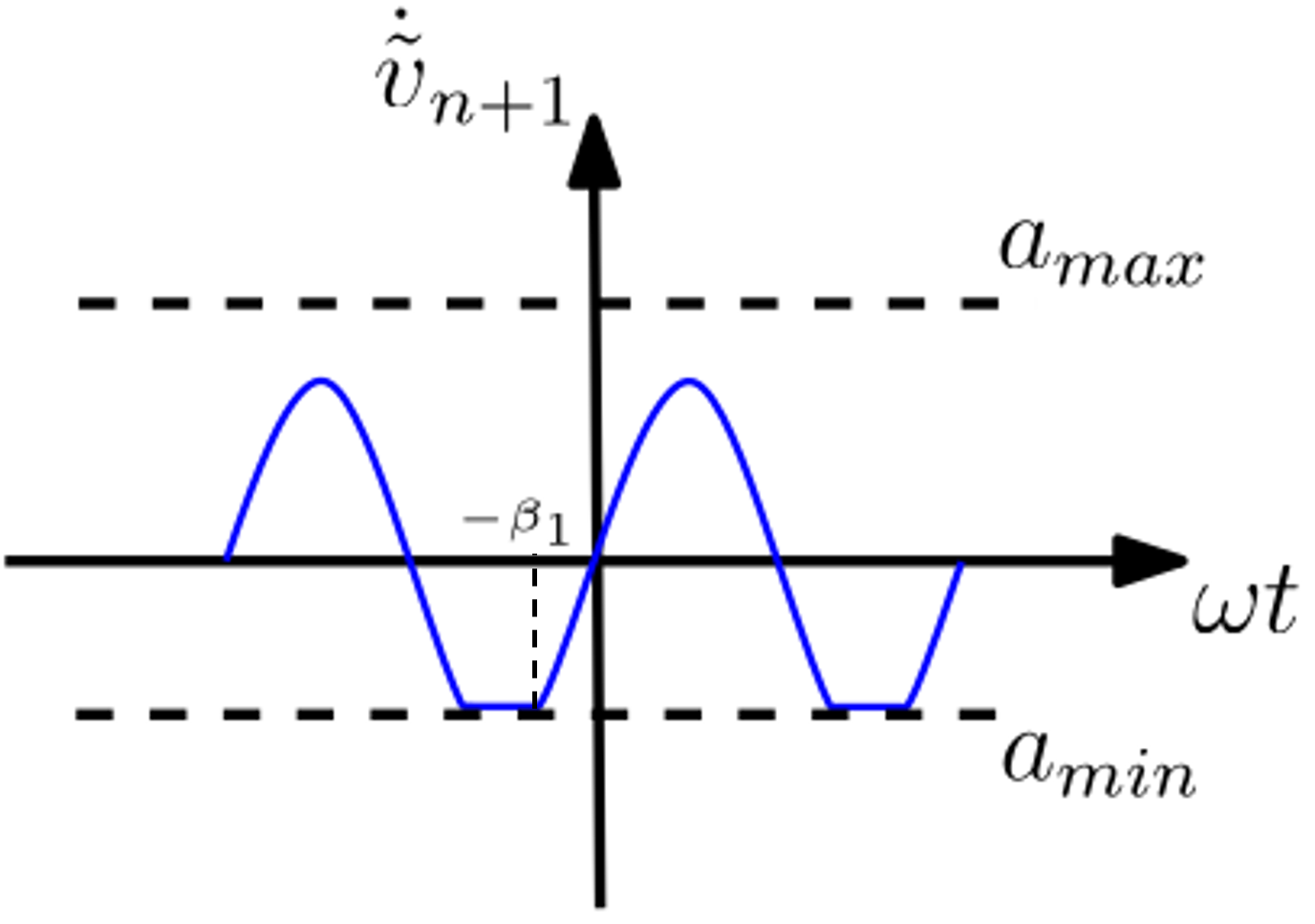}}\\
    \subcaptionbox{}
    {\includegraphics[width=0.29\textwidth]{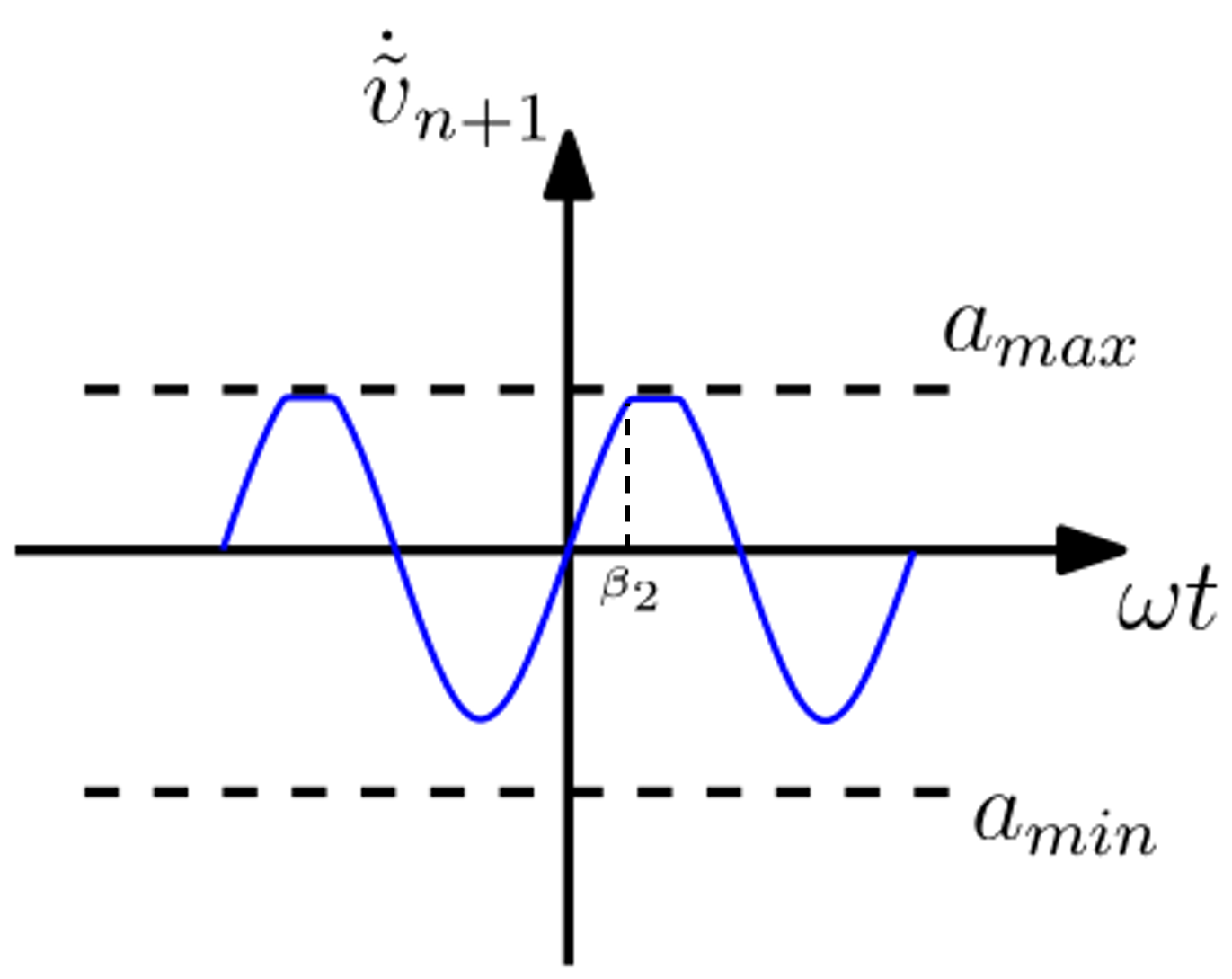}}
    \subcaptionbox{}
    {\includegraphics[width=0.3\textwidth]{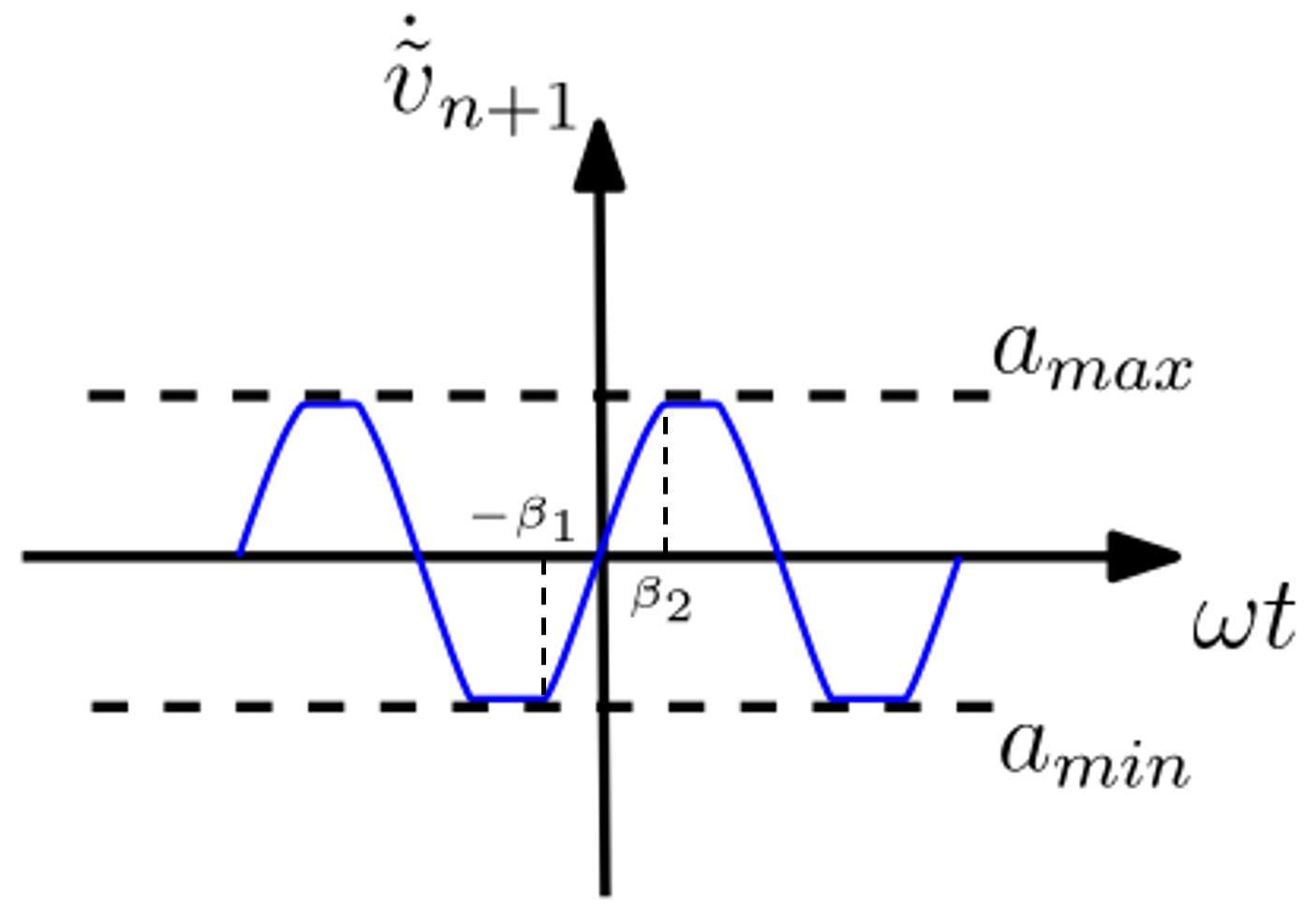}}
    \caption{{Cases of control saturation: (a) Case (1); (b) Case (2); (c) Case (3); (d) Case (4)}}
    \label{fig:case_1}
\end{figure}

The control saturation element doesn't change the input signal, hence

\begin{equation}
\begin{aligned}
N_{a,sat}(B_{a}) &= 1 
\end{aligned}
\label{eq:a DF 1}
\end{equation}

\textbf{Case (2)}: limit activeness by $a_{min}$ (i.e., $-a_{min}\leq B_a\leq a_{max}$), as illustrated in Fig. \ref{fig:case_1}(b), where $\beta_1=sin^{-1}(\frac{-a_{min}}{B_{a}})$.

Within one period, $\dot{\tilde{v}}_{n+1}(t)$ becomes
\begin{equation}
\dot{\tilde{v}}_{n+1}(t) = 
\begin{cases} 
B_a \sin(\omega t) & \text{if } -\beta_1 < \omega t \leq \pi + \beta_1 \\
a_{\min} & \text{if } \pi + \beta_1 < \omega t \leq 2\pi - \beta_1
\end{cases}
\end{equation}

\noindent by Eqs. (\ref{eq:Y11})-(\ref{eq:DF general}), we have
\begin{equation}
Y_{1,1} = \frac{B_a}{2} - \frac{1}{\pi} a_{\min} \sqrt{1 - \left(\frac{a_{\min}^2}{B_a^2}\right)} - \frac{B_a}{\pi} \sin^{-1}\left(\frac{a_{\min}}{B_a}\right)
\end{equation}

\begin{equation}
Y_{1,2} = 0
\end{equation}

\begin{equation}
N_{a,sat}(B_a) = \frac{1}{2} - \frac{1}{\pi B_a} a_{\min} \sqrt{1 - \left(\frac{a_{\min}^2}{B_a^2}\right)} - \frac{1}{\pi} \sin^{-1}\left(\frac{a_{\min}}{B_a}\right) \label{control DF2}
\end{equation}

\textbf{Case (3)}: limit activeness by $a_{max}$ (i.e., $a_{max}\leq B_a\leq -a_{min}$), as illustrated in Fig. \ref{fig:case_1}(c), where $\beta_2=sin^{-1}(\frac{a_{max}}{B_{a}})$.

Within one period, $\dot{\tilde{v}}_{n+1}(t)$ becomes
\begin{equation}
\dot{\tilde{v}}_{n+1}(t) = 
\begin{cases} 
B_a \sin(\omega t) & \text{if } 0 < \omega t \leq \beta_2 \\
a_{\max} & \text{if } \beta_2 < \omega t \leq \pi - \beta_2 \\
B_a \sin(\omega t) & \text{if } \pi - \beta_2 < \omega t \leq 2\pi \\
\end{cases}
\end{equation}

\noindent by Eqs. (\ref{eq:Y11})-(\ref{eq:DF general}), we have
\begin{equation}
Y_{1,1} = \frac{B_a}{2} + \frac{1}{\pi} a_{\max} \sqrt{1 - \left(\frac{a_{\max}^2}{B_a^2}\right)} + \frac{B_a}{\pi} \sin^{-1}\left(\frac{a_{\max}}{B_a}\right)
\end{equation}

\begin{equation}
Y_{1,2} = 0
\end{equation}

\begin{equation}
N_{a,sat}(B_a) = \frac{1}{2} + \frac{1}{\pi B_a} a_{\max} \sqrt{1 - \left(\frac{a_{\max}^2}{B_a^2}\right)} + \frac{1}{\pi} \sin^{-1}\left(\frac{a_{\max}}{B_a}\right) \label{control DF3}
\end{equation}

\textbf{Case (4)}: limit activeness by both $a_{max}$ and $a_{min}$ (i.e., $B_a\geq a_{max}$ and $B_a\geq -a_{min}$), as illustrated in Fig. \ref{fig:case_1}(d).

Within one period, $\dot{\tilde{v}}_{n+1}(t)$ becomes
\begin{equation}
\dot{\tilde{v}}_{n+1}(t) = 
\begin{cases} 
B_a \sin(\omega t) & \text{if } 0 < \omega t \leq \beta_2 \\
a_{\max} & \text{if } \beta_2 < \omega t \leq \pi - \beta_2 \\
B_a \sin(\omega t) & \text{if } \pi - \beta_2 < \omega t \leq \pi+\beta_1 \\
a_{\min} & \text{if } \pi + \beta_1 < \omega t \leq 2\pi - \beta_1\\
B_a \sin(\omega t) & \text{if } 2\pi - \beta_1 < \omega t \leq 2\pi
\end{cases}
\end{equation}

\noindent by Eqs. (\ref{eq:Y11})-(\ref{eq:DF general}), we have
\begin{equation}
Y_{1,1} = \frac{1}{\pi} \left( a_{\max} \sqrt{1 - \left(\frac{a_{\max}^2}{B_a^2}\right)} - a_{\min} \sqrt{1 - \left(\frac{a_{\min}^2}{B_a^2}\right)} + B_a \sin^{-1}\left(\frac{a_{\max}}{B_a}\right) - B_a \sin^{-1}\left(\frac{a_{\min}}{B_a}\right) \right)
\end{equation}

\begin{equation}
Y_{1,2} = 0
\end{equation}

\begin{equation}
N_{a,sat}(B_a) = \frac{1}{\pi} \left( \frac{a_{\max}}{B_a} \sqrt{1 - \left(\frac{a_{\max}^2}{B_a^2}\right)} - \frac{a_{\min}}{B_a} \sqrt{1 - \left(\frac{a_{\min}^2}{B_a^2}\right)} +  \sin^{-1}\left(\frac{a_{\max}}{B_a}\right) - \sin^{-1}\left(\frac{a_{\min}}{B_a}\right) \right) \label{control DF4}
\end{equation}

{\begin{rem}
    The DF analysis performs best under symmetric limits but remains applicable to cases with moderate asymmetry, as shown in Cases (2) and (3) (illustrated in Fig. \ref{fig:case_1}(b)-(c)). For highly asymmetric limits in acceleration or velocity, the dual-input describing function (DIDF) method can be used to introduce a bias term, effectively reconstructing symmetry (Vander Velde et al., Chapter 6). The process of applying DIDF follows similarly for both acceleration and velocity limits, and DIDF has been successfully applied in various fields to handle systems with asymmetric nonlinear effects (\citep{zacharias1981influence,chen2009modeling,shah2019large}). However, since this paper introduces a novel framework for the first time, we prioritize clarity and foundational development over added complexity. Therefore, we focus on pure or nearly symmetric limits, leaving the explicit treatment of highly asymmetric cases for future work.
\end{rem}
\begin{rem}
    The cases in this subsection only describe how a signal passes through the saturation element without considering the closed-loop effect. The closed-loop response is inherently much more complex than the open-loop case due to feedback interactions (illustrated in Fig. 2). A full characterization of the closed-loop behavior requires analyzing how signals propagate throughout the feedback loop. This is systematically addressed in Section 4, where we examine the overall system response.
\end{rem}}

\subsubsection{State saturation \label{sec3.2.2}}
Similar, the DF of the state saturation (Saturation 2 in Fig. \ref{fig:block diagram}) with input $\tilde{v}_{n+1}(t)=B_v \sin(\omega t)$ and its output limited within $[\tilde{v}_{min},\tilde{v}_{max}]$ can be derived in the same manner:

\textbf{Case (1)}: limit inactive (i.e., $B_v\leq-\tilde{v}_{min}$ and $B_v\leq \tilde{v}_{max}$).

\begin{equation}
\begin{aligned}
N_{v,sat}(B_{v}) &= 1 \label{state DF1}
\end{aligned}
\end{equation}

\textbf{Case (2)}: limit activeness by $\tilde{v}_{min}$ (i.e., $-\tilde{v}_{min}\leq B_v\leq \tilde{v}_{max}$).

\begin{equation}
N_{v,sat}(B_v) = \frac{1}{2} - \frac{1}{\pi B_v} \tilde{v}_{\min} \sqrt{1 - \left(\frac{\tilde{v}_{\min}^2}{B_v^2}\right)} - \frac{1}{\pi} \sin^{-1}\left(\frac{\tilde{v}_{\min}}{B_v}\right) \label{state DF2}
\end{equation}

\textbf{Case (3)}: limit activeness by $\tilde{v}_{max}$ (i.e., $\tilde{v}_{max}\leq B_v\leq -\tilde{v}_{min}$).

\begin{equation}
N_{v,sat}(B_v) = \frac{1}{2} + \frac{1}{\pi B_v} \tilde{v}_{\max} \sqrt{1 - \left(\frac{\tilde{v}_{\max}^2}{B_v^2}\right)} + \frac{1}{\pi} \sin^{-1}\left(\frac{\tilde{v}_{\max}}{B_v}\right) \label{state DF3}
\end{equation}

\textbf{Case (4)}: limit activeness by both $\tilde{v}_{max}$ and $\tilde{v}_{min}$ (i.e., $B_v\geq \tilde{v}_{max}$ and $B_v\geq -\tilde{v}_{min}$).

\begin{equation}
N_{v,sat}(B_v) = \frac{1}{\pi} \left( \frac{\tilde{v}_{\max}}{B_v} \sqrt{1 - \left(\frac{\tilde{v}_{\max}^2}{B_v^2}\right)} - \frac{\tilde{v}_{\min}}{B_v} \sqrt{1 - \left(\frac{\tilde{v}_{\min}^2}{B_v^2}\right)} +  \sin^{-1}\left(\frac{\tilde{v}_{\max}}{B_v}\right) - \sin^{-1}\left(\frac{\tilde{v}_{\min}}{B_v}\right) \right) \label{state DF4}
\end{equation}

The deductions in \ref{sec3.2.1} and \ref{sec3.2.2} reveal two significant characteristics of a saturation nonlinear element: (1) Given the saturation limits, the DF of saturation is independent of the frequency of its input and solely depends on the amplitude of the input; (2) The phase shift of a single sinusoidal input passing through a saturation, computed by $\tan^{-1}(\frac{Y_{1,2}}{Y_{1,1}})$, is invariably zero.

For illustration purposes, we plot the DFs of symmetric control saturation elements (i.e., $a_{min}=-a_{max}$) with different limits, in relation to the input amplitude, as depicted in Fig. \ref{fig:N_sat,a}. It is observed that the DFs initiate at a value of 1 and subsequently diminish as the input amplitude escalates beyond the established limits.

\begin{figure}[h]
    \centering
    \setlength{\abovecaptionskip}{0pt}
    \includegraphics[width=0.6\textwidth]{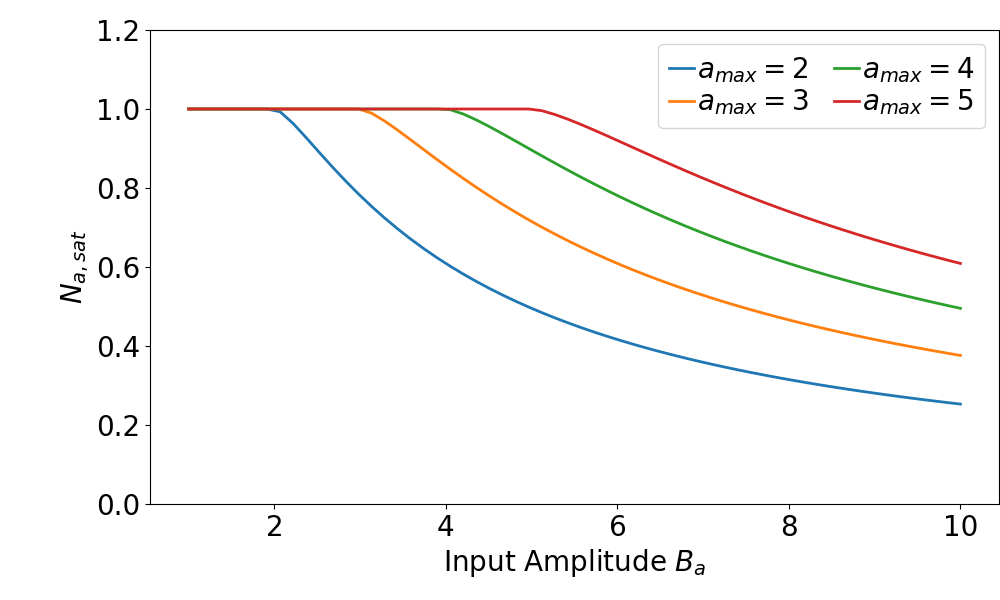}
    \caption{DF of symmetric control saturation vs. input amplitude for different values of $a_{max}$}
    \label{fig:N_sat,a}
\end{figure}

\section{Frequency response analysis based on the DF method \label{sec4}}

Having derived the DFs, each nonlinear element can be substituted by its corresponding DF, allowing us to approximate the output with its first harmonic components, as illustrated in  Fig. \ref{fig:DF block diagram}. Leveraging this DF method, the first harmonics balance is applied to analyze the frequency response (the steady response subjected to a sinusoidal input) of the AV, by equating the first harmonic terms through the block diagram. By the equating, oscillation candidates are found. Specifically, we first analyze the limit cycles to ensure there is no intrinsic oscillation in the controlled AV system in Subsection \ref{sec4.1}. Subsequently, Subsection \ref{sec4.2} conducts the analysis on the frequency response of the closed-loop AV system with various saturation. 

\begin{figure}[h]
    \centering
    \setlength{\abovecaptionskip}{0pt}
    \includegraphics[width=0.8\textwidth]{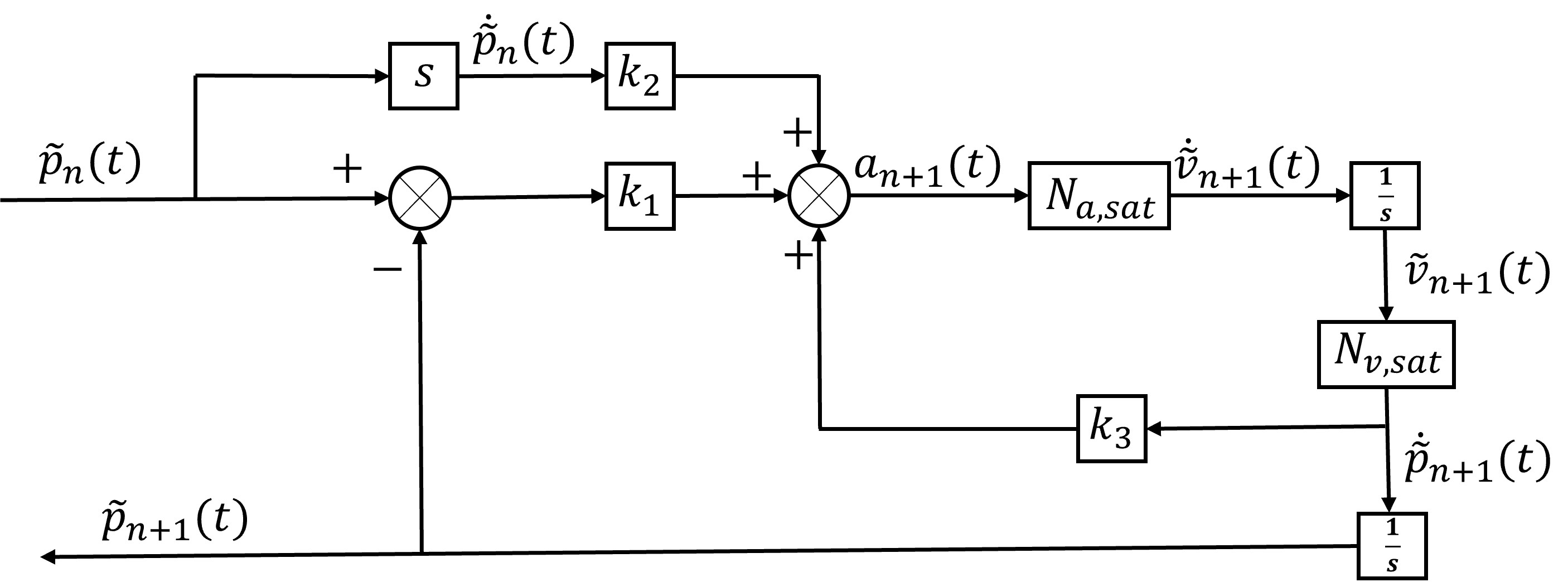}
    \caption{AV system block diagram with DF substitutions}
    \label{fig:DF block diagram}
\end{figure}

\subsection{Non-existence of limit cycles \label{sec4.1}}

Prior to analyzing the frequency response of the closed-loop system utilizing the DF method, it is imperative to demonstrate that limit cycles (system intrinsic oscillation occurring when $\tilde{p}_n(t)=0$) are non-existent within such a system, meaning that when the predecessor operates at the equilibrium states without oscillation, the following AV does not initiate a stable oscillation on its own. This proof ensures that the following vehicle oscillation is directly caused by the leading vehicle rather than the system intrinsic oscillation, as indicated by (\citep{vander1968multiple,atherton2011introduction}).

\begin{prop}
For a controlled AV system characterized by Eqs. (\ref{eq:AV model1})-(\ref{eq:AV model3}), limit cycles do not exist.
\end{prop}
\begin{proof}
This can be proved through the principle of first harmonics balance. With $\tilde{p}_n(t)=0$, suppose a limit cycle exists such that there's a self-sustained oscillation $\tilde{p}_{n+1}(t)=A\sin (\omega t)$, we proceed to equate the first harmonic terms throughout the loop:

\begin{equation}
\begin{aligned}
 A\sin(\omega t)=&-k_1 |G_{int}|^2|N_{a,sat}||N_{v,sat}|A\sin (\omega t+2\angle G_{int}+\angle N_{a,sat}+\angle N_{sat,v})\\
 &+k_3 |G_{der}||G_{int}|^2|N_{a,sat}||N_{v,sat}|A\sin (\omega t+\angle G_{der}+2\angle G_{int}+\angle N_{a,sat}+\angle N_{sat,v})
\end{aligned}
\end{equation}

\noindent where $G_{int}=\frac{1}{\omega}$ and $\angle G_{int}=-\frac{\pi}{2}$ are the amplification and phase shift caused by an integral element, $G_{der}=\omega$ and $\angle G_{der}=\frac{\pi}{2}$ are the amplification and phase shift caused by a derivative element. It is shown in Section \ref{sec3} that $\angle N_{a,sat}=\angle N_{sat,v}=0$. Therefore, we further have:

\begin{equation}
\begin{aligned}
 A\sin(\omega t)=\frac{-k_1 |N_{a,sat}||N_{v,sat}|A\sin (\omega t-\pi)}{\omega^2}+\frac{k_3 |N_{a,sat}||N_{v,sat}|A\sin (\omega t-\frac{\pi}{2})}{\omega} \label{limit cycle balance}
\end{aligned}
\end{equation}

Note that $\sin(\omega t-\pi)=-\sin(\omega t)$, therefore an alternative of Eq. (\ref{limit cycle balance}) is:
\begin{equation}
\begin{aligned}
 (A-\frac{k_1 |N_{a,sat}||N_{v,sat}|A}{\omega^2})\sin (\omega t)=\frac{k_3 |N_{a,sat}||N_{v,sat}|A}{\omega}\sin (\omega t-\frac{\pi}{2}) \label{limit cycle balance2}
\end{aligned}
\end{equation}

From equation Eq. (\ref{limit cycle balance2}), we can find that unless $A=0$, a phase balance cannot be achieved between the two sides of the equation $\forall ~t$. Consequently, limit cycles are non-existent for the system described by Eqs. (\ref{eq:AV model1})-(\ref{eq:AV model3}).
\end{proof}

For closed-loop systems incorporating only one of the saturation elements, the non-existence of limit cycles can be proved similarly. By this proof, we find that the oscillation of the following AV is only caused by the leading vehicle's oscillation, rather than the intrinsic system oscillation. Based on the property, we further conduct the frequency response analysis as below.

\subsection{Frequency response analysis\label{sec4.2}}
Based on the DF method, the principle of first harmonics balance is also instrumental in deriving the conditions that a frequency response must meet. We investigate systems with solely control saturation, solely state saturation, and, ultimately, systems integrating both control and state saturation, to sequentially demonstrate the idea. The signs of the parameters discussed in Remark \ref{rem1} should be kept in mind throughout the analysis of this subsection.

\subsubsection{Systems with solely control saturation \label{sec4.2.1}}
Systems incorporating only control saturation as the singular nonlinear element can be conceptualized as omitting $N_{v,sat}$ in Fig. \ref{fig:DF block diagram} or, equivalently, by maintaining $N_{v,sat}=1$ consistently. Suppose the input to the AV (i.e., the oscillatory trajectory of its predecessor) is represented as $\tilde{p}_n(t)=R\sin(\omega t)$. Additionally, let the input to the control saturation be $a_{n+1}(t)=B_a\sin(\omega t+\phi_a)$, where $B_a$ denotes the amplitude of $a_{n+1}(t)$ and $\phi_a$ represents the phase shift of $a_{n+1}(t)$ relative to $\tilde{p}_n(t)$. Both $B_a$ and $\phi_a$ are unknown variables to be determined. We proceed with the analysis by equating the first harmonic terms of the input and output of the control saturation:

\begin{equation}
\begin{aligned}
B_a \sin(\omega t + \phi_a) =&  - |G_{int}|^2 k_1 |N_{a,sat}(B_a)| B_a \sin(\omega t + \phi_a + \angle N_{a,sat}(B_a) + 2 \angle G_{int}) + k_{2} |G_{der}| R \sin(\omega t + \angle G_{der}) \\
& + k_{3} |G_{der}| |G_{int}|^2 |N_{a,sat}(B_a)| B_a \sin(\omega t + \phi_a + \angle N_{a,sat}(B_a) + \angle G_{der} + 2 \angle G_{int})+k_1 R \sin(\omega t) \label{control balance 1}
\end{aligned}
\end{equation}

\noindent where as mentioned in Subsection \ref{sec4.1}, $|G_{int}|=\frac{1}{\omega}$ and $\angle G_{int}=-\frac{\pi}{2}$ are the amplification and phase shift caused by an integral element, $|G_{der}|=\omega$ and $\angle G_{der}=\frac{\pi}{2}$ are the amplification and phase shift caused by a derivative element. Additionally, it is shown in Section \ref{sec3} that $\angle N_{a,sat}(B_a)=0$. Substituting these values into Eq. (\ref{control balance 1}):
\begin{equation}
\begin{aligned}
B_a \sin(\omega t + \phi_a) =& k_1 R \sin(\omega t) -  \frac{k_1 |N_{a,sat}(B_a)| B_a}{\omega^2} \sin(\omega t + \phi_a -\pi) + \omega k_{2} R \sin(\omega t + \frac{\pi}{2}) 
\\& + \frac{k_{3}|N_{a,sat}(B_a)| B_a}{\omega} \sin(\omega t + \phi_a -\frac{\pi}{2}) \label{control balance 2}
\end{aligned}
\end{equation}

Further by $\sin(\omega t-\pi)=-\sin(\omega t)$ and $\sin(\omega t-\frac{\pi}{2})=-\cos(\omega t)$:
\begin{equation}
B_a (1 - \frac{k_1 |N_{a,sat}(B_a)|}{\omega^2}) \sin(\omega t + \phi_a) + \frac{k_{3} |N_{a,sat}(B_a)| B_a}{\omega} \cos(\omega t + \phi_a) = k_1 R \sin(\omega t) + \omega k_{2} R \cos(\omega t) \label{control subtitute sin-pi}
\end{equation}

Finally, by the trigonometric sum and difference identities, or more directly, by the auxiliary angle method (\citep{clarke1953}), Eq. (\ref{control subtitute sin-pi}) can be simplified and classified into two cases depending on the sign of $1 - \frac{k_1 |N_{a,sat}(B_a)|}{\omega^2}$:
\begin{equation}
\begin{aligned}
&\sqrt{B_a^2 \left(1 - \frac{k_1 |N_{a,sat}(B_a)|}{\omega^2}\right)^2 + \left(\frac{k_3 |N_{a,sat}(B_a)| B_a}{\omega}\right)^2} \sin\left(\omega t + \phi_a - \tan^{-1}\left(\frac{-\omega k_3 |N_{a,sat}(B_a)|}{\omega^2 - k_1 |N_{a,sat}(B_a)|}\right)\right) \\&= \sqrt{(k_1 R)^2 + (\omega k_2 R)^2} \sin\left(\omega t + \tan^{-1}\left(\frac{k_2 \omega}{k_1}\right)\right) \quad\quad \text{if } 1 - \frac{k_1 |N_{a,sat}(B_a)|}{\omega^2} \geq 0 \label{control triag1}
\end{aligned}
\end{equation}

\begin{equation}
\begin{aligned}
&\sqrt{B_a^2 \left(1 - \frac{k_1 |N_{a,sat}(B_a)|}{\omega^2}\right)^2 + \left(\frac{k_3 |N_{a,sat}(B_a)| B_a}{\omega}\right)^2} \sin\left(\omega t + \phi_a - \tan^{-1}\left(\frac{-\omega k_3 |N_{a,sat}(B_a)|}{\omega^2 - k_1 |N_{a,sat}(B_a)|}\right)+\pi\right) \\&= \sqrt{(k_1 R)^2 + (\omega k_2 R)^2} \sin\left(\omega t + \tan^{-1}\left(\frac{k_2 \omega}{k_1}\right)\right) \quad\quad \text{if } 1 - \frac{k_1 |N_{a,sat}(B_a)|}{\omega^2} < 0 \label{control triag2}
\end{aligned}
\end{equation}

By equating the amplitude and phase respectively of the two sides of Eq. (\ref{control triag1}) or Eq. (\ref{control triag2}) at each frequency point, we can solve for candidates of oscillation characterized by $B_a$ and $\phi_a$ (it is possible to have multiple candidates at each frequency). At a specific frequency $\omega$, let $B_{a,s}$ and $\phi_{a,s}$ be a pair of solutions obtained from Eq. (\ref{control triag1}) or Eq. (\ref{control triag2}), the corresponding potential frequency responses of the closed-loop system are calculated as follows: 
\begin{equation}
\begin{aligned}
|F(\omega,B_{a,s},\phi_{a,s})|=\frac{A_{\tilde{p},{n+1}}}{R}=\frac{B_{a,s}|N_{a,sat}(B_{a,s})|}{\omega^2R} \label{control FR mag}
\end{aligned}
\end{equation}

\begin{equation}
\begin{aligned}
\angle F(\omega,B_{a,s},\phi_{a,s})=\angle\tilde{p}_{n+1}-0=\phi_{a,s}-\pi \label{control FR phase}
\end{aligned}
\end{equation}

where $F(\omega,B_{a,s},\phi_{a,s})$ is the potential frequency response at the frequency point $\omega$ corresponding to candidates $B_{a,s}$ and $\phi_{a,s}$. $|F(\omega,B_{a,s},\phi_{a,s})|$ and $\angle F(\omega,B_{a,s},\phi_{a,s})$ are the magnitude and phase of the potential frequency response, respectively (\citep{franklin2002feedback}). $A_{\tilde{p},{n+1}}$ and $\angle\tilde{p}_{n+1}$ represent the amplitude and phase of $\tilde{p}_{n+1}$, respectively. The magnitude of frequency response captures the damping ratio of the oscillation, characterized by the amplitude ratio of $\tilde{p}_{n+1}$ and $\tilde{p}_n$. The phase of frequency response captures the phase difference between $\tilde{p}_{n+1}$ and $\tilde{p}_n$, where the phase of $\tilde{p}_n$ is set as the reference thus equals zero. An equivalent complex representation of the frequency response is $F(\omega,B_{a,s},\phi_{a,s})=|F(\omega,B_{a,s},\phi_{a,s})|e^{j\angle F(\omega,B_{a,s},\phi_{a,s})}$.

\subsubsection{Systems with solely state saturation \label{sec4.2.2}}
Similarly, systems incorporating only state saturation as the singular nonlinear element can be conceptualized as omitting $N_{a,sat}$ in Fig. \ref{fig:DF block diagram} or, equivalently, by maintaining $N_{a,sat}=1$ consistently. Suppose the input to the AV (i.e., the oscillatory trajectory of its predecessor) is represented as $\tilde{p}_n(t)=R\sin(\omega t)$. And let the input to the state saturation be $\tilde{v}_{n+1}(t)=B_v\sin(\omega t+\phi_v)$, where $B_v$ denotes the amplitude of $\tilde{v}_{n+1}(t)$ and $\phi_v$ represents the phase shift of $\tilde{v}_{n+1}(t)$ relative to $\tilde{p}_n(t)$. Both $B_v$ and $\phi_v$ are unknown variables to be determined. The analysis is conducted by equating the first harmonic terms of the input and output of the state saturation:

\begin{equation}
\begin{aligned}
B_v \sin(\omega t + \phi_v)=&|G_{int}| k_1 R \sin(\omega t + \angle G_{int}) - |G_{int}|^2 k_1 |N_{v, sat}(B_v)| B_v \sin(\omega t + \phi_v + \angle N_{v, sat}(B_v) + 2 \angle G_{int})\\&  + k_3 |G_{der}| |G_{int}|^2 |N_{v, sat}(B_v)| B_v \sin(\omega t + \phi_v + \angle N_{v, sat}(B_v) + \angle G_{der} + 2 \angle G_{int}) \\
&+ k_2 |G_{der}| |G_{int}| R \sin(\omega t + \angle G_{der} + \angle G_{int}) \label{state balance 1}
\end{aligned}
\end{equation}

Substituting the values of $|G_{int}|=\frac{1}{\omega}$, $\angle G_{int}=-\frac{\pi}{2}$, $|G_{der}|=\omega$, $\angle G_{der}=\frac{\pi}{2}$ , $\angle N_{sat,v}=0$ into Eq. (\ref{state balance 1}) and rearranging with $\sin(\omega t-\pi)=-\sin(\omega t)$ and $\sin(\omega t-\frac{\pi}{2})=-\cos(\omega t)$:
\begin{equation}
\begin{aligned}
-\frac{k_1 R \cos(\omega t)}{\omega} + \frac{k_1 |N_{v, sat}(B_v)| B_v \sin(\omega t + \phi_v)}{\omega^2} + k_2 R \sin(\omega t) - \frac{k_3 |N_{v, sat}(B_v)| B_v \cos(\omega t + \phi_v)}{\omega} = B_v \sin(\omega t + \phi_v) \label{state subtitute sin-pi}
\end{aligned}
\end{equation}

Finally, by the trigonometric sum and difference identities or the auxiliary angle method, Eq. (\ref{state subtitute sin-pi}) can be simplified and classified into two cases depending on the sign of $1 - \frac{k_1 |N_{v, sat}(B_v)|}{\omega^2}$:
\begin{equation}
\begin{aligned}
&\sqrt{B_v^2 \left(1 - \frac{k_1 |N_{v, sat}(B_v)|}{\omega^2}\right)^2 + \left(\frac{k_3 |N_{v, sat}(B_v)| B_v}{\omega}\right)^2} \sin\left(\omega t + \phi_v - \tan^{-1}\left(\frac{-\omega k_3 |N_{v, sat}(B_v)|}{\omega^2 - k_1 |N_{v, sat}(B_v)|}\right)\right) \\
&= \sqrt{\left(\frac{-k_1 R}{\omega}\right)^2 + (k_2 R)^2} \sin\left(\omega t - \tan^{-1}\left(\frac{k_1}{k_2 \omega}\right)\right) \quad\quad \text{if } 1 - \frac{k_1 |N_{v, sat}(B_v)|}{\omega^2} \geq 0 \label{state triag1}
\end{aligned}
\end{equation}

\begin{equation}
\begin{aligned}
&\sqrt{B_v^2 \left(1 - \frac{k_1 |N_{v, sat}(B_v)|}{\omega^2}\right)^2 + \left(\frac{k_3 |N_{v, sat}(B_v)| B_v}{\omega}\right)^2} \sin\left(\omega t + \phi_v - \tan^{-1}\left(\frac{-\omega k_3 |N_{v, sat}(B_v)|}{\omega^2 - k_1 |N_{v, sat}(B_v)|}\right)+\pi\right) \\
&= \sqrt{\left(\frac{-k_1 R}{\omega}\right)^2 + (k_2 R)^2} \sin\left(\omega t - \tan^{-1}\left(\frac{k_1}{k_2 \omega}\right)\right) \quad\quad \text{if } 1 - \frac{k_1 |N_{v, sat}(B_v)|}{\omega^2} < 0 \label{state triag2}
\end{aligned}
\end{equation}

By equating the amplitude and phase respectively of the two sides of Eq. (\ref{state triag1}) or Eq. (\ref{state triag2}) at each frequency point, we can solve for candidates of oscillation (again, possibly multiple ones) characterized by $B_v$ and $\phi_v$. Similar to Eqs. (\ref{control FR mag})-(\ref{control FR phase}), the potential frequency response corresponding to a candidate pair $B_{v,s}$ and $\phi_{v,s}$, at frequency $\omega$ is obtained as:

\begin{equation}
\begin{aligned}
|F(\omega,B_{v,s},\phi_{v,s})|=\frac{A_{\tilde{p},{n+1}}}{R}=\frac{B_{v,s}|N_{v,sat}(B_{v,s})|}{\omega R} \label{state FR mag}
\end{aligned}
\end{equation}

\begin{equation}
\begin{aligned}
\angle F(\omega,B_{v,s},\phi_{v,s})=\angle\tilde{p}_{n+1}-0=\phi_{v,s}-\frac{\pi}{2} \label{state FR phase}
\end{aligned}
\end{equation}

\subsubsection{Systems with both control and state saturation \label{sec4.2.3}}
Systems incorporating both control and state saturation are exactly depicted in Fig. \ref{fig:DF block diagram}. Suppose the input to the AV (i.e., the oscillatory trajectory of its predecessor) is represented as $\tilde{p}_n(t)=R\sin(\omega t)$. And let the input to the control saturation be $a_{n+1}(t)=B_a\sin(\omega t+\phi_a)$. Both $B_a$ and $\phi_a$ are unknown variables to be determined. Suppose the input to the state saturation is $\tilde{v}_{n+1}(t)=B_v\sin(\omega t+\phi_v)$, and it can be obtained that $B_v=\frac{B_a|N_{a,sat}(B_a)|}{\omega}$, $\phi_v=\phi_a-\frac{\pi}{2}$ The analysis is conducted by equating the first harmonic terms of the input and output of the control saturation:

\begin{equation}
\begin{aligned}
B_a &\sin(\omega t + \phi_a)\\
&=k_1 R \sin(\omega t)  + k_2 |G_{der}| R \sin(\omega t + \angle G_{der}) \\
&- k_1 |G_{int}|^2 |N_{v, sat}(B_v)| |N_{a, sat}(B_a)| B_a \sin(\omega t + \phi_a + 2 \angle G_{int}+\angle N_{v, sat}(B_v)+\angle N_{a, sat}(B_a))\\
&+ k_3 |G_{der}| |G_{int}|^2 |N_{v, sat}(B_v)| |N_{a, sat}(B_a)| B_a \sin(\omega t + \phi_a + \angle G_{der}+\angle N_{v, sat}(B_v)+\angle N_{a, sat}(B_a) + 2 \angle G_{int}) \label{control state balance 1}
\end{aligned}
\end{equation}

Substituting the values of $|G_{int}|=\frac{1}{\omega}$, $\angle G_{int}=-\frac{\pi}{2}$, $|G_{der}|=\omega$, $\angle G_{der}=\frac{\pi}{2}$ , $\angle N_{sat,v}=0$ into Eq. (\ref{control state balance 1}) and rearranging with $\sin(\omega t-\pi)=-\sin(\omega t)$ and $\sin(\omega t-\frac{\pi}{2})=-\cos(\omega t)$:
\begin{equation}
\begin{aligned}
B_a \sin(\omega t + \phi_a)=&k_1 R \sin(\omega t) + \frac{k_1 |N_{v, sat}(B_v)| |N_{a, sat}(B_a)| B_a \sin(\omega t + \phi_a)}{\omega^2} + \omega k_2 R \cos(\omega t) \\
&- \frac{k_3 |N_{v, sat}(B_v)| |N_{a, sat}(B_a)| B_a \cos(\omega t + \phi_a)}{\omega}  \label{both subtitute sin-pi}
\end{aligned}
\end{equation}

Finally, by the trigonometric sum and difference identities or the auxiliary angle method, Eq. (\ref{both subtitute sin-pi}) can be simplified and classified into two cases depending on the sign of $1 - \frac{k_1 |N_{v, sat}(B_v)||N_{a, sat}(B_a)|}{\omega^2}$:
\begin{equation}
\begin{aligned}
&\sqrt{\left(k_1 R\right)^2 + (\omega k_2 R)^2} \sin\left(\omega t + \tan^{-1}\left(\frac{k_2 \omega}{k_1}\right)\right)\\
=&\sqrt{B_a^2 \left(1 - \frac{k_1 |N_{v, sat}(B_v)||N_{a, sat}(B_a)|}{\omega^2}\right)^2 + \left(\frac{k_3 |N_{v, sat}(B_v)||N_{a, sat}(B_a)| B_a}{\omega}\right)^2} \\
&\cdot\sin\left(\omega t + \phi_a -\tan^{-1}\left(\frac{-\omega k_3 |N_{v, sat}(B_v)||N_{a, sat}(B_a)|}{\omega^2 - k_1 |N_{v, sat}(B_v)||N_{a, sat}(B_a)|}\right)\right)   \quad\quad \text{if } 1 - \frac{k_1 |N_{v, sat}(B_v)||N_{a, sat}(B_a)|}{\omega^2} \geq 0 \label{both triag1}
\end{aligned}
\end{equation}

\begin{equation}
\begin{aligned}
&\sqrt{\left(k_1 R\right)^2 + (\omega k_2 R)^2} \sin\left(\omega t + \tan^{-1}\left(\frac{k_2 \omega}{k_1}\right)\right)\\
=&\sqrt{B_a^2 \left(1 - \frac{k_1 |N_{v, sat}(B_v)||N_{a, sat}(B_a)|}{\omega^2}\right)^2 + \left(\frac{k_3 |N_{v, sat}(B_v)||N_{a, sat}(B_a)| B_a}{\omega}\right)^2} \\
&\cdot\sin\left(\omega t + \phi_a -\tan^{-1}\left(\frac{-\omega k_3 |N_{v, sat}(B_v)||N_{a, sat}(B_a)|}{\omega^2 - k_1 |N_{v, sat}(B_v)||N_{a, sat}(B_a)|}\right)+\pi\right)   \quad\quad \text{if } 1 - \frac{k_1 |N_{v, sat}(B_v)||N_{a, sat}(B_a)|}{\omega^2} < 0 \label{both triag2}
\end{aligned}
\end{equation}

By substituting $B_v=\frac{B_a|N_{a,sat}(B_a)|}{\omega}$ and equating the amplitude and phase respectively of the two sides of Eq. (\ref{both triag1}) or Eq. (\ref{both triag2}) at each frequency point, we can solve for candidates of oscillation characterized by $B_a$ and $\phi_a$. The potential frequency response corresponds to a candidate pair $B_{a,s}$ and $\phi_{a,s}$, at frequency $\omega$, similar to Eqs. (\ref{control FR mag})-(\ref{control FR phase}) and Eqs. (\ref{state FR mag})-(\ref{state FR phase}), becomes: 
\begin{equation}
\begin{aligned}
|F(\omega,B_{a,s},\phi_{a,s})|=\frac{A_{\tilde{p},n+1}}{R}=\frac{B_{a,s}|N_{a,sat}(B_{a,s})||N_{v,sat}(\frac{B_{a,s}|N_{a,sat}(B_{a,s})|}{\omega})|}{\omega^2R} \label{both FR mag}
\end{aligned}
\end{equation}

\begin{equation}
\begin{aligned}
\angle F(\omega,B_{a,s},\phi_{a,s})=\angle\tilde{p}_{n+1}-0=\phi_{a,s}-\pi \label{both FR phase}
\end{aligned}
\end{equation}

\begin{rem}
    \label{rem 2}
    An important point to emphasize is that while the DF frequency response analysis extends beyond conventional linear analysis by considering saturation limits, it encompasses the linear analysis when these limits are not reached. We prove this in the following proposition.
\end{rem}

\begin{prop}
    When saturation limits are not reached, the DF frequency response analysis degenerates to the linear frequency response analysis.
\end{prop}
\begin{proof}
    Consider the controlled AV systems with both control and state saturation. When the saturation limits are not reached, by Eqs. (\ref{eq:a DF 1}) and (\ref{state DF1}), $N_{a,sat}(B_a)=N_{v,sat}(B_v)=1$. Therefore, Eq. (\ref{control state balance 1}) degenerates to
    \begin{equation}
\begin{aligned}
B_a \sin(\omega t + \phi_a)
=&k_1 R \sin(\omega t)  + k_2 |G_{der}| R \sin(\omega t + \angle G_{der}) - k_1 |G_{int}|^2  B_a \sin(\omega t + \phi_a + 2 \angle G_{int})\\
&+ k_3 |G_{der}| |G_{int}|^2 B_a \sin(\omega t + \phi_a + \angle G_{der} + 2 \angle G_{int}) \label{control state balance 1 simplify}
\end{aligned}
\end{equation}

Furthermore, without the effect of nonlinearies, $B_a \sin(\omega t + \phi_a)=\Ddot{\tilde{p}}_{n+1}$. By applying the Fourier transform (\citep{franklin2002feedback}) on Eq. (\ref{control state balance 1 simplify}), the following frequency domain relation is obtained:

\begin{equation}
    -\omega^2 \tilde{p}_{n+1}(\omega)=k_1\tilde{p}_n(\omega)+j\omega k_2\tilde{p}_n(\omega)-k_1\tilde{p}_{n+1}(\omega)+j\omega k_3 \tilde{p}_{n+1}(\omega)
\end{equation}
where $\tilde{p}_{n+1}(\omega)$ and $\tilde{p}_n(\omega)$ are the Fourier transform of $\tilde{p}_{n+1}$ and $\tilde{p}_n$, respectively. Hence, the frequency response is obtained:
\begin{equation}
    F(\omega)=\frac{\tilde{p}_{n+1}(\omega)}{\tilde{p}_n(\omega)}=\frac{j\omega k_2+k_1}{-\omega^2-j\omega k_3+k_1}
    \label{frequency response simplify}
\end{equation}

We next derive the linear frequency response analysis through the well-known Laplace transform approach (\citep{franklin2002feedback}). By Eqs. (\ref{eq:AV model1})-(\ref{eq:AV model3}), without the saturation limits, the governing equation of the system becomes:
\begin{equation}
\Ddot{\tilde{p}}_{n+1}(t)=k_{1}(\tilde{p}_n(t)-\tilde{p}_{n+1}(t))+k_{2}\dot{\tilde{p}}_n(t)+k_{3}\dot{\tilde{p}}_{n+1}(t) \label{eq:AV model3 simplify}
\end{equation}

Applying the Laplace transform to the above equation yields:
\begin{equation}
s^2{\tilde{p}}_{n+1}(s)=k_{1}(\tilde{p}_n(s)-\tilde{p}_{n+1}(s))+k_{2}s{\tilde{p}}_n(s)+k_{3}s{\tilde{p}}_{n+1}(s) \label{laplace simplify}
\end{equation}
where $\tilde{p}_{n+1}(s)$ and $\tilde{p}_n(s)$ are the Laplace transform of $\tilde{p}_{n+1}$ and $\tilde{p}_n$, respectively. Hence, the transfer function is obtained:
\begin{equation}
    T(s)=\frac{\tilde{p}_{n+1}(s)}{\tilde{p}_n(s)}=\frac{k_2s+k_1}{s^2-k_3s+k_1}
\end{equation}
The frequency response is obtained by substituting $s=j\omega$ into the transfer function:
\begin{equation}
    T(s=j\omega)=\frac{\tilde{p}_{n+1}(s=j\omega)}{\tilde{p}_n(s=j\omega)}=\frac{j\omega k_2+k_1}{-\omega^2-j\omega k_3+k_1}
\end{equation}
which is clearly identical to Eq. (\ref{frequency response simplify}).

With similar proofs for systems with only one saturation, the proposition is proved.
\end{proof}

\subsubsection{Solving for forced oscillation candidates  \label{sec4.2.4}}
For the three types of systems analyzed above, given the values of $B_{a,s}$ or $B_{v,s}$, the variables $\phi_{a,s}$ or $\phi_{v,s}$ can be analytically determined with ease by equating the phase of both sides of equations Eqs. (\ref{control triag1})-(\ref{control triag2}), Eqs. (\ref{state triag1})-(\ref{state triag2}), or Eqs. (\ref{both triag1})-(\ref{both triag2}) within the studied range of $\omega$. $B_{a,s}$ or $B_{v,s}$ can be deduced by equating the amplitude of the two sides of the aforementioned equations. However, the nonlinearity introduced by the DFs can render the analytical solution of these equations challenging. Several studies have proposed graphical methods to identify the roots of such equations (\citep{vander1968multiple,atherton2011introduction,li2012prediction,maraini2018nonlinear}). Nevertheless, implementing these graphical solutions for heterogeneous AVs can be cumbersome and repetitive. Fortunately, iterative strategies, such as the Brent-Dekker method  (\citep{dekker1969finding,brent2013algorithms}), offer numerical solutions. These strategies are largely dependent on an initial guess, hence, we sweep the initial guess across the range $[0,B_{ini,max}]$ to identify all possible solutions, where $B_{ini,max}$ represents a predetermined maximum value of the initial guess.

While the DF method enables us to derive the necessary conditions and identify oscillation candidates, it does not provide an efficient way to ensure the stability of these candidates (i.e., if the closed-loop system will converge back to the same oscillation after a small perturbation is introduced). A stability analysis on the oscillation candidates is necessary as it validates if the oscillation candidates can be observed in practice and helps us discard unstable solutions. To address this, we introduce the incremental-input DF in Section \ref{sec5} to assess the stability of the identified candidates.

\section{Incremental-input DF analysis and oscillation stability analysis \label{sec5}}
In this section, the incremental-input DF is employed to examine the stability of the candidates of oscillation identified in Section \ref{sec4}. Precisely, Subsection \ref{sec5.1} introduces the concept of the incremental-input DF. Subsequently, Subsection \ref{sec5.2} elaborates on the derivation of incremental-input DFs. Conclusively, Subsection \ref{sec5.3} delves into the oscillation candidates stability analysis.

\subsection{The incremental-input DF \label{sec5.1}}
The incremental-input DF, a specialized form of the two-sinusoidal-input describing function (\citep{gibson1963new,west1956dual}), characterizes the impact exerted by a nonlinear element on a sinusoidal perturbation centered around a designated forced oscillation candidate (\citep{bonenn1963relative}). The decay of such perturbations signifies the stability of the pinpointed candidate oscillation. Given that our investigation does not account for sub-harmonics, and the AV systems exhibit a low-pass nature that filters out higher-frequency components, we focus on a category of perturbations consistent in frequency with $\tilde{p}_n(t)$. Consequently, the generic input to a nonlinear saturation element, $u(t)$, is expressed as follows:
\begin{equation}
\begin{aligned}
u(t)=B_s\sin(\omega t)+\epsilon\sin(\omega t+\theta) \label{incremental input}
\end{aligned}
\end{equation}

\noindent where the first term is the predominant input which is one of the oscillation candidates acquired from the analysis in Section \ref{sec4}. It's shifted on the time axis so that its first cycle starts at $t=0$, simplifying notation without affecting the analysis due to our interest in the phase relation between primary and incremental inputs. The second term represents the incremental input which can be viewed as a perturbation to the predominant input. In this context, $\epsilon\ll B_s$, indicating the perturbation's magnitude is notably less than the predominant input. Additionally, $\theta$ represents the phase shift of the incremental input in relation to the predominant input, which can span any value in the interval $[0,2\pi]$. 
{\begin{rem}
     Due to the low-pass characteristic of AV systems and the assumption of no subharmonics on the saturation, the output of the saturation remains sinusoidal with the same frequency as the external input to the AV system (i.e., the oscillatory position of the AV’s predecessor). The input to the nonlinear element, being a linear combination of the external input and feedback signals, also remains sinusoidal at the same frequency by trigonometric identities. Therefore, the predominant and perturbation terms are both considered sinusoidal with the same frequency, and perturbations are analyzed with all possible phases.
\end{rem}}

By expanding and regrouping Eq. (\ref{incremental input}) and omitting higher order terms in $\frac{\epsilon}{B_s}$, the incremental-input DF $N_{inc}(B_s,\theta)$, which is a complex ratio of the output caused by $\epsilon \sin(\omega t+\phi)$ and $B_s \sin(\omega t)$ itself, can be readily derived following \citep{vander1968multiple,bonenn1963relative}:
\begin{equation}
\begin{aligned}
N_{inc}(B_s,\theta)=N(B_s)+\frac{B_s}{2}\cdot\frac{dN(B_s)}{dB_s}(1+e^{-j2\theta}) \label{incremental DF}
\end{aligned}
\end{equation}

\noindent where $N(B_s)$ denotes the DF of a nonlinear element with an input of amplitude $B_s$, as elaborated in Subsection \ref{sec3.1}. The incremental-input DF is shown to be invariant to both the frequency of the oscillation candidate under scrutiny and the magnitude of the incremental-input.

\subsection{Incremental-input DFs of control and state saturation\label{sec5.2}}
In this subsection, we conduct an analysis of the incremental-input DF for both the control saturation and state saturation elements.

\subsubsection{Control saturation \label{sec5.2.1}}
For the control saturation element, let the predominant input of it be denoted as $a_{n+1}(t)=B_{a,s}sin(\omega t)$, where $B_{a,s}$ can be derived from Eqs. (\ref{control triag1})-(\ref{control triag2}) or Eqs. (\ref{both triag1})-(\ref{both triag2}). Considering the acceleration limits as $[a_{min}, a_{max}]$. It has been shown that the DF of the control saturation, $N_{a,sat}(B_{a,s})$, can be characterized into four cases depending on different limit activeness. Analogously, the incremental-input DF, $N_{inc}(B_{a,s},\theta_a)$, can also be classified into these four scenarios:

\textbf{Case (1)}: limit inactive (i.e., $B_{a,s}\leq-a_{min}$ and $B_{a,s}\leq a_{max}$).

\begin{equation}
\begin{aligned}
N_{a,sat}(B_{a,s}) &= 1 
\end{aligned}
\end{equation}

\noindent and the corresponding incremental-input DF can be obtained by Eq. (\ref{incremental DF}):

\begin{equation}
\begin{aligned}
N_{a,inc}(B_{a,s},\theta_a)&=N_{a,sat}(B_{a,s})+\frac{B_{a,s}}{2}\cdot\frac{dN_{a,sat}(B_{a,s})}{dB_{a,s}}(1+e^{-j2\theta_a})\\
&=1\label{control incremental DF1}
\end{aligned}
\end{equation}

\textbf{Case (2)}: limit activeness by $a_{min}$ (i.e., $-a_{min}\leq B_{a,s}\leq a_{max}$). 
\begin{equation}
N_{a,sat}(B_{a,s}) = \frac{1}{2} - \frac{1}{\pi B_{a,s}} a_{\min} \sqrt{1 - \left(\frac{a_{\min}^2}{B_{a,s}^2}\right)} - \frac{1}{\pi} \sin^{-1}\left(\frac{a_{\min}}{B_{a,s}}\right)
\end{equation}

\noindent correspondingly, by Eq. (\ref{incremental DF}):
\begin{equation}
\begin{aligned}
N_{a,inc}(B_{a,s},\theta_a)=&N_{a,sat}(B_{a,s})+\frac{B_{a,s}}{2}\cdot\frac{dN_{a,sat}(B_{a,s})}{dB_{a,s}}(1+e^{-j2\theta_a})\\
=&\frac{1}{2} - \frac{1}{\pi B_{a,s}} a_{\min} \sqrt{1 - \left(\frac{a_{\min}^2}{B_{a,s}^2}\right)} - \frac{1}{\pi} \sin^{-1}\left(\frac{a_{\min}}{B_{a,s}}\right)+(\frac{a_{min}\sqrt{B_{a,s}^2-a_{min}^2}}{B_{a,s}^2\pi})(1+e^{-j2\theta_a})\\
=&\frac{1}{2}+\frac{e^{-j2\theta_a}}{B_{a,s}\pi}a_{min}\sqrt{1-\frac{a_{min}^2}{B_{a,s}^2}}-\frac{1}{\pi}sin^{-1}(\frac{a_{min}}{B_{a,s}})\label{control incremental DF2}
\end{aligned}
\end{equation}

\textbf{Case (3)}: limit activeness by $a_{max}$ (i.e., $a_{max}\leq B_{a,s}\leq -a_{min}$).
\begin{equation}
N_{a,sat}(B_{a,s}) = \frac{1}{2} + \frac{1}{\pi B_{a,s}} a_{\max} \sqrt{1 - \left(\frac{a_{\max}^2}{B_{a,s}^2}\right)} + \frac{1}{\pi} \sin^{-1}\left(\frac{a_{\max}}{B_{a,s}}\right)
\end{equation}

\noindent correspondingly, by Eq. (\ref{incremental DF}):
\begin{equation}
\begin{aligned}
N_{a,inc}(B_{a,s},\theta_a)=&N_{a,sat}(B_{a,s})+\frac{B_{a,s}}{2}\cdot\frac{dN_{a,sat}(B_{a,s})}{dB_{a,s}}(1+e^{-j2\theta_a})\\
=&\frac{1}{2} + \frac{1}{\pi B_{a,s}} a_{\max} \sqrt{1 - \left(\frac{a_{\max}^2}{B_{a,s}^2}\right)} + \frac{1}{\pi} \sin^{-1}\left(\frac{a_{\max}}{B_{a,s}}\right)+(-\frac{a_{max}\sqrt{B_{a,s}^2-a_{max}^2}}{B_{a,s}^2\pi})(1+e^{-j2\theta_a})\\
=&\frac{1}{2}-\frac{e^{-j2\theta_a}}{B_{a,s}\pi}a_{max}\sqrt{1-\frac{a_{max}^2}{B_{a,s}^2}}+\frac{1}{\pi}sin^{-1}(\frac{a_{max}}{B_{a,s}})\label{control incremental DF3}
\end{aligned}
\end{equation}

\textbf{Case (4)}: limit activeness by both $a_{max}$ and $a_{min}$ (i.e., $B_{a,s}\geq a_{max}$ and $B_{a,s}\geq -a_{min}$). 

\begin{equation}
N_{a,sat}(B_{a,s}) = \frac{1}{\pi} \left( \frac{a_{\max}}{B_{a,s}} \sqrt{1 - \left(\frac{a_{\max}^2}{B_{a,s}^2}\right)} - \frac{a_{\min}}{B_{a,s}} \sqrt{1 - \left(\frac{a_{\min}^2}{B_{a,s}^2}\right)} +  \sin^{-1}\left(\frac{a_{\max}}{B_{a,s}}\right) - \sin^{-1}\left(\frac{a_{\min}}{B_{a,s}}\right) \right)
\end{equation}

\noindent correspondingly, by Eq. (\ref{incremental DF}):
\begin{equation}
\begin{aligned}
N_{a,inc}(B_{a,s},\theta_a)=&N_{a,sat}(B_{a,s})+\frac{B_{a,s}}{2}\cdot\frac{dN_{a,sat}(B_{a,s})}{dB_{a,s}}(1+e^{-j2\theta_a})\\
=&\frac{1}{\pi} \left( \frac{a_{\max}}{B_{a,s}} \sqrt{1 - \left(\frac{a_{\max}^2}{B_{a,s}^2}\right)} - \frac{a_{\min}}{B_{a,s}} \sqrt{1 - \left(\frac{a_{\min}^2}{B_{a,s}^2}\right)} +  \sin^{-1}\left(\frac{a_{\max}}{B_{a,s}}\right) - \sin^{-1}\left(\frac{a_{\min}}{B_{a,s}}\right) \right)\\
&+(\frac{-a_{max}\sqrt{B_{a,s}^2-a_{max}^2}+a_{min}\sqrt{B_{a,s}^2-a_{min}^2}}{B_{a,s}\pi})(1+e^{-j2\theta_a})\\
=&\frac{e^{-j2\theta_a}}{B_{a,s}\pi}(a_{max}\sqrt{1-\frac{a_{max}^2}{B_{a,s}^2}}-a_{min}\sqrt{1-\frac{a_{min}^2}{B_{a,s}^2}})+\frac{1}{\pi}(sin^{-1}(\frac{a_{max}}{B_{a,s}})-sin^{-1}(\frac{a_{min}}{B_{a,s}}))\label{control incremental DF4}
\end{aligned}
\end{equation}

\subsubsection{State saturation \label{sec5.2.2}}
Similarly, consider the predominant input to the state saturation $\tilde{v}_{n+1}(t)=B_{v,s} \sin(\omega t)$, where $B_{v,s}$ can be obtained by solving Eqs. (\ref{state triag1})-(\ref{state triag2}) or Eqs. (\ref{both triag1})-(\ref{both triag2}), and the output of the saturation limited within $[\tilde{v}_{min},\tilde{v}_{max}]$. The incremental-input DF of the state saturation, $N_{v,inc}(B_{v,s},\theta_v)$, can be obtained as follows:

\textbf{Case (1)}: limit inactive (i.e., $B_{v,s}\leq-\tilde{v}_{min}$ and $B_{v,s}\leq \tilde{v}_{max}$).

\begin{equation}
\begin{aligned}
N_{v,inc}(B_{v,s},\theta_v)=1\label{state incremental DF1}
\end{aligned} 
\end{equation}

\textbf{Case (2)}: limit activeness by $\tilde{v}_{min}$ (i.e., $-\tilde{v}_{min}\leq B_{v,s}\leq \tilde{v}_{max}$).

\begin{equation}
\begin{aligned}
N_{v,inc}(B_{v,s},\theta_v)=\frac{1}{2}+\frac{e^{-j2\theta_v}}{B_{v,s}\pi}\tilde v_{min}\sqrt{1-\frac{\tilde v_{min}^2}{B_{v,s}^2}}-\frac{1}{\pi}sin^{-1}(\frac{\tilde v_{min}}{B_{v,s}})\label{state incremental DF2}
\end{aligned}
\end{equation}

\textbf{Case (3)}: limit activeness by $\tilde{v}_{max}$ (i.e., $\tilde{v}_{max}\leq B_{v,s}\leq -\tilde{v}_{min}$).

\begin{equation}
\begin{aligned}
N_{v,inc}(B_{v,s},\theta_v)=\frac{1}{2}-\frac{e^{-j2\theta_v}}{B_{v,s}\pi}\tilde v_{max}\sqrt{1-\frac{\tilde v_{max}^2}{B_{v,s}^2}}+\frac{1}{\pi}sin^{-1}(\frac{\tilde v_{max}}{B_{v,s}})\label{state incremental DF3}
\end{aligned}
\end{equation}

\textbf{Case (4)}: limit activeness by both $\tilde{v}_{max}$ and $\tilde{v}_{min}$ (i.e., $B_{v,s}\geq \tilde{v}_{max}$ and $B_{v,s}\geq -\tilde{v}_{min}$).

\begin{equation}
\begin{aligned}
N_{v,inc}(B_{v,s},\theta_v)=\frac{e^{-j2\theta_v}}{B_{v,s}\pi}(\tilde v_{max}\sqrt{1-\frac{\tilde v_{max}^2}{B_{v,s}^2}}-\tilde v_{min}\sqrt{1-\frac{\tilde v_{min}^2}{B_{v,s}^2}})+\frac{1}{\pi}(sin^{-1}(\frac{\tilde v_{max}}{B_{v,s}})-sin^{-1}(\frac{\tilde v_{min}}{B_{v,s}}))\label{state incremental DF4}
\end{aligned}
\end{equation}

\subsection{Oscillation stability analysis based on the incremental-input DF\label{sec5.3}}
With the incremental-input DFs derived, the stability of each oscillation candidate obtained by the approach in Subsection \ref{sec4.2} can be analyzed by viewing the incremental-input as a perturbation around the potential oscillation. We first present a representation of the AV system that's equivalent to the one shown in Fig. \ref{fig:DF block diagram}, as showcased in Fig. \ref{fig:equivalent system diag}. Notably, this representation diverges in the way the feedback term of $\dot{\tilde{p}}_{n+1}$ in Eq. (\ref{eq:AV model3}) is split into $k_3\dot{\tilde{p}}_{n+1}=k_2\dot{\tilde{p}}_{n+1}+(k_3-k_2)\dot{\tilde{p}}_{n+1}$.
\begin{figure}[h]
    \centering
    \setlength{\abovecaptionskip}{0pt}
    \includegraphics[width=0.99\textwidth]{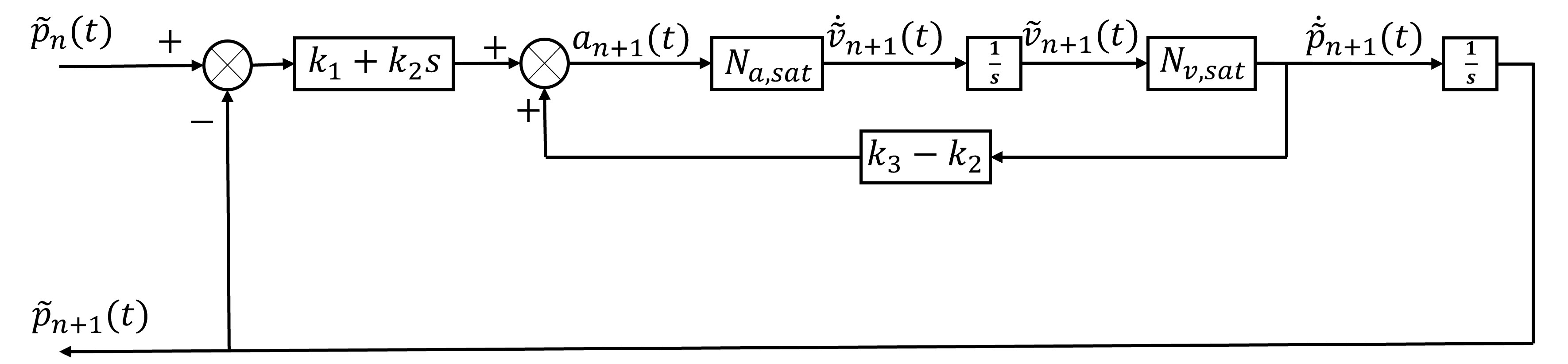}
    \caption{Block diagram equivalent to Fig. \ref{fig:DF block diagram}}
    \label{fig:equivalent system diag}
\end{figure}

Based on Fig. \ref{fig:equivalent system diag}, the incremental system that the perturbation passes through is depicted in Fig. \ref{fig:incremental-input DF block diagram}. Within this system, the transfer function of elements inside the dashed box is denoted by $H(s)$. By scrutinizing this incremental system, a potential oscillation of Fig. \ref{fig:equivalent system diag} is stable if a perturbation around it decays when transiting through Fig. \ref{fig:incremental-input DF block diagram}. Leveraging the principles of the Nyquist method  (\citep{west1956dual}) in control theory, the decay of the perturbation is checked by examining its open-loop counterpart. Stability is achieved when the open-loop incremental frequency response locus avoids enclosing the coordinate (-1,0) in the complex plane (A comprehensive exposition of closed-loop/open-loop, frequency response, the Nyquist method, and locus can be found in \citep{franklin2002feedback}). 
\begin{figure}[h]
    \centering
    \setlength{\abovecaptionskip}{0pt}
    \includegraphics[width=0.9\textwidth]{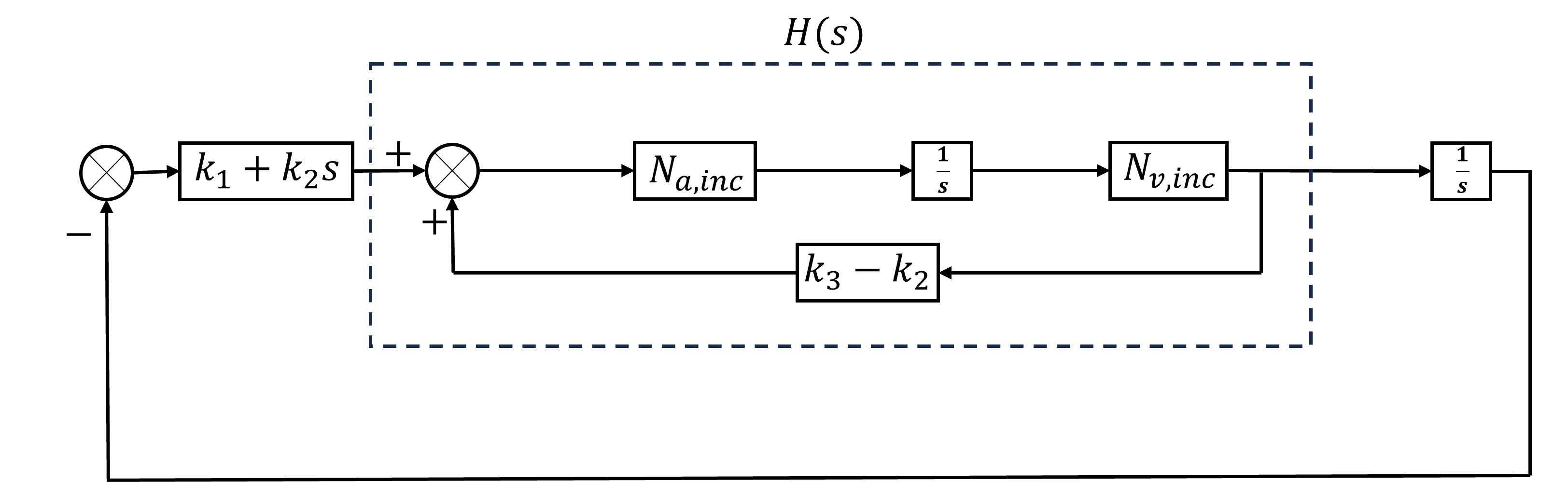}
    \caption{Incremental system corresponding to the AV system}
    \label{fig:incremental-input DF block diagram}
\end{figure}

In an AV system featuring only the control saturation, the corresponding transfer function $H(s)$ can be obtained by removing $N_{v,inc}$ from Fig. \ref{fig:incremental-input DF block diagram}:  

\begin{equation}
\begin{aligned}
H(s)=\frac{N_{a,inc}(B_{a,s},\theta_a)}{s-(k_3-k_2)N_{a,inc}(B_{a,s},\theta_a)} \label{incremental characteristic control}
\end{aligned}
\end{equation}

In contrast, for an AV system with just the state saturation, the transfer function $H(s)$ can be obtained by excluding $N_{a,inc}$ from Fig. \ref{fig:incremental-input DF block diagram}is:  

\begin{equation}
\begin{aligned}
H(s)=\frac{N_{v,inc}(B_{v,s},\theta_v)}{s-(k_3-k_2)N_{v,inc}(B_{v,s},\theta_v)} \label{incremental characteristic state}
\end{aligned}
\end{equation}

As for AV systems containing both control and state saturation, $H(s)$ becomes:  

\begin{equation}
\begin{aligned}
H(s)=\frac{N_{a,inc}(B_{a,s},\theta_a)\cdot N_{v,inc}(B_{v,s},\theta_v)}{s-(k_3-k_2)N_{v,inc}(B_{a,s},\theta_a)\cdot N_{v,inc}(B_{v,s},\theta_v)} \label{incremental characteristic both}
\end{aligned}
\end{equation}

\noindent where $B_{v,s}=\frac{B_{a,s}|N_{a,sat}(B_{a,s})|}{\omega}$ as detailed in Subsection \ref{sec4.2.3}, and $\theta_v=\theta_a+\angle N_{a,inc}(B_{a,s},\theta_a)-\frac{\pi}{2}$ wherein $\angle N_{a,inc}(B_{a,s},\theta_a)=tan^{-1}(\frac{Im[N_{a,inc}(B_{a,s},\theta_a)]}{Re[N_{a,inc}(B_{a,s},\theta_a)]})$. 

With $H(s)$, the open-loop transfer function of the incremental system is obtained as follows: 
\begin{equation}
    T_o(s)=(k_1+k_2s)\cdot H(s)\cdot\frac{1}{s}
\end{equation}
The open-loop incremental frequency response is obtained by substituting $s=j\omega$ into $T_o(s)$. The locus can be plotted by varying $\theta_a$ from 0 to $2\pi$ for systems with $H(s)$ defined by Eqs. (\ref{incremental characteristic control}) and (\ref{incremental characteristic both}), or varying $\theta_v$ from 0 to $2\pi$ for systems with $H(s)$ defined by Eq. (\ref{incremental characteristic state}). It has been shown that the locus of an incremental-input DF traces a circle in the complex plane (\citep{west1956dual,bonenn1958stability}), with linear transformations merely affecting its center and radius. Figure \ref{fig:locus example} illustrates loci that indicate stable oscillation candidates (where the open-loop incremental frequency response locus does not encircle (-1,0)) and unstable oscillation candidates (where it does encircle (-1,0)).

\begin{figure}[h]
    \centering
    \setlength{\abovecaptionskip}{0pt}
    \subcaptionbox{}{\includegraphics[width=0.4\textwidth]{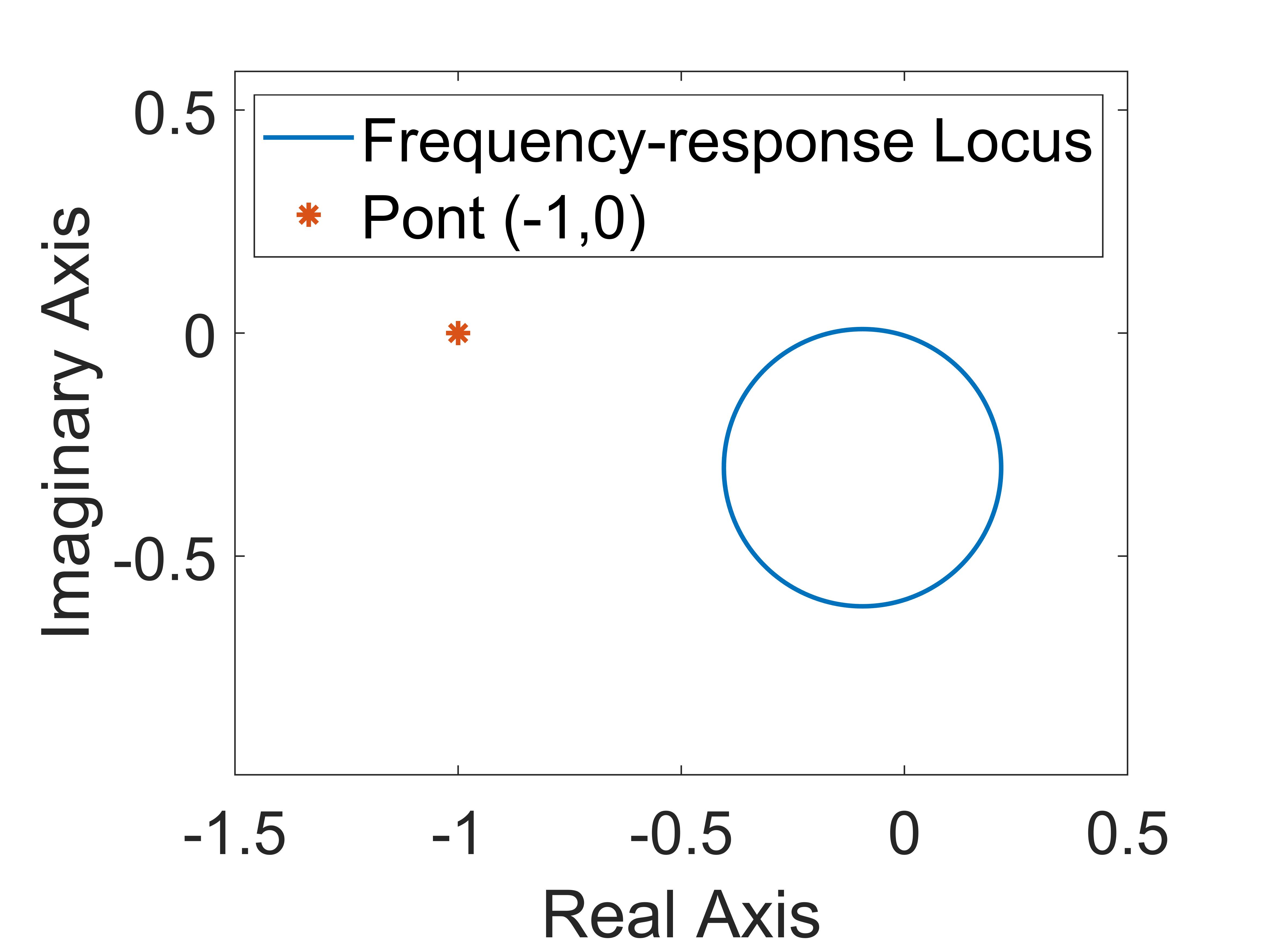}}
    \subcaptionbox{}{\includegraphics[width=0.4\textwidth]{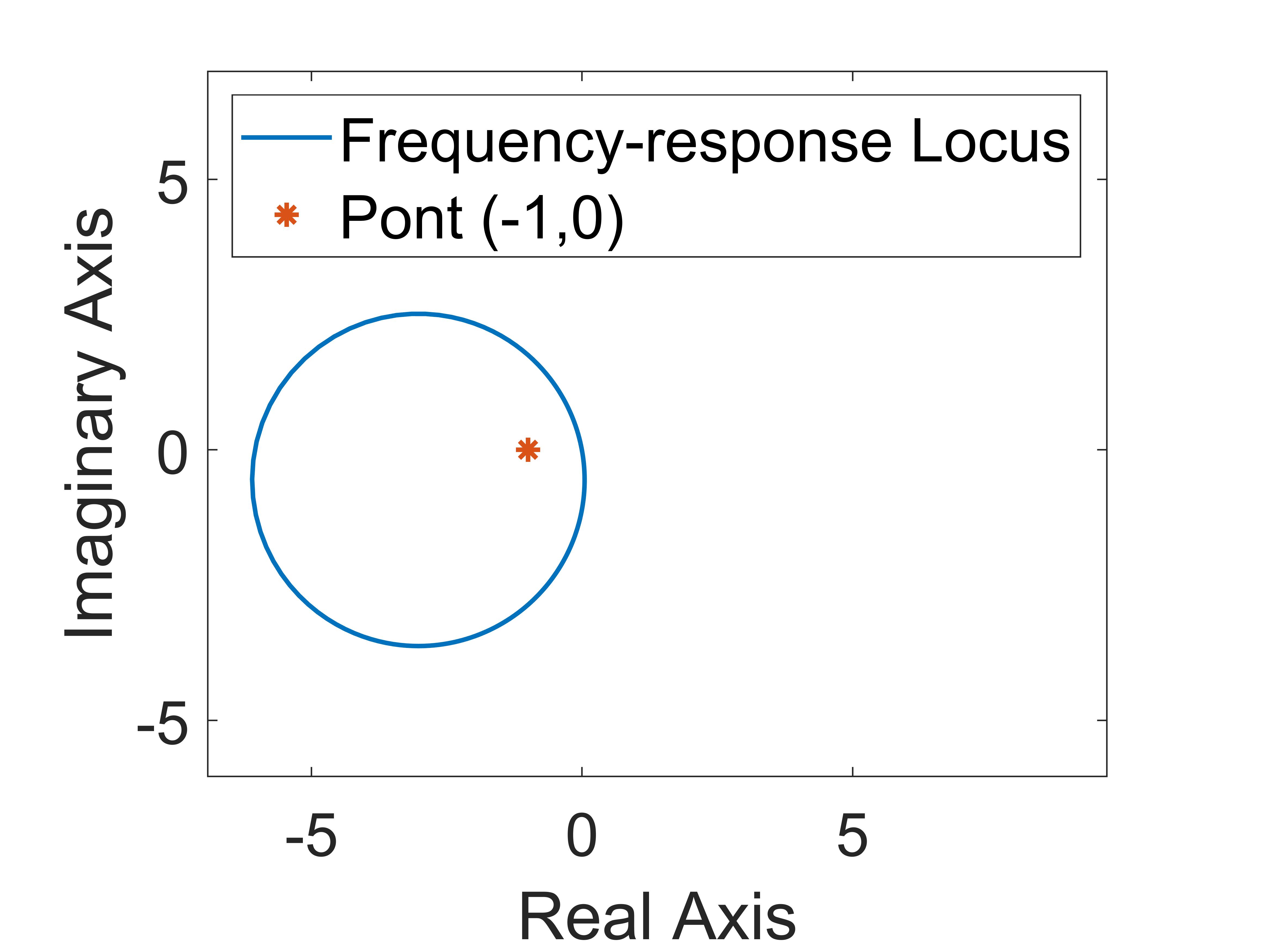}}
    \caption{Examples of open-loop incremental frequency response locus: (a) stable closed-loop; (b) unstable closed-loop}
    \label{fig:locus example}
\end{figure}

Thus far, the theoretical deductions of the entire analysis procedure in Fig. \ref{fig:steps} have been presented. 

\section{Simulation and results\label{sec6}}

To demonstrate the effectiveness and accuracy of the proposed theoretical analysis on the frequency response of AV in traffic oscillations, we validate the frequency responses calculated based on the proposed method with the frequency responses estimated from a Simulink AV model characterized by Eqs. (\ref{eq:AV model1})-(\ref{eq:AV model3}). The results are also compared with other methodologies. Specifically, Subsection \ref{sec6.1} compares the theoretical results of the proposed method with state-of-the-art methods for AV systems containing one saturation. Subsection \ref{sec6.3} extends the comparison to AV systems with both traffic state and control input saturation. Furthermore, Subsection \ref{sec6.new} provides a comprehensive analysis of how the nonlinear frequency response varies as the predecessor's oscillation changes. Finally, Subsection \ref{sec6.4} discusses the usage of the proposed method in cases where conventional string stability analysis by the linear method gives misleading conclusions.  

A linear feedback controller employing a constant time gap spacing policy $d_{n+1}^*=l_{n+1}+\tau \dot{\tilde{p}}_{n+1}$ and control law $a_{n+1}(t)=k_d\Delta d_{n+1}(t)+k_v\Delta v_{n+1}(t)$ is considered. Unless otherwise specified, the default parameters are selected as follows. $\tau=1s$ represents the pre-defined time gap and $k_d=1s^{-2}$, $k_v=2s^{-1}$ act as the feedback gains. Correspondingly, it can be calculated that $k_1=k_d=1s^{-2}$, $k_2=k_v=2s^{-1}$, and $k_3=-k_v-k_d\tau=-3s^{-1}$. Similar to \citep{li2010measurement}, we apply the Fourier transform (\cite{frigo1998fftw}) to trajectories extracted from the NGSIM dataset to obtain the frequency spectrum inherent to the oscillatory trajectories. It was found that the frequencies predominantly fall within the range of $[0, 0.1]$ Hz. Therefore, we consider that investigating the frequency responses within the range of 0 to 0.5 Hz is sufficient to cover the entire frequency spectrum.
{\begin{rem}
     By the two assumptions stated in Subsection \ref{sec3.1}, and considering AVs' capacity to reduce rather than create oscillations (\citep{cui2017stabilizing}), it is appropriate to use HDV-generated oscillations (from NGSIM data) to establish the relevant frequency range.
\end{rem}}

\subsection{Simulation of system containing one saturation\label{sec6.1}}
We first investigate AV systems integrated with a single control saturation, which is equivalent to letting $\tilde{v}_{max}=\infty$ and $\tilde{v}_{min}=-\infty$ in Eq. (\ref{eq:AV model1}) or removing Saturation 2 from Fig. \ref{fig:block diagram}. For this type of systems, the proposed method (referred as the IDF method hereafter) characterizes the control saturation with Eqs. (\ref{eq:a DF 1}), (\ref{control DF2}), (\ref{control DF3}), and (\ref{control DF4}); obtains the amplitude and phase of oscillation candidates by solving Eqs. (\ref{control triag1})-(\ref{control triag2}); analyzes the stability of each oscillation candidate with the approach proposed in Subsection \ref{sec5.3}, with $H(s)$ represented by Eq. (\ref{incremental characteristic control}) and $N_{a,inc}$ characterized by Eqs. (\ref{control incremental DF1}), (\ref{control incremental DF2}), (\ref{control incremental DF3}), and (\ref{control incremental DF4}); and calculates the magnitude and phase of the frequency response using Eqs. (\ref{control FR mag}) and (\ref{control FR phase}), respectively. The acceleration limits are chosen as $a_{max}=-a_{min}=a_{bound}=5m/s^2$. Since the impact of the saturation essentially hinges on the correlation between the amplitude of the leader trajectory and the saturation limit, we investigate instances with varied ratios of $R/a_{bound}$. We compare the magnitude and phase of frequency response as determined by the IDF method, the linear theoretical approach, and the approximation technique presented in \citep{zhou2023data} (referred to as the DF approximation hereafter, note that it is different from the DF method presented in this paper). The Simulink estimations are used for validation. These comparisons are visually represented in Fig. \ref{fig:control FR}.

\begin{figure}[h]
    \centering
    \setlength{\abovecaptionskip}{0pt}
    \subcaptionbox{}{\includegraphics[width=0.4\textwidth]{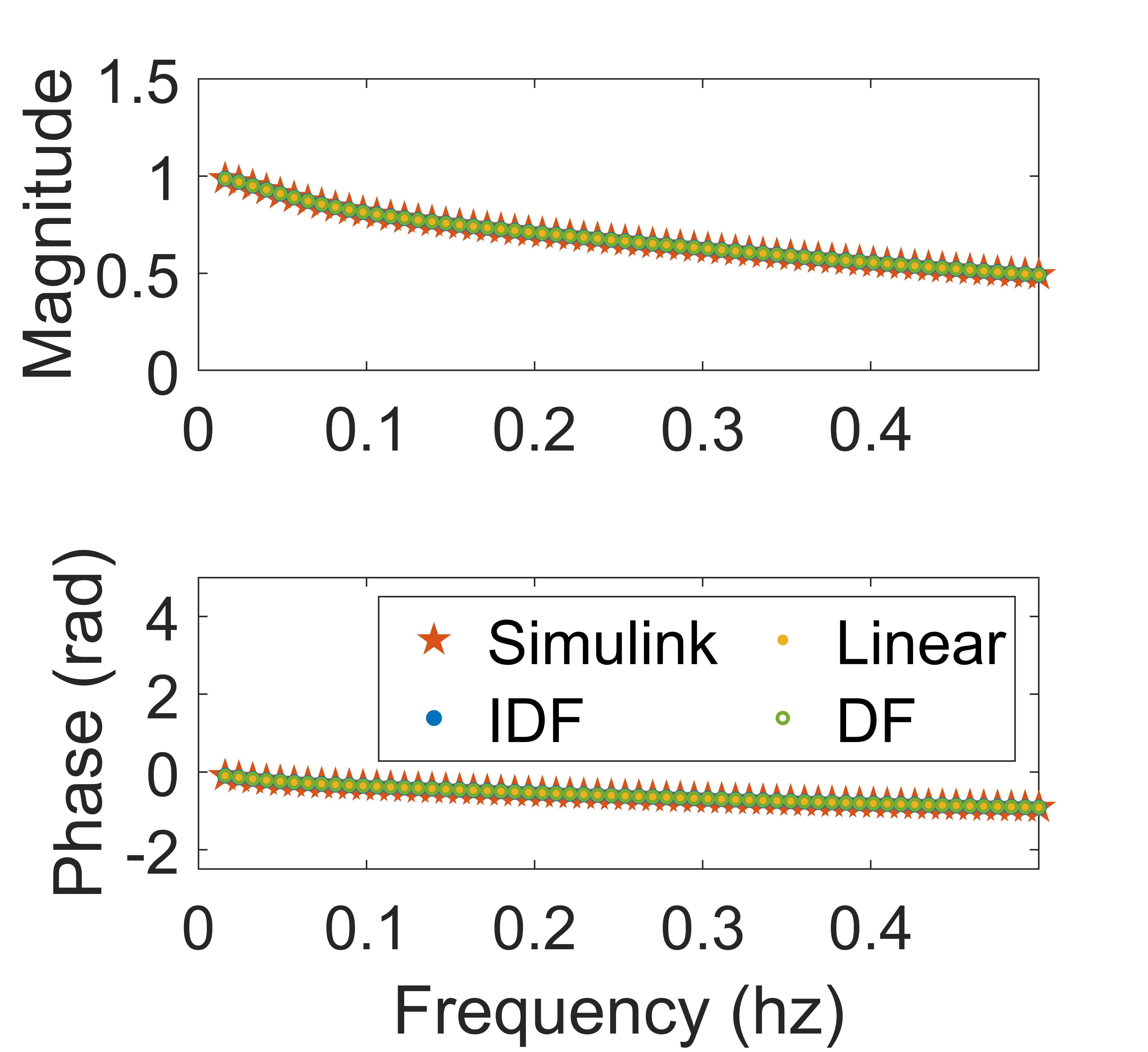}}
    \subcaptionbox{}{\includegraphics[width=0.4\textwidth]{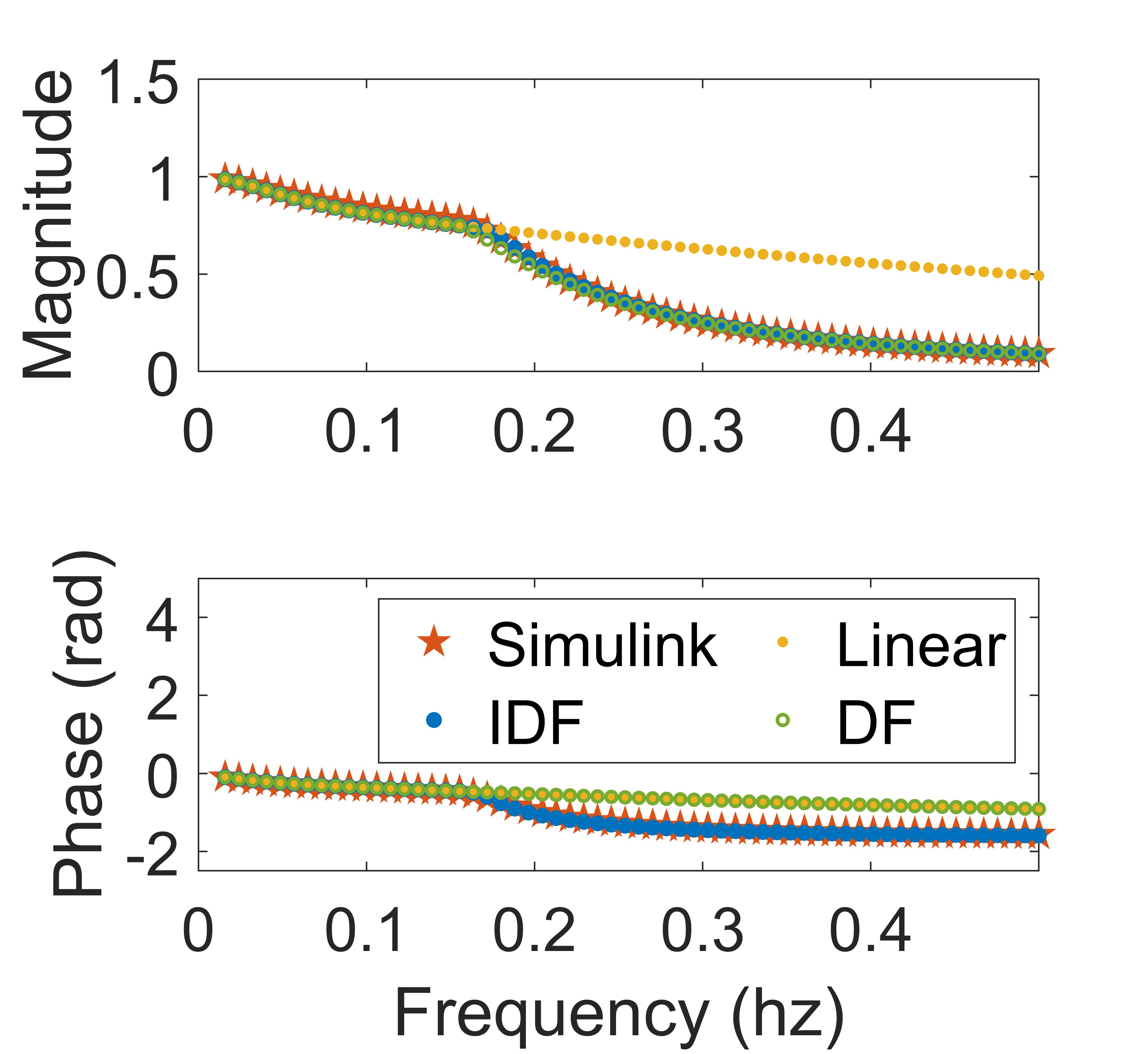}}
    \subcaptionbox{}{\includegraphics[width=0.4\textwidth]{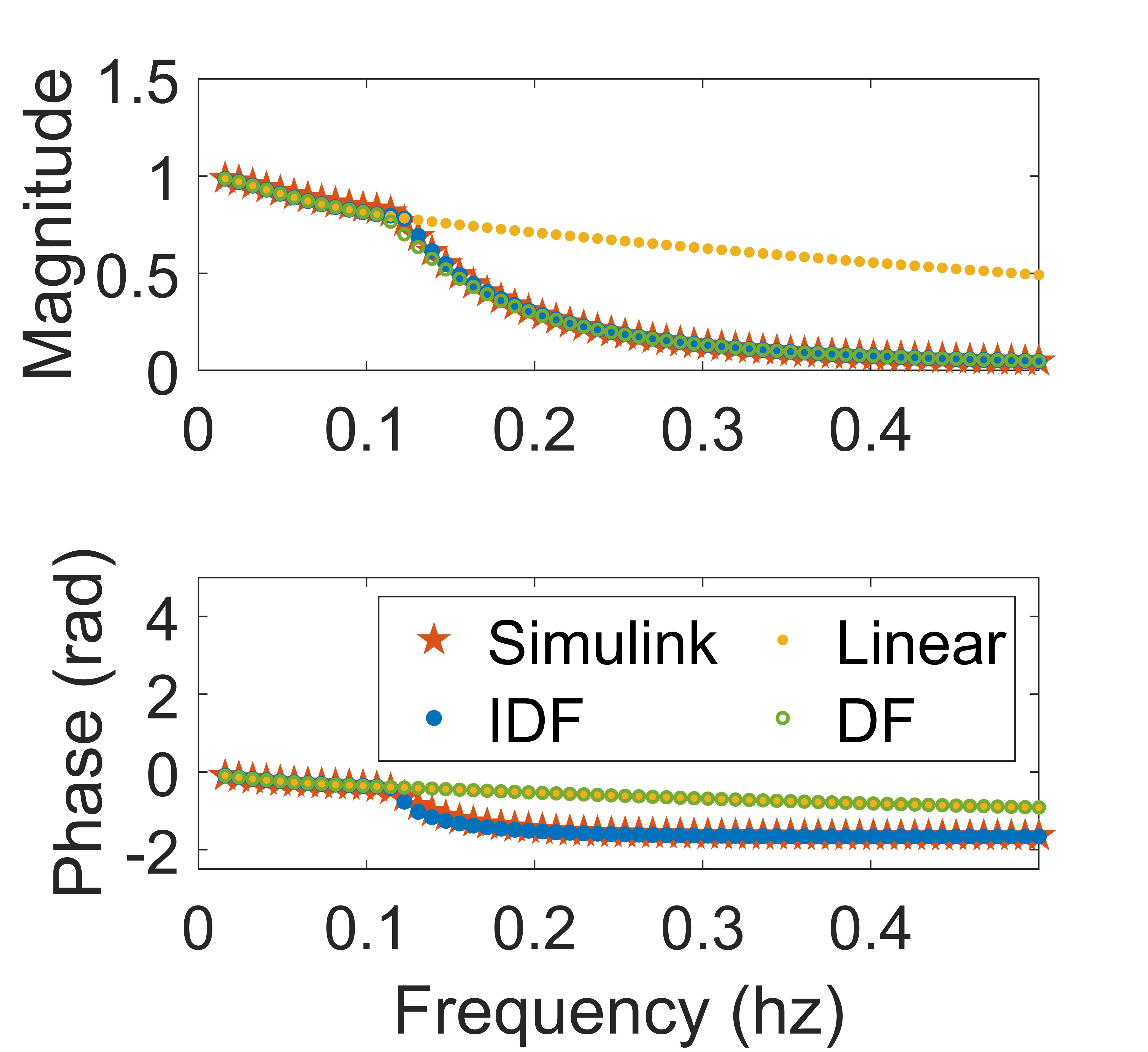}}
    \subcaptionbox{}{\includegraphics[width=0.4\textwidth]{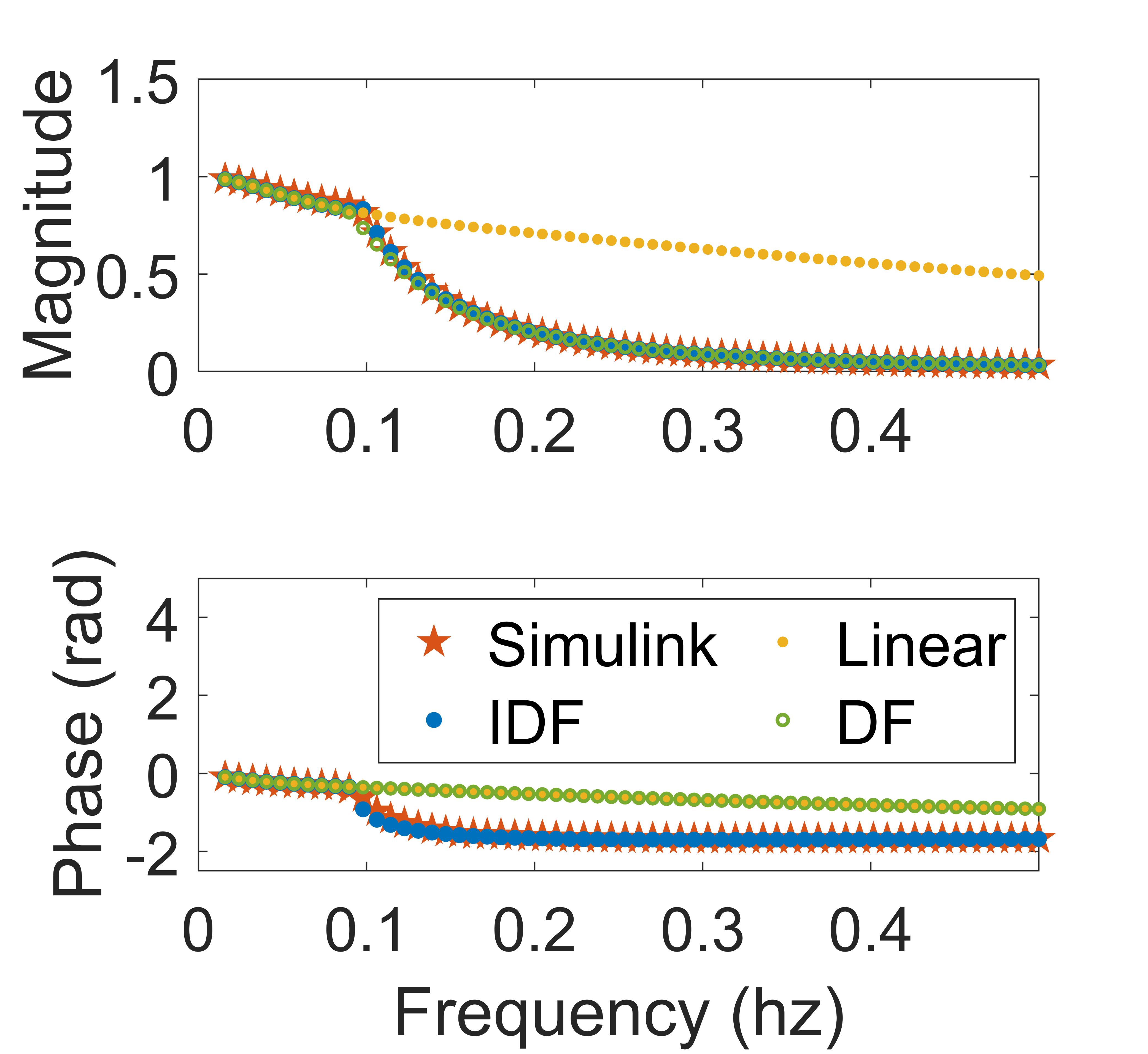}}
    \caption{Comparison of frequency responses of systems containing solely control saturation derived by different methods: (a) $R/a_{bound}=0.1$; (b) $R/a_{bound}=1.4$; (c) $R/a_{bound}=2.7$; (d) $R/a_{bound}=4$}
    \label{fig:control FR}
\end{figure}

Similarly, for AV systems integrated with a singular state saturation, the system is equivalent to letting ${a}_{max}=\infty$ and ${a}_{min}=-\infty$ in Eq. (\ref{eq:AV model2}) or removing Saturation 1 from Fig. \ref{fig:block diagram}. The IDF method characterizes the state saturation with Eqs. (\ref{state DF1}), (\ref{state DF2}), (\ref{state DF3}), and (\ref{state DF4}); derives the amplitude and phase of oscillation candidates by solving Eqs. (\ref{state triag1})-(\ref{state triag2}); analyzes the stability of each oscillation candidate with the approach proposed in Subsection \ref{sec5.3} with $H(s)$ represented by Eq. (\ref{incremental characteristic state}) and $N_{v,inc}$ characterized by Eqs. (\ref{state incremental DF1}), (\ref{state incremental DF2}), (\ref{state incremental DF3}), and (\ref{state incremental DF4}); and calculates the magnitude and phase of the frequency response using Eqs. (\ref{state FR mag}) and (\ref{state FR phase}), respectively. The velocity limits are set with an upper limit at free-flow speed $v_{max}=20m/s$ and a lower limit of $v_{min}=0m/s$. With the equilibrium speed $v_e=10m/s$ chosen, it can be readily obtained that $\tilde{v}_{max}=-\tilde{v}_{min}=10m/s$. Scenarios across various $R/\tilde{v}_{bound}$ ratios are investigated. The comparison of frequency response is illustrated in Fig. \ref{fig:state FR}.

\begin{figure}[h]
    \centering
    \setlength{\abovecaptionskip}{0pt}
    \subcaptionbox{}{\includegraphics[width=0.4\textwidth]{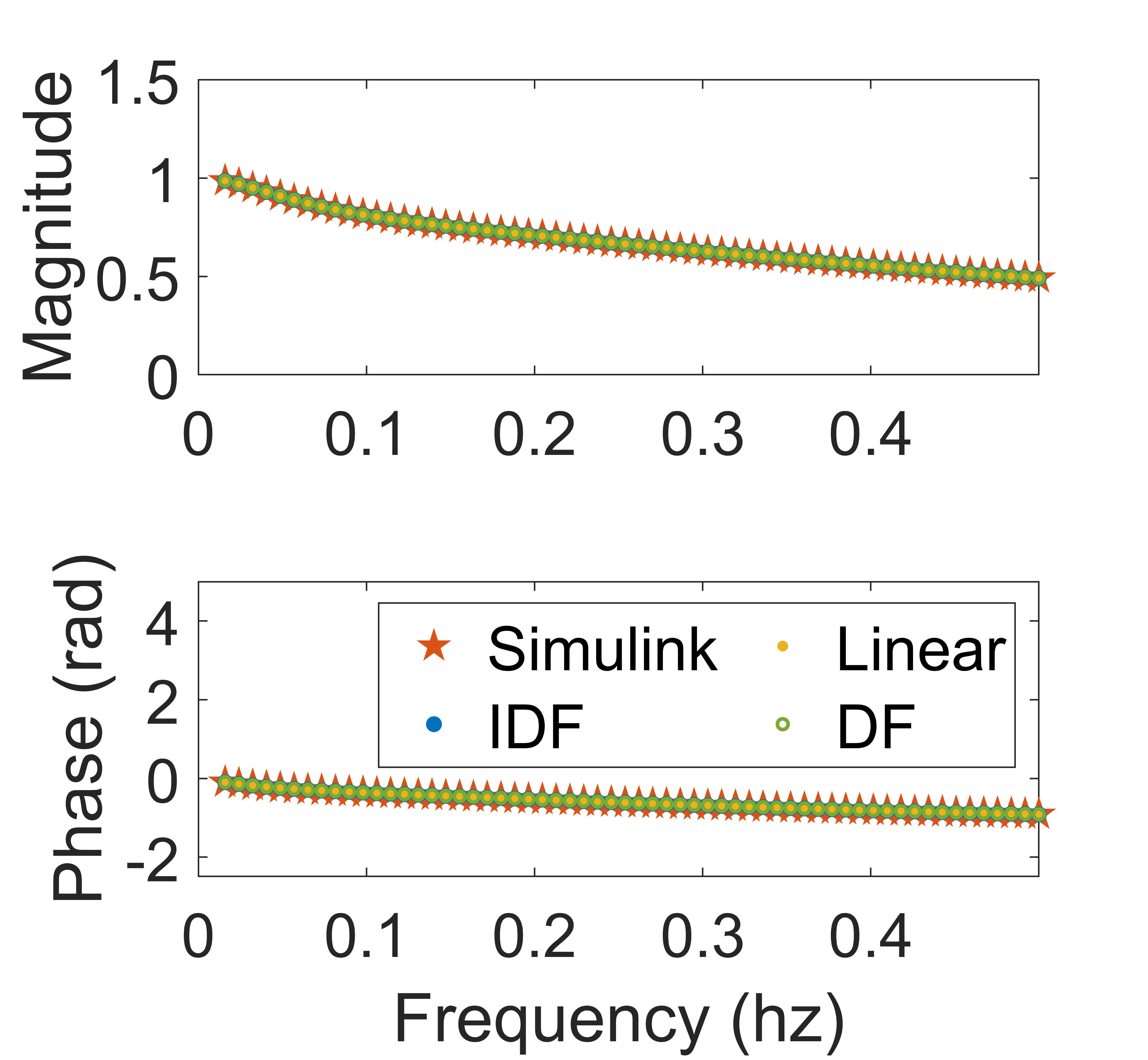}}
    \subcaptionbox{}{\includegraphics[width=0.4\textwidth]{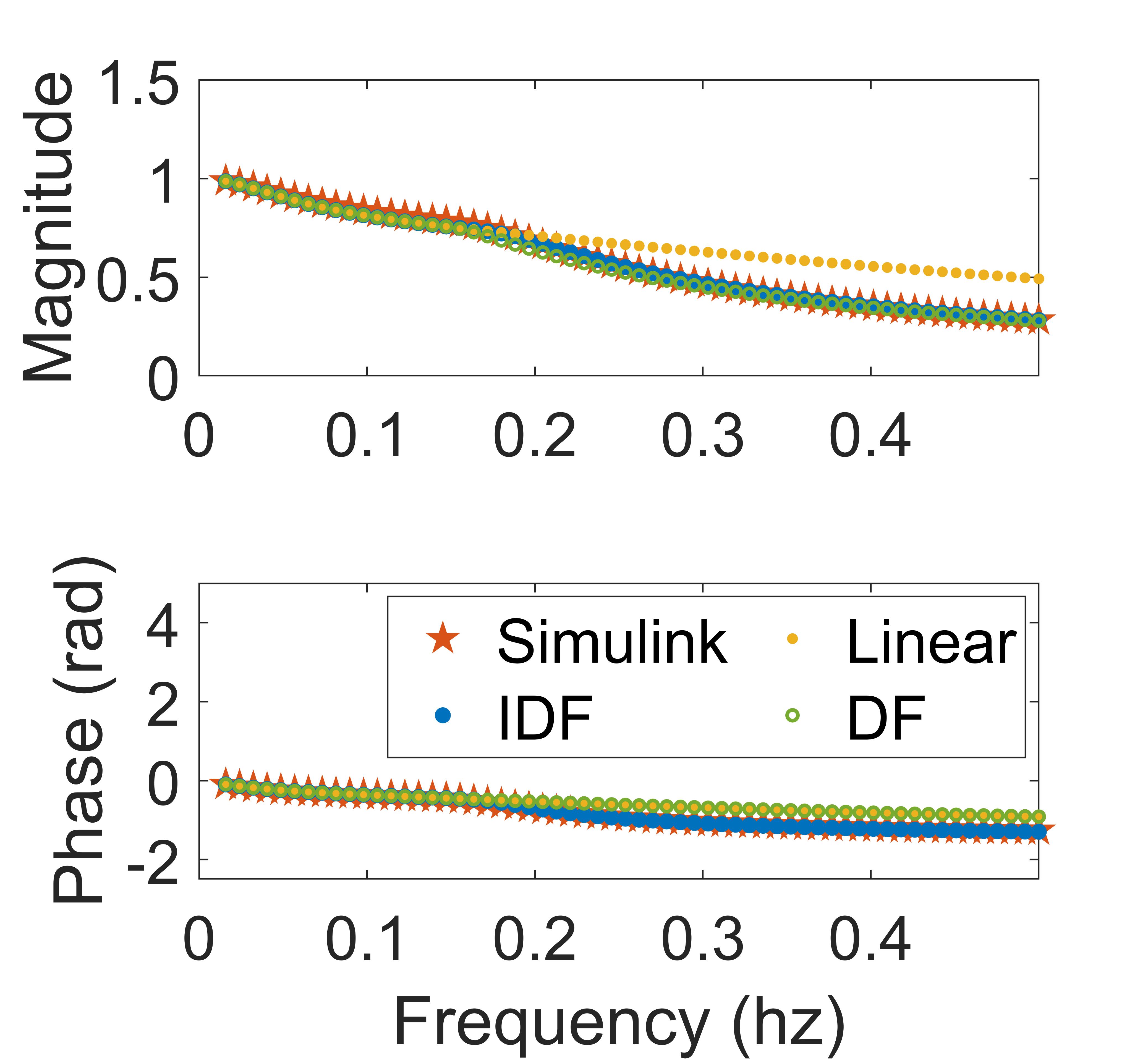}}
    \subcaptionbox{}{\includegraphics[width=0.4\textwidth]{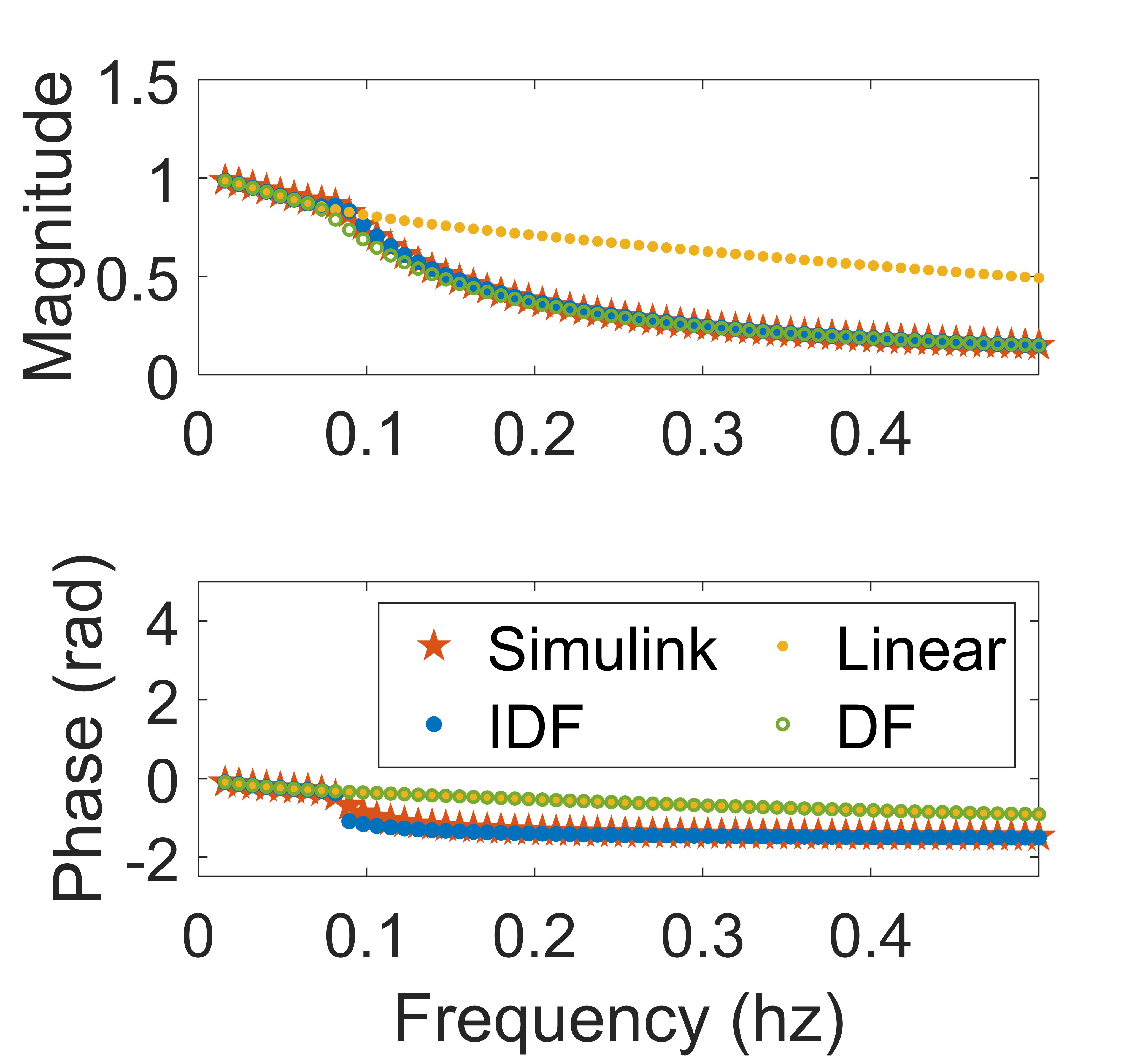}}
    \subcaptionbox{}{\includegraphics[width=0.4\textwidth]{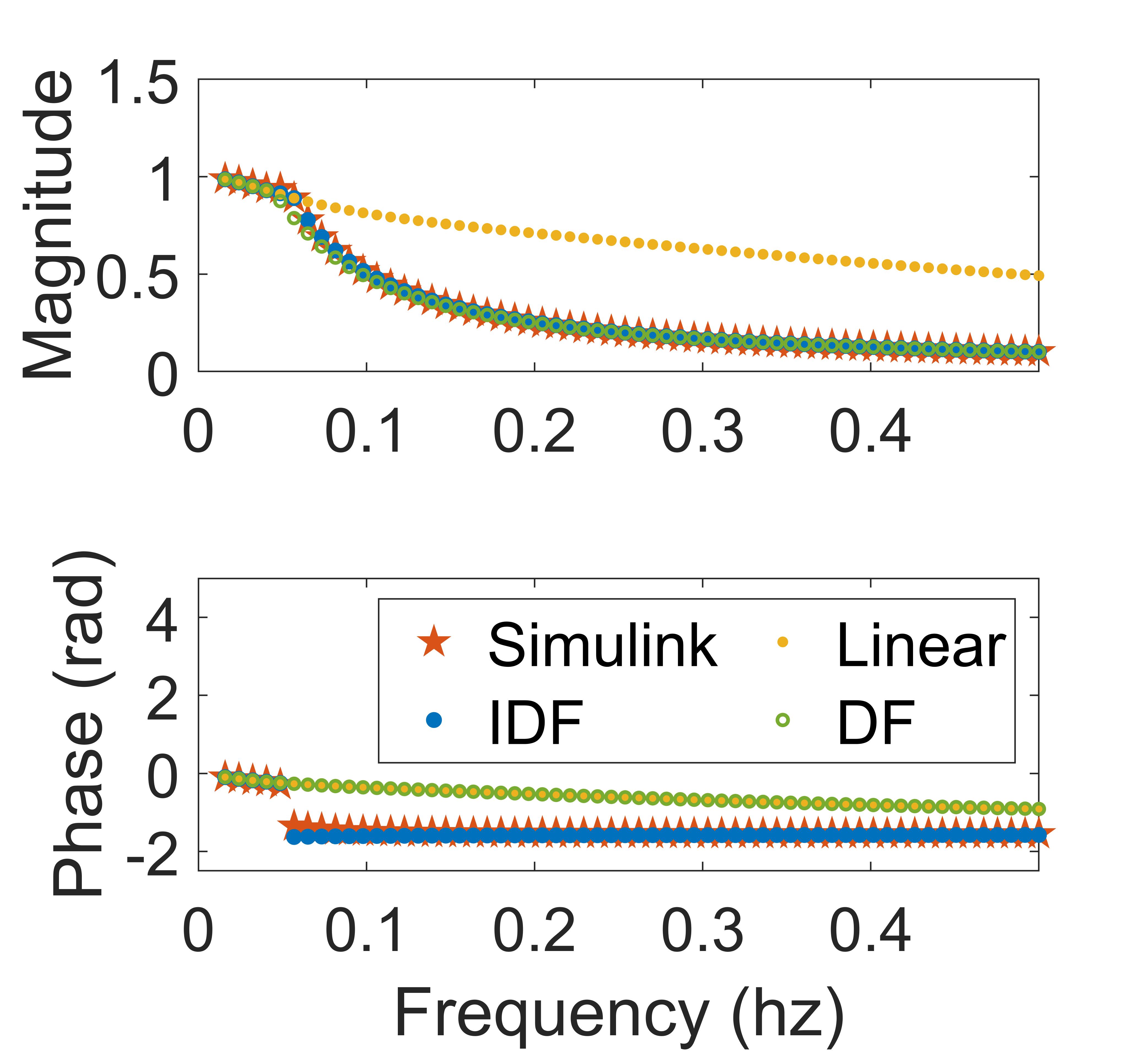}}
    \caption{Comparison of frequency responses of systems containing solely state saturation derived by different methods: (a) $R/\tilde{v}_{bound}=0.1$; (b) $R/\tilde{v}_{bound}=1.4$; (c) $R/\tilde{v}_{bound}=2.7$; (d) $R/\tilde{v}_{bound}=4$}
    \label{fig:state FR}
\end{figure}

From Figs. \ref{fig:control FR}(a) and \ref{fig:state FR}(a), it is evident that for smaller values of $R/a_{bound}$ and $R/v_{bound}$, the system operates within its linear regime as the saturation limits are not breached. Consequently, all three analytical approaches yield identical frequency responses, matching perfectly with the Simulink estimation. As we transition to Figs. \ref{fig:control FR}(b)-(d) and \ref{fig:state FR}(b)-(d), it becomes clear that with increasing $R/a_{bound}$ and $R/v_{bound}$ values, indicating more pronounced saturation effects, there's a noticeable decline in the magnitude of the frequency response. Simultaneously, there's a more negative shift in the phase as compared to the case where limits are not reached. The IDF method's predictions align almost seamlessly with the Simulink estimations, reaffirming its accuracy. The DF approximation presented by \citep{zhou2023data}, while slightly underestimating the magnitude, commendably traces the magnitude's evolution, especially for frequencies exceeding 0.25Hz. However, its approximation falls short of capturing the phase shifts due to saturation. This is because this method essentially models the system by positioning the non-linear component outside the feedback loop and doesn't take the harmonic balance into account. Meanwhile, the linear theoretical approach, as seen across Figs. \ref{fig:control FR}(a)-(d) and \ref{fig:state FR}(a)-(d), remains the same and is unable to account for the system's nonlinearity caused by saturation.

The results demonstrate that upon reaching the saturation limits, the linear theoretical method, which is commonly employed in much of the current research on string stability, significantly overestimates the magnitude of the frequency response, $|F|$. In other words, an underestimation of the attenuation effect in traffic oscillation largely exists for linear analysis. On the other hand, the DF approximation introduced by \citep{zhou2023data} marginally overestimates the oscillation attenuation. Both the linear method and DF approximation considerably overestimate the phase of the frequency response, $\angle F$. Consequently, the response time (defined as the duration before the AV responds to a predecessor's driving behavior change), which can be calculated by $-\frac{\angle F}{\omega}$, is underestimated.

\subsection{Simulation of the system containing both control and state saturation\label{sec6.3}}

\begin{figure}[t]
    \centering
    \setlength{\abovecaptionskip}{0pt}
    \subcaptionbox{}{\includegraphics[width=0.4\textwidth]{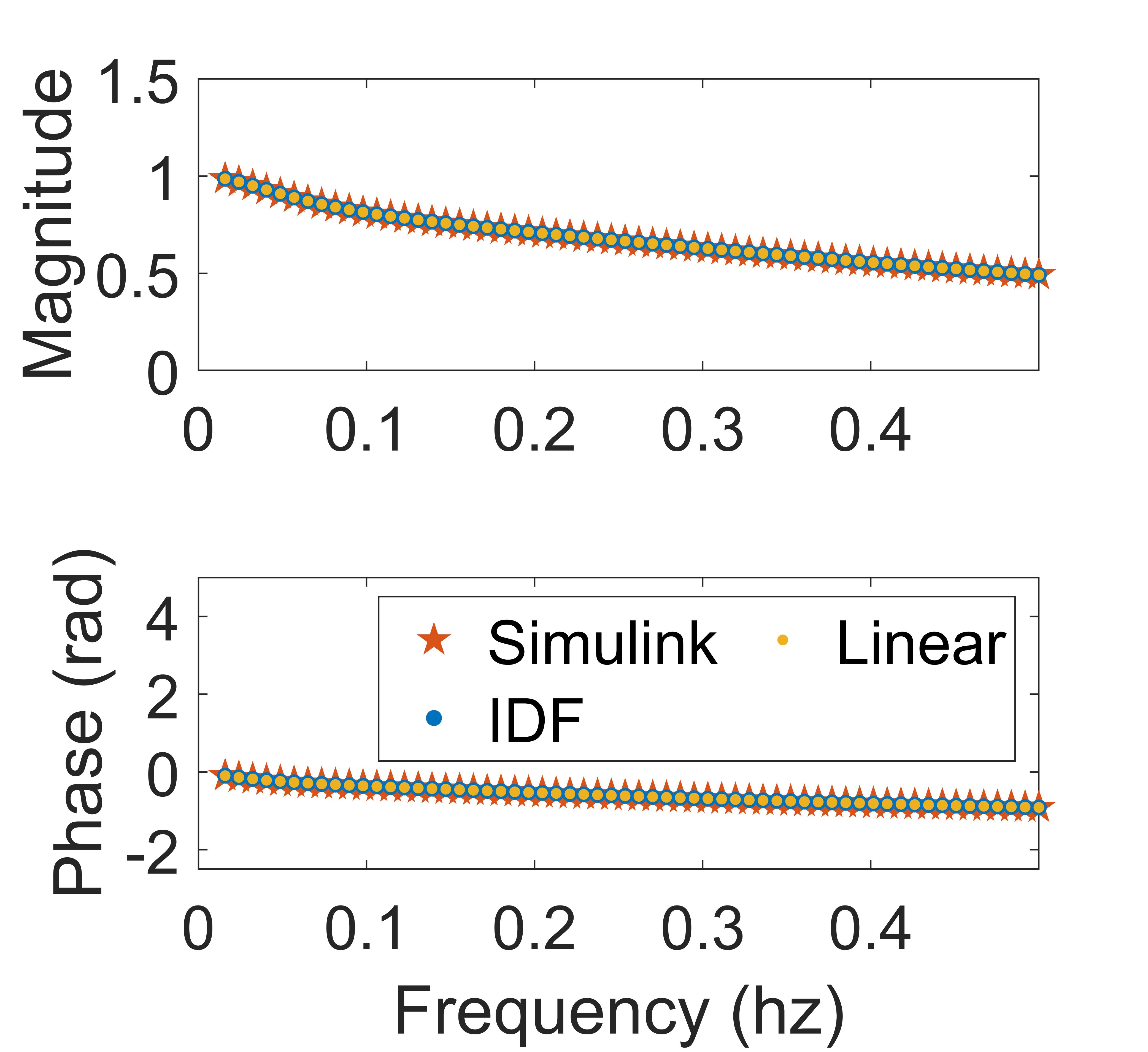}}
    \subcaptionbox{}{\includegraphics[width=0.4\textwidth]{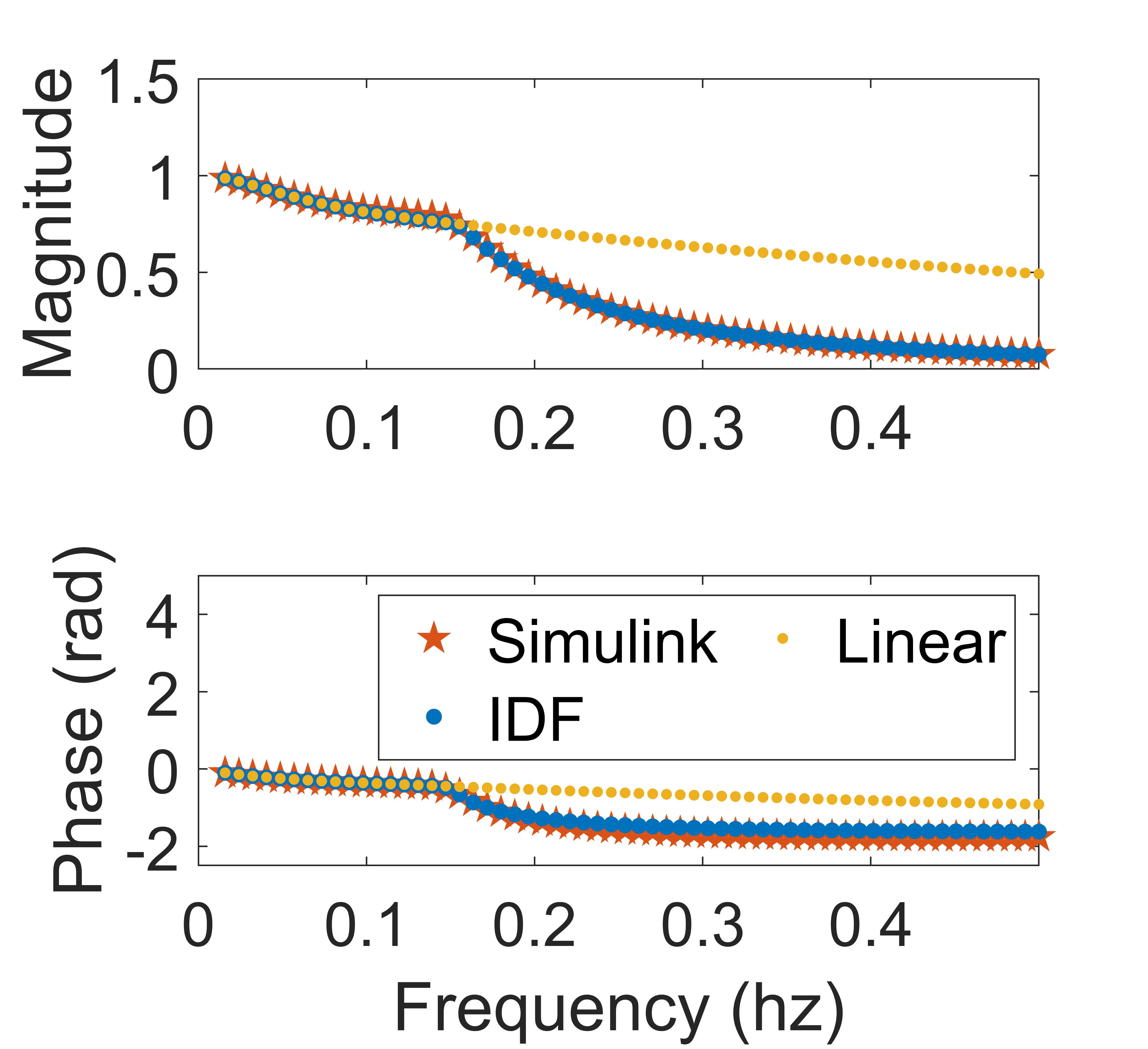}}
    \subcaptionbox{}{\includegraphics[width=0.4\textwidth]{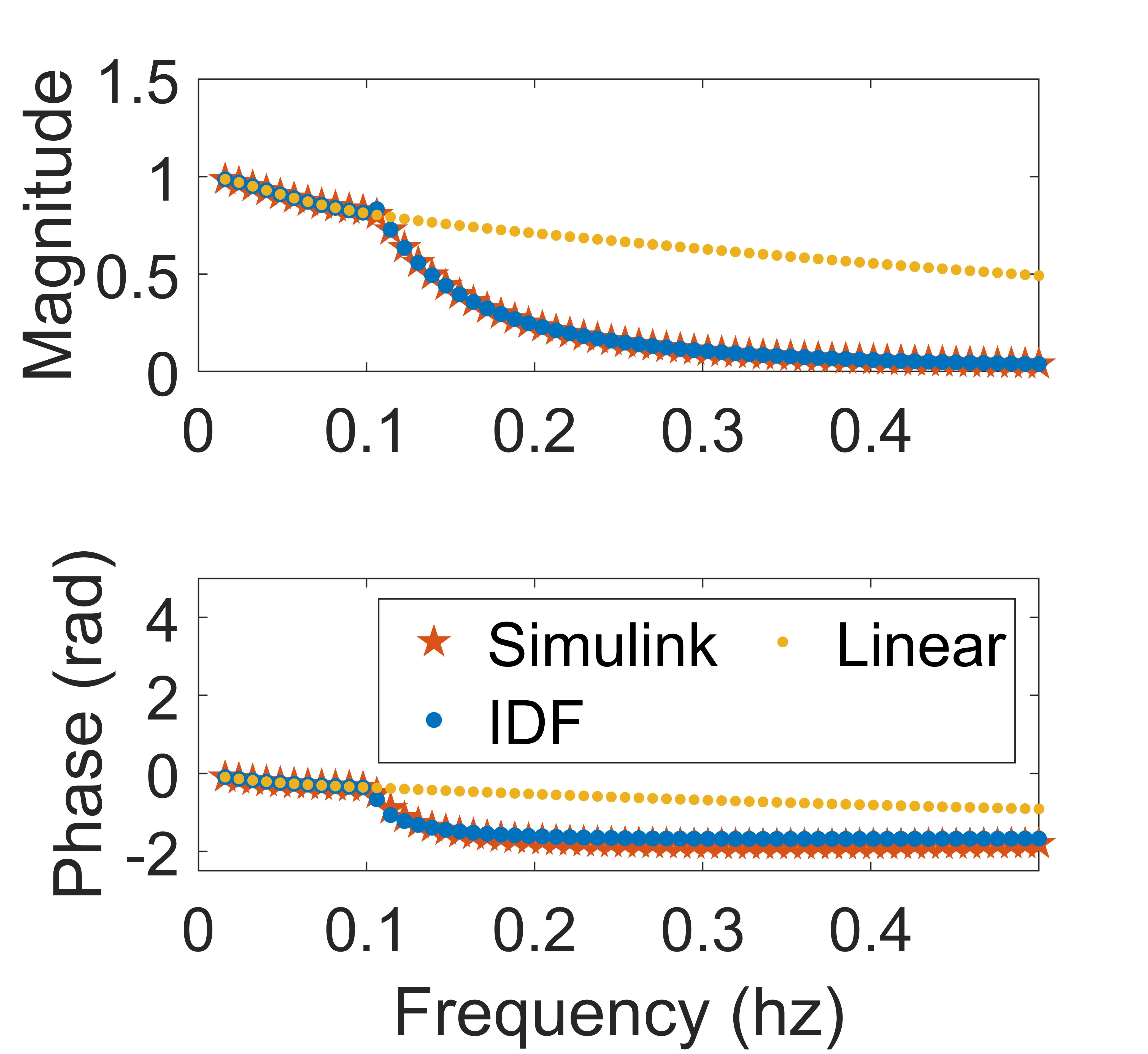}}
    \subcaptionbox{}{\includegraphics[width=0.4\textwidth]{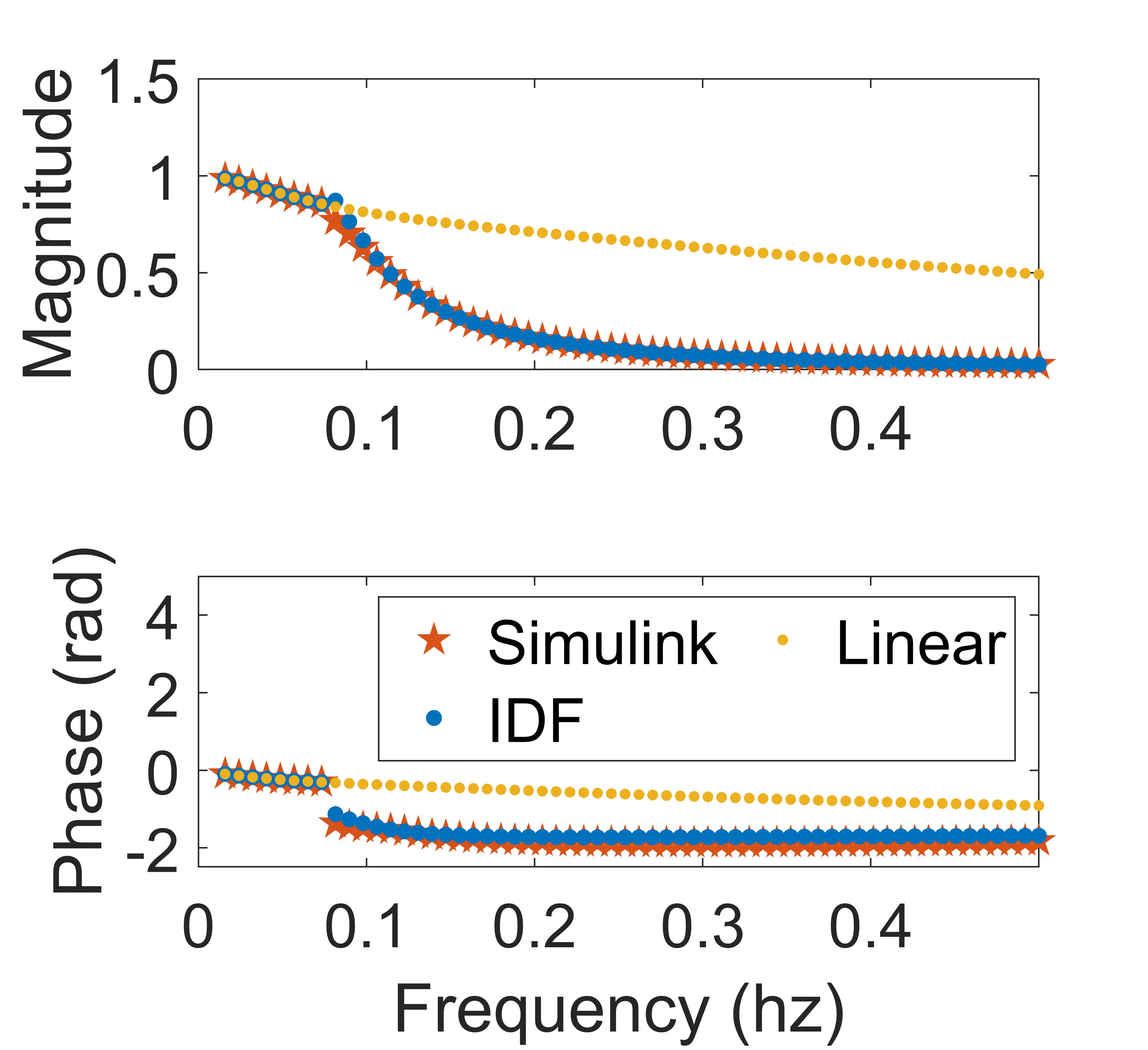}}
    \caption{Comparison of frequency responses of systems containing both control and state saturation derived by different methods: (a) $R/a_{bound}=0.1, R/\tilde{v}_{bound}=0.05$; (b) $R/a_{bound}=1.4, R/\tilde{v}_{bound}=0.7$; (c) $R/a_{bound}=2.7, R/\tilde{v}_{bound}=1.35$; (d) $R/a_{bound}=4, R/\tilde{v}_{bound}=2$}
    \label{fig:both FR}
\end{figure}

For AV systems integrated with both control and state saturation, which is exactly represented by Eqs. (\ref{eq:AV model1})-(\ref{eq:AV model3}) and Fig. \ref{fig:block diagram}. For this type of systems, the IDF method characterizes control saturation with Eqs. (\ref{eq:a DF 1}), (\ref{control DF2}), (\ref{control DF3}) and (\ref{control DF4}), and the state saturation with Eqs. (\ref{state DF1}), (\ref{state DF2}), (\ref{state DF3}), and (\ref{state DF4}); determines the amplitude and phase of potential oscillations by addressing Eqs. (\ref{both triag1})-(\ref{both triag2}), and examines the stability of these oscillation candidates using the methodology outlined in Subsection \ref{sec5.3}, where $H(s)$ is denoted by Eq. (\ref{incremental characteristic both}), while $N_{a,inc}$ characterized by Eqs. (\ref{control incremental DF1}), (\ref{control incremental DF2}), (\ref{control incremental DF3}) and (\ref{control incremental DF4}), and $N_{v,inc}$ characterized by Eqs. (\ref{state incremental DF1}), (\ref{state incremental DF2}), (\ref{state incremental DF3}), and (\ref{state incremental DF4}); and computes the magnitude and phase of the frequency response using Eqs. (\ref{both FR mag}) and (\ref{both FR phase}), respectively. The limits for acceleration and velocity adhere to the specifications from Subsection \ref{sec6.1}. We explore varying ratios of $R/{a}_{bound}$ and $R/\tilde{v}_{bound}$. A comparison between the frequency response's magnitude and phase as determined by both the IDF and the linear theoretical methods is drawn, using Simulink estimations for validation. This is visualized in Fig. \ref{fig:both FR}. The DF approximation proposed in \citep{zhou2023data} is excluded from this comparison, given its inability to address systems with dual nonlinearities.

It is observed from Fig. \ref{fig:both FR}(a) that when $R/a_{bound}$ and $R/\tilde{v}_{bound}$ are both small, the saturation limits are not reached and the system operates linearly. In Fig. \ref{fig:both FR}(b), only the $a_{bound}$ saturation limit is attained, leaving $\tilde{v}_{bound}$ unaffected, producing results consistent with Fig. \ref{fig:control FR}(b). Fig. \ref{fig:both FR}(c)-(d) represent conditions where both $a_{bound}$ and $\tilde{v}_{bound}$ are exceeded. Here, a noticeable decline in frequency response magnitude is observed, accompanied by a further decline in phase into the negative spectrum, comparing to Fig. \ref{fig:control FR}(c)-(d). The IDF method consistently demonstrates its proficiency in accurately representing the frequency response for systems encompassing both control and state saturation.

The findings align with the previous section. When compared with scenarios where only control or state saturation limits are reached, the underestimation of the oscillation attenuation effect and response time becomes even more pronounced.

\subsection{Sensitivity analysis\label{sec6.new}}

Having demonstrated the accuracy of the proposed method in the previous subsections, we now undertake a more comprehensive analysis of how the frequency and amplitude of the predecessor's oscillation affect the frequency response of an AV, considering the presence of saturation limits. 

To this end, we analyze the system using both the linear method and the proposed IDF method, considering both state and control input saturation, with the saturation limits set as specified in Subsections \ref{sec6.1} and \ref{sec6.3}, and using the default controller parameters. The frequency response is examined over a range from $0.1$ to $0,5Hz$ and an extended range of $R/a_{bound}$ from $0$ to $8$ (correspondingly, $0$ to $4$ for $R/v_{bound}$). The results are presented as heat maps in Fig. \ref{fig:heatmap stable}, where darker colors represent smaller magnitudes and more negative phases. From, Fig. \ref{fig:heatmap stable}(a)-(b), using the linear method, we observe that the magnitude and phase remain unchanged with varying $R/a_{bound}$ since the linear method does not account for saturation nonlinearity. However, as shown in Fig. \ref{fig:heatmap stable} (c)-(d), given the state and control input limits, the frequency response is in fact dependent on the oscillation amplitude of the predecessor. Generally, when the saturation limits are reached, the magnitude decreases and the phase becomes more negative compared to the linear method. This indicates that the linear method overestimates the amplification and underestimates the response time for an AV when these saturation limits are present. As $R/a_{bound}$ increases, the frequency at which the saturation limits are reached decreases. Furthermore, for the same value of $R/a_{bound}$, as the frequency increases, the magnitude gradually decreases after the limits are reached, while the phase undergoes a sudden shift when $R/a_{bound}$ exceeds approximately 5. These patterns are consistent with and generalize the results illustrated in Subsections \ref{sec6.1} and \ref{sec6.3}.

\begin{figure}[t]
    \centering
    \setlength{\abovecaptionskip}{0pt}
    \subcaptionbox{}{\includegraphics[width=0.35\textwidth]{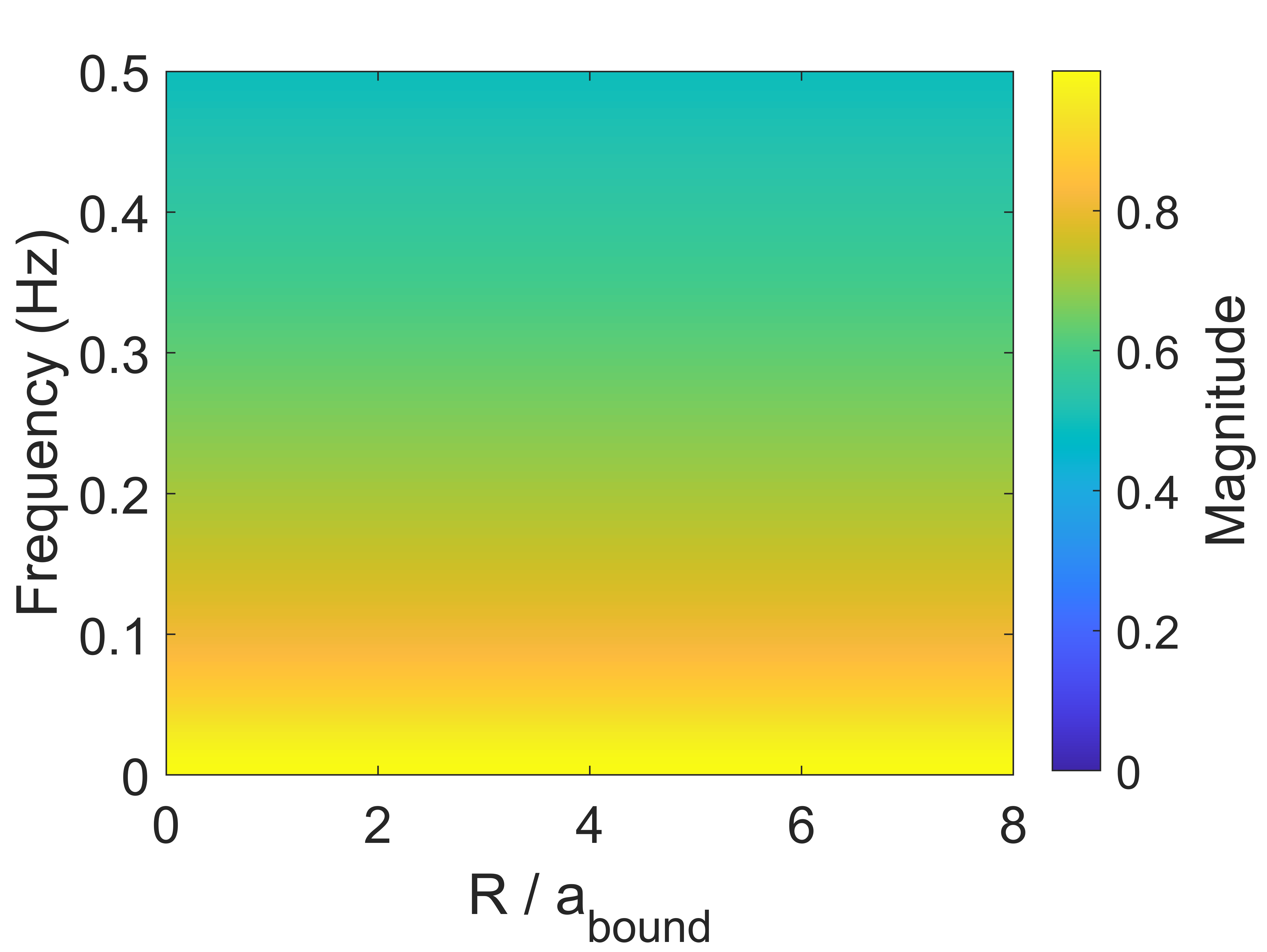}}
    \subcaptionbox{}{\includegraphics[width=0.35\textwidth]{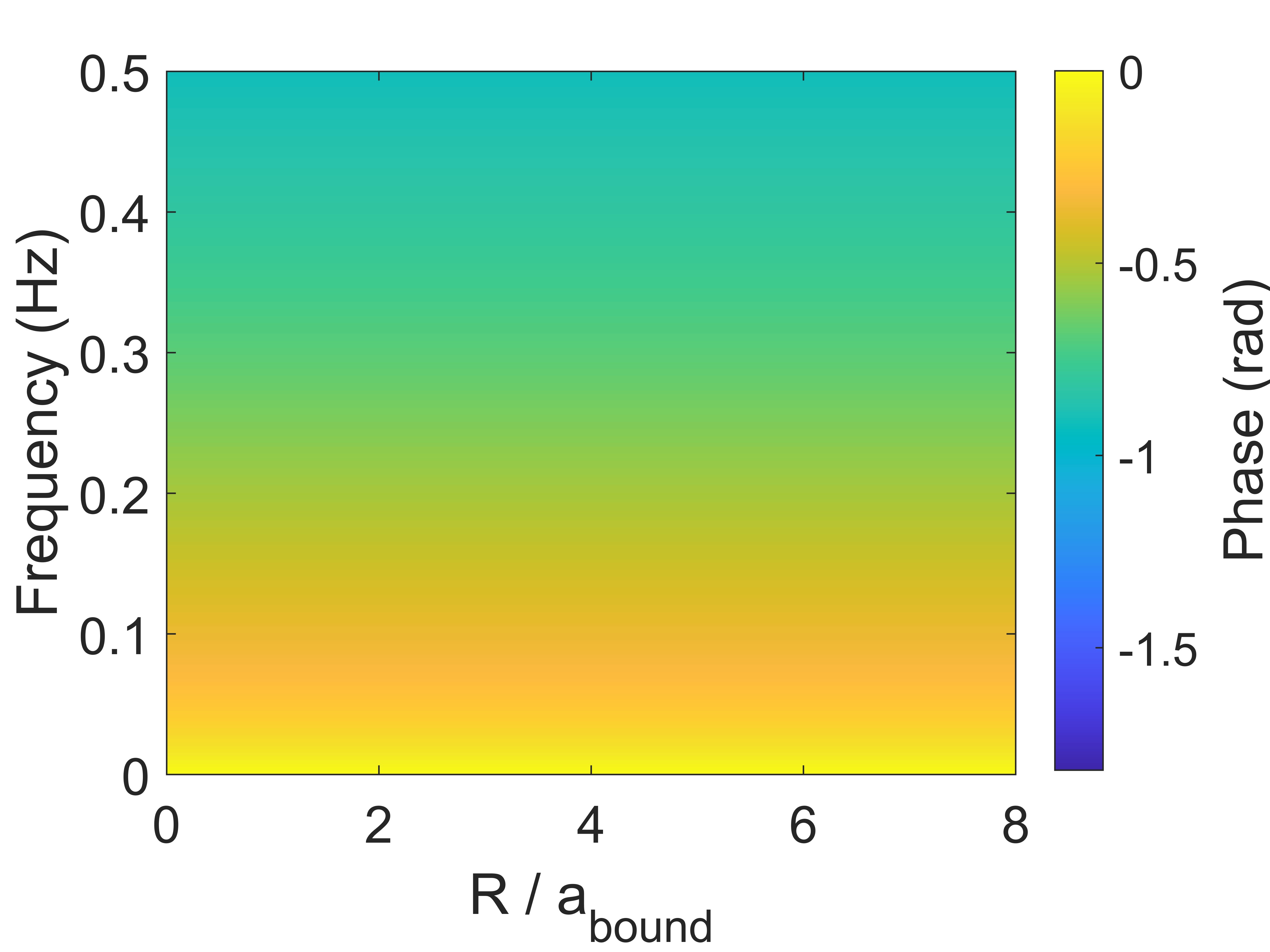}}
    \subcaptionbox{}{\includegraphics[width=0.35\textwidth]{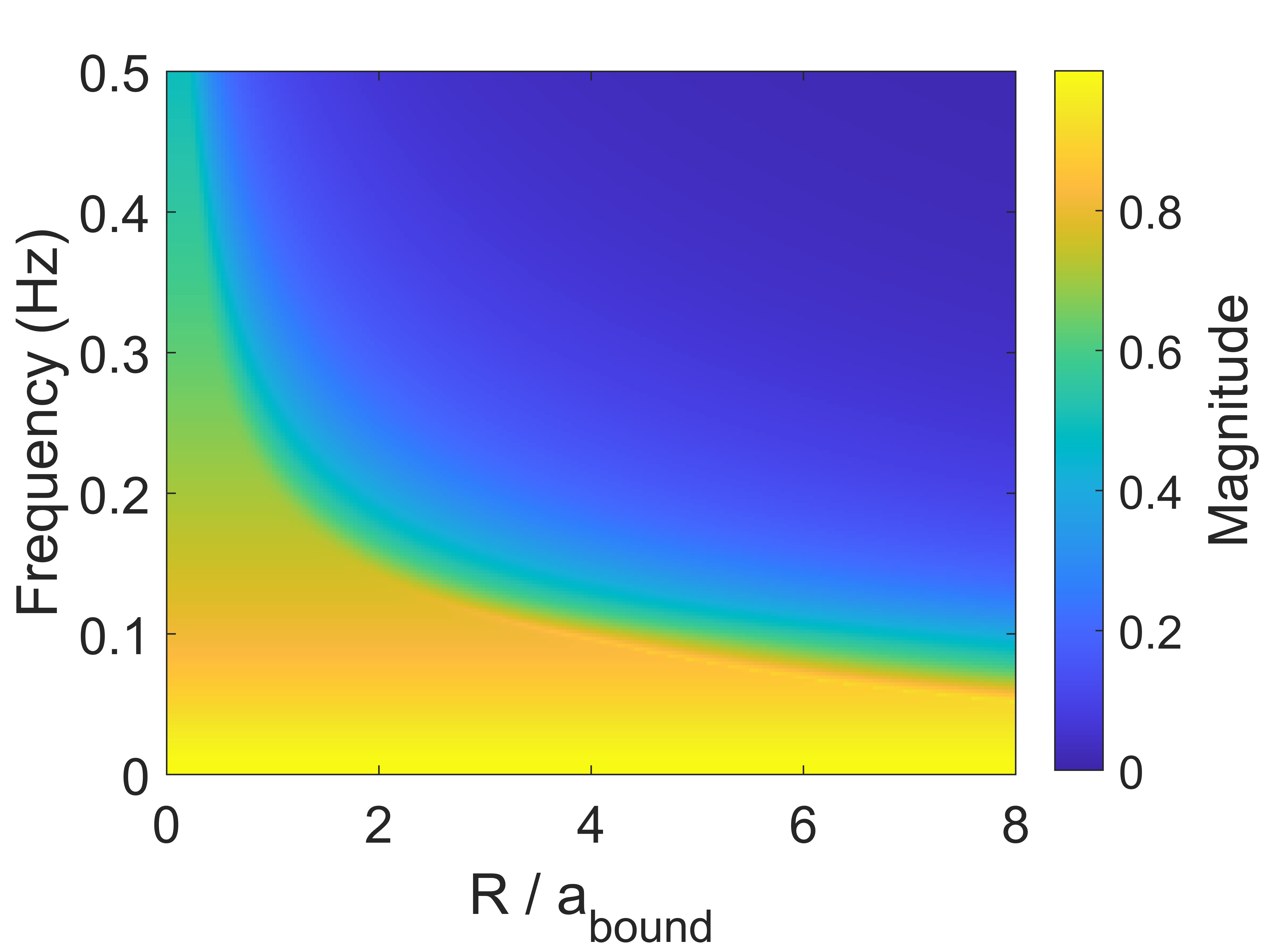}}
    \subcaptionbox{}{\includegraphics[width=0.35\textwidth]{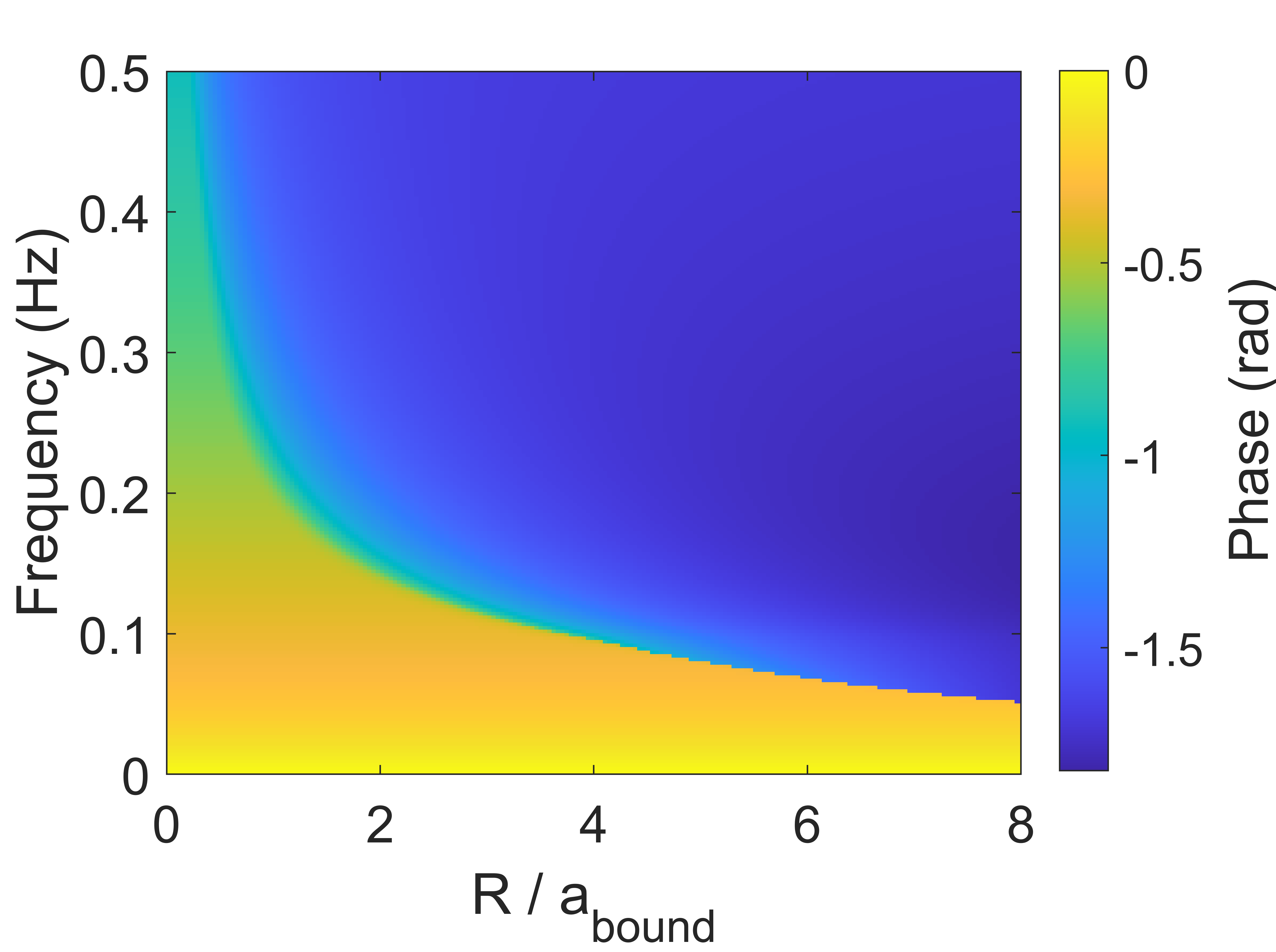}}
    \caption{Frequency responses of systems containing both control and state saturation with a linearly string-stable controller: (a) magnitude from linear analysis; (b) phase from linear analysis; (c) magnitude from IDF analysis; (d) phase from IDF analysis}
    \label{fig:heatmap stable}
\end{figure}

\begin{figure}[t]
    \centering
    \setlength{\abovecaptionskip}{0pt}
    \subcaptionbox{}{\includegraphics[width=0.35\textwidth]{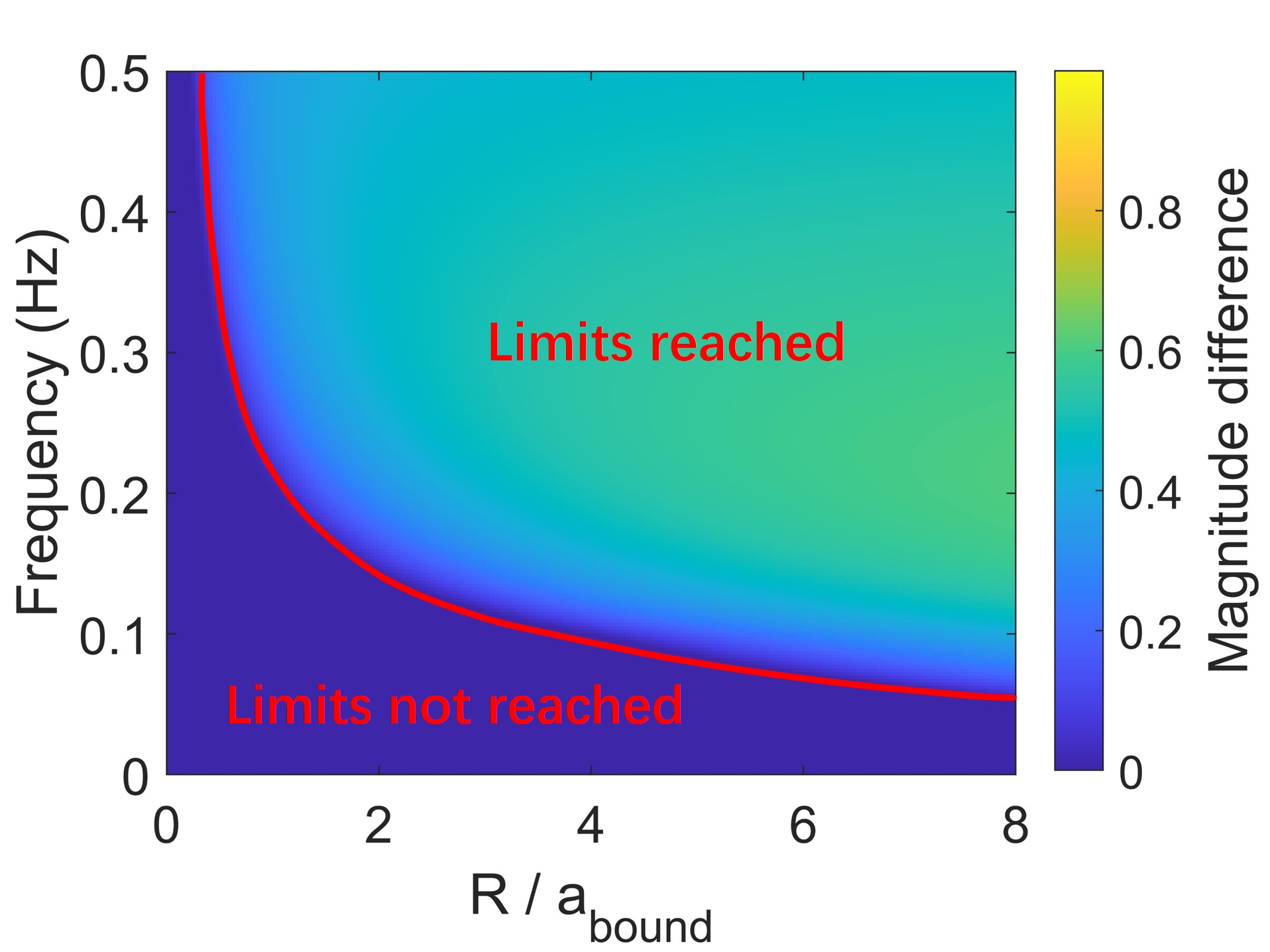}}
    \subcaptionbox{}{\includegraphics[width=0.35\textwidth]{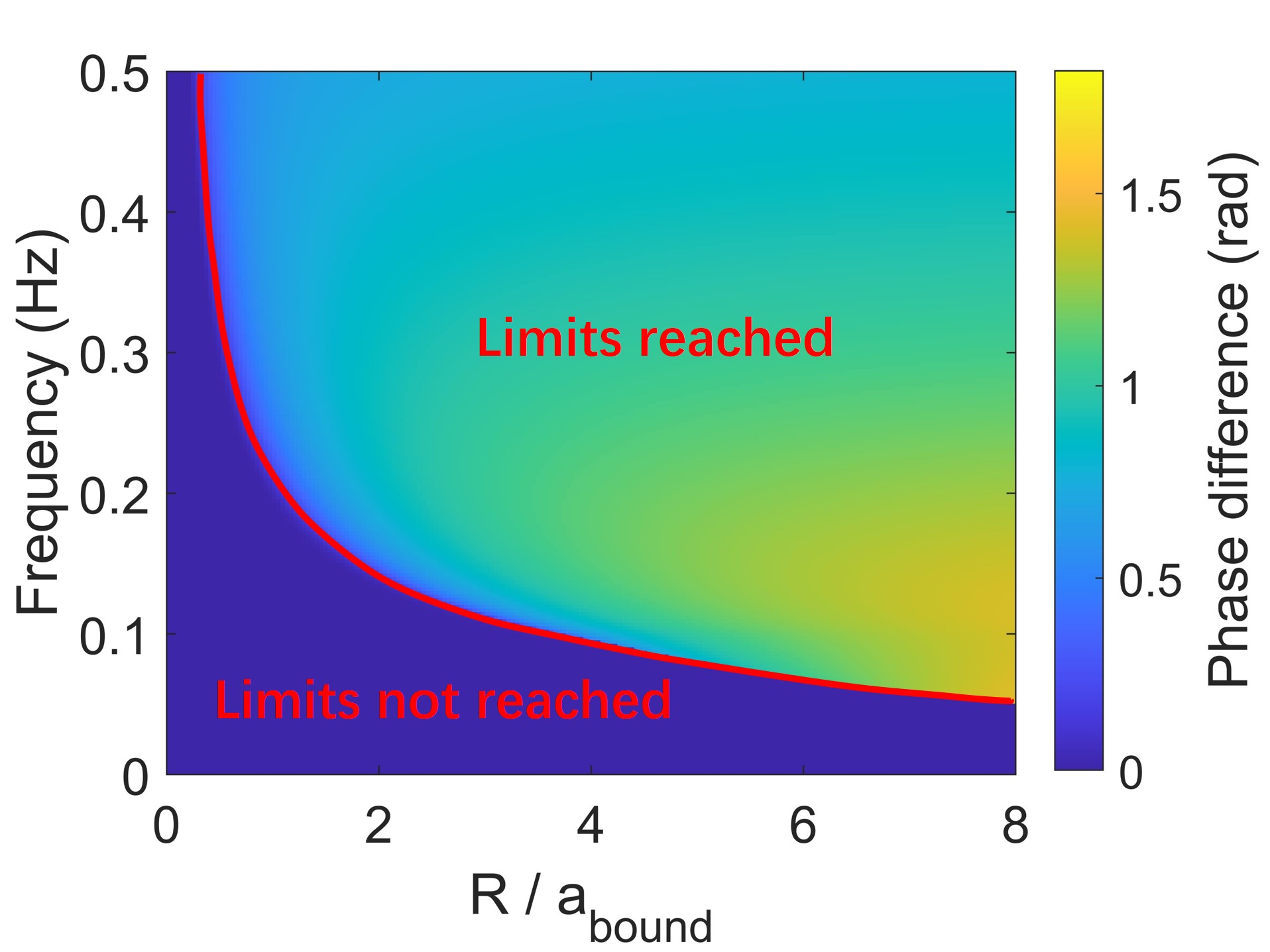}}
    \caption{Frequency response differences: (a) magnitude difference; (b) phase difference}
    \label{fig:heatmap stable dif}
\end{figure}

For clearer visualization of the discrepancies, the magnitude difference, $|F|_{lin}-|F|_{IDF}$, and phase difference, $\angle F_{lin}-\angle F_{IDF}$, are shown in Fig. \ref{fig:heatmap stable dif}. Here, the subscripts "$lin$" and "$IDF$" represent results obtained using the linear and $IDF$ methods, respectively. The regions where the limits are not reached and where they are reached are separated by a red line and appropriately labeled.

\subsection{Misjudged string stability of the linear theoretical method \label{sec6.4}}

In studies focusing on string stability of vehicular platoons, a widely accepted criterion indicating string stability is that the frequency response magnitude of all vehicles in the platoon should consistently stay at or below the threshold of 1 across the entire frequency spectrum (\citep{naus2010string,li2024sequencing}). However, as demonstrated in Subsections \ref{sec6.1}-\ref{sec6.new}, if the saturation limits are reached, the linear method tends to overestimate the amplification of AVs, which can lead to misleading results. This raises a natural question: might the linear method incorrectly classify a genuinely string-stable platoon as string-unstable? This potential error is significant given the widespread use of the linear method in current research. Highlighting such discrepancies, we turn our attention to a specific scenario. Since the dependence on the input amplitude of DF introduces additional complexity, we provide a simple indicative example of how the parameters are selected. Consider a linearly string-unstable, fully loaded automated truck with acceleration limit $\tilde{a}_{bound}=1ms^{-2}$ following a passenger car freely oscillating with an acceleration of amplitude $a_p=3ms^{-2}$ at frequency $f=0.02Hz$. The angular frequency is obtained as $\omega=2\pi f$, and the amplitude of the passenger car is $R=\omega^{-2}a_p=189.98m$. For a linearly string-unstable controller, we adjust the controller parameters to $\tau=0.4s$, $k_d=1s^{-2}$ and $k_v=0.4s^{-1}$, hence $k_1=1s^{-2}$, $k_2=0.4s^{-1}$ and $k_3=-0.8s^{-1}$. Further with the specified $\tilde{a}_{bound}$ and $R$, we analyze the response of the following automated truck within the whole frequency range of interest. A comparison of the magnitude and phase of the frequency response, derived from the IDF method, the linear theoretical method, the DF approximation proposed by \citep{zhou2023data}, and the Simulink estimation is showcased in Fig. \ref{fig:string misjudge}. As depicted by the yellow dots in Fig. \ref{fig:string misjudge}, the magnitudes obtained by the linear theoretical method exceed 1 at lower frequencies less than $0.2Hz$, indicating string-instability. However, the oscillatory speed limits are reached for the whole frequency range. As shown by the red pentagrams, the true magnitudes stay within the stable threshold for all frequencies. This observation is in line with the results presented by both the IDF method and DF approximation.

\begin{figure}[h]
    \centering
    \setlength{\abovecaptionskip}{0pt}
    \includegraphics[width=0.4\textwidth]{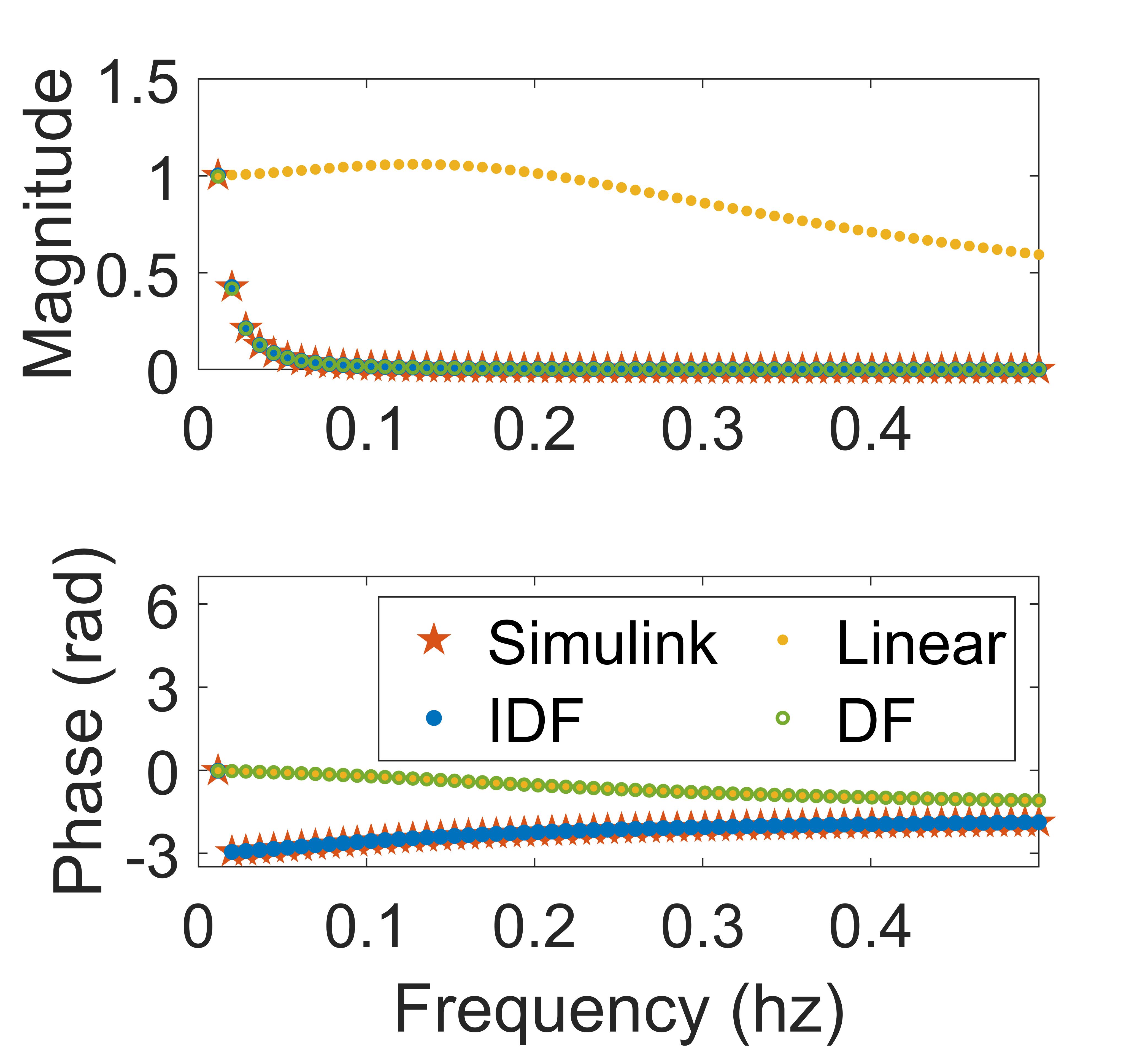}
    \caption{Exemplary comparison of frequency response indicating misjudged string stability of the linear theoretical method}
    \label{fig:string misjudge}
\end{figure}

\section{Conclusion\label{sec7}}
This paper presents an incremental-input DF method to analytically obtain the frequency response of AV systems with nonlinearity induced by traffic state and control input saturation. The main contribution of the proposed framework lies in analytically and accurately discerning the CF oscillation characteristics of AV closed-loop systems with nonlinearity induced by saturation. We derive the DF of control and state saturation in AV systems and replace the saturation elements with their DFs. Based on that, we further analyze the closed-loop frequency response of the AV systems and obtain oscillation candidates by first harmonics balancing. Subsequently, we derived the incremental-input DF of the control and state saturation for each candidate, to characterize nonlinear saturation's impact on perturbations. By further applying the Nyquist method, we evaluated the stability of each oscillation candidate. 

The proposed method is evaluated for AV systems integrated with solely the control saturation, solely the state saturation, and a combination of both. The theoretical frequency response results are consistent with the Simulink estimations and surpass other theoretical frameworks in terms of analyzing the oscillation attenuation and response time of AVs. It is found that when a saturation limit is reached, generally speaking, the widely used linear string stability analysis tends to underestimate both the attenuation effect and response time of AVs. When both traffic state and control input limits are reached, this phenomenon is even more pronounced. Our study also reveals that when saturation limits are reached for all frequencies, the commonly used linear method can misjudge string stability, while the proposed method remains accurate. {Future work includes investigating highly asymmetric limits using the dual-input describing function.}

\printcredits

\bibliographystyle{cas-model2-names}
\bibliography{main}


\end{document}